\documentclass[11pt]{article}
\def\comments{0}
\usepackage{base}
\usepackage{macros}
\newcommand*\samethanks[1][\value{footnote}]{\footnotemark[#1]}

\title{\LARGE{A Meta-Complexity Characterization of\\ Minimal Quantum Cryptography}}
\author{
    Bruno Cavalar\thanks{University of Oxford}, 
    Boyang Chen\thanks{Tsinghua University}, 
    Andrea Coladangelo\thanks{University of Washington}, 
    Matthew Gray\samethanks[1], 
    \\
    Zihan Hu\thanks{EPFL}, 
    Zhengfeng Ji\samethanks[2], 
    Xingjian Li\samethanks[2]
}

\begin{document}
\maketitle

\begin{abstract}
We give a meta-complexity characterization of $\EFI$ pairs, which are
considered the ``minimal'' primitive in quantum cryptography (and are equivalent to quantum commitments). More precisely, we show that
the existence of $\EFI$ pairs is equivalent to the following: there exists a
non-uniformly samplable distribution over pure states such that the problem
of estimating a certain Kolmogorov-like complexity measure is hard given a
\emph{single copy}. %

A key technical step in our proof, which may be of independent interest, is
to show that the existence of $\EFI$ pairs is equivalent to the existence of
\emph{non-uniform} single-copy secure pseudorandom state generators (nu
1-PRS). As a corollary, we get an alternative, arguably simpler, 
construction of a universal $\EFI$ pair.

\end{abstract}

\tableofcontents

\section{Introduction}

What is the minimum complexity assumption that implies cryptography? 
This
question has proved extremely fruitful in the classical world and has
sparked recent works in the new worlds of quantum complexity and
cryptography~\cite{DBLP:conf/eurocrypt/CavalarGGH25,DBLP:conf/crypto/HirokaM25, DBLP:conf/stoc/KhuranaT25}. In the classical world, the central cryptographic primitive of study is the one-way function ($\OWF$). One-way functions (functions that are efficiently computable, but hard to invert) are natural, their existence is implied by the security of almost all cryptographic primitives and schemes\footnote{With the exceptions of indistinguishability obfuscation and information theoretic primitives like secret sharing.}, and their existence is equivalent to the existence of the wide variety of cryptographic primitives and schemes found in the crypto-complexity class ``Minicrypt''~\cite{hastadPseudorandomGeneratorAny1999, impagliazzoPseudorandomGenerationOneway1989, DBLP:conf/coco/Impagliazzo95, goldreichNoteComputationalIndistinguishability1990}. Since an $\NP$ oracle can invert one-way functions, it is clear that $\P \neq \NP$ is implied by non-trivial cryptography. But it is much less clear whether it in turn implies anything non-trivial. Proving this converse statement (that $\P \neq \NP$ implies one-way functions) remains the ``holy grail''~\cite{DBLP:conf/crypto/LiuP23} of complexity-theoretic cryptography, and would result in the elimination of Heuristica and Pessiland, two of Impagliazzo's ``five worlds''~\cite{DBLP:conf/coco/Impagliazzo95}.

There have been three main approaches to this question. The first has been to find concrete problems that are believed to be hard, and to build one-way functions from them. These problems include factoring~\cite{DBLP:journals/cacm/RivestSA78}, discrete logarithm~\cite{DBLP:journals/tit/DiffieH76}, and learning with errors~\cite{DBLP:conf/stoc/Regev05}. This approach has enabled our current world of widespread practical cryptography, but is unlikely to ultimately answer our question. While the hardness of these problems does imply the existence of one-way functions, their hardness is not thought to be implied by the existence of OWFs, and they do not seem to capture anything fundamental about computation or complexity. Consequently they are very unlikely to prove to be minimal.

The second approach is due to Levin, who showed that there exists a
specific universal OWF $U$ that is one-way so long as any OWF exists at
all~\cite{Levin87}. This does, to some extent, provide an answer to our
question. The existence of this (or any other) universal one-way functions
is exactly equivalent to the existence of one-way functions at all, and so
``characterizes'' their existence. However these universal one-way
functions are quite unnatural, their hardness has not been of much
independent interest, and they are only ``weakly'' one-way.
Moreover they are difficult to study, and researchers have not been able to
directly connect them to more fundamental complexity questions.

Both of these approaches then seem unlikely to provide a satisfactory answer to our starting question, or bring us closer to a holy grail style result. However, over the last six years, there has been surprising progress coming from a new approach, meta-complexity.

\paragraph{Meta-complexity.}

Meta-complexity refers to the study of problems such as the Minimum Circuit
Size Problem (MCSP) and Kolmogorov complexity estimation, which are themselves concerned with the complexity of other
problems. Some of these problems (such as estimating time-unbounded
Kolmogorov complexity) are undecidable in the worst case, but most (such as
time-bounded Kolmogorov complexity and MCSP) are known to be in NP, but are
not known to be NP-complete. These problems are widely applicable in areas
including learning~\cite{DBLP:conf/focs/HiraharaN23}, hardness
magnification~\cite{DBLP:conf/innovations/ChenHOPRS20}, and sampling
complexity~\cite{DBLP:journals/eccc/Aaronson10a}, and have been studied in
their own right for decades~\cite{DBLP:journals/annals/Trakhtenbrot84}. One
of those applications, and one of the key motivations for their initial
introduction, is that they can be used to measure the amount of randomness
in specific objects. In particular this means that, in general, the
meta-complexity of random and pseudo-random objects diverges significantly.

Using the hardness of meta-complexity for cryptography is relatively new,
but has developed quickly since Santhanam's result showing that (under a
conjecture) we can base the existence of pseudorandomness on the hardness
of MCSP~\cite{DBLP:journals/eccc/Santhanam19}. Since then, a number of
papers have shown that a wide variety of meta-complexity problems can be
used to characterize one-way functions. These include the hardness of
time-bounded Kolmogorov complexity on the universal
distribution~\cite{DBLP:conf/focs/LiuP20}, the hardness of gap Kolmogorov
complexity on any samplable
distribution~\cite{DBLP:journals/eccc/IlangoRS21}, and ``breakdown of
symmetry of information'' for probabilistic time-bounded Kolmogorov
complexity~\cite{DBLP:conf/stoc/HiraharaILNO23}. 

Unlike the problems studied in the first approach, the hardness of these problems is equivalent to $\OWF$. Unlike the universal functions studied in the second approach, these problems are of significant independent interest, have closely related problems which have been shown to be NP-complete~\cite{DBLP:conf/focs/Ilango20}, and seem more amenable to the kind of worst-case results that are required to achieve the holy grail~\cite{DBLP:conf/crypto/LiuP25}. Since these results also tell us that these problems are easy in a world without one-way functions, they have re-contextualized Pessiland from the worst of all worlds into a ``wonderland for learning''~\cite{DBLP:conf/focs/HiraharaN23}.

However, in one important sense, none of these approaches have actually answered our question, which was ``what is the minimal complexity assumption that implies cryptography?''. The rational for studying $\OWF$ was that they have long been thought to be the minimal cryptographic assumption. However the recent explosion of work on quantum cryptographic primitives has weakened $\OWF$ position as minimal. This has raised the possibility of cryptography from weaker assumptions, including assumptions that do not even imply~$ \P \neq \NP$.

\paragraph{Quantum Cryptography}
Following the introduction of pseudorandom states by Ji, Liu, and Song~\cite{DBLP:conf/crypto/JiL018} a wide variety of works have investigated a new world of quantum cryptographic primitives. To date more than thirty such primitives have been introduced and investigated. Their study is significantly motivated by the results of Kretschmer, which showed that there are oracles relative to which these primitives exist even when $\BQP = \QMA$~\cite{DBLP:conf/tqc/Kretschmer21} or $\P = \NP$~\cite{kretschmerQuantumCryptographyAlgorithmica2022} (a world in which one-way functions do not exist). This world of primitives has been named ``Microcrypt'' --- a world of computational cryptography weaker than one-way functions.

As these primitives have been investigated, the picture that has emerged is
very different from the relatively clean world of classical cryptography.
The vast majority of classical cryptographic primitives are equivalent to
$\OWF$, and those that are not mostly fall into an orderly ladder-like
hierarchy of strictly more powerful primitives.
In contrast the current map of
quantum primitives appears much more complicated (and much less is known).

In a recent work~\cite{DBLP:journals/iacr/GoldinMMY24} Goldin et al. split these primitives into three sub-worlds defined by the classical oracle that can be used to break them. Their first world is ``QuantuMania'', and includes primitives such as quantum-computable post-quantum one-way functions, and efficiently verifiable one-way puzzles. The security of these primitives is defined in terms of a classical-input, classical-output problem in 
$\QCMA$ which is quantumly hard, meaning that these primitives can be broken by $\QCMA$ oracles.

Their second world is ``CountCrypt'' and includes primitives such as
pseudo-random states, one-way state generators, and
one-way puzzles. The security of these
primitives is typically defined in terms a quantum-input, classical output
which is hard given many copies of the quantum input. Since many copies of
the input are available, these problems all reduce, via shadow tomography,
to one-way puzzles, which can be broken with a $\PP$
oracle~\cite{DBLP:journals/quantum/CavalarGGHLP25}. Consequently, all
CountCrypt primitives can be broken by $\PP$ oracles. Recent
works~\cite{DBLP:conf/eurocrypt/CavalarGGH25, DBLP:conf/crypto/HirokaM25}
showed a meta-complexity characterization for one-way puzzles, 
and 
left open the question of whether
meta-complexity characterizations could be extended to cryptographic
primitives below one-way puzzles. 
Moreover,
separation between primitives in CountCrypt have been recently
shown~\cite{chen2025power,ananth2024cryptography,DBLP:conf/eurocrypt/BostanciCN25},
which illustrates the
depth of CountCrypt.

Our work focuses on their final world, ``NanoCrypt'' which contains
primitives whose security is defined in terms of hardness of a single-copy
quantum input, classical output problem. These primitives are not known to
be breakable by any classical oracle, and there exists an oracle under
which these primitives are secure against a single query to any classical
oracle~\cite{DBLP:conf/stoc/LombardiMW24}. Moreover, formal black-box
separations exist between NanoCrypt and CountCrypt~\cite{DBLP:conf/eurocrypt/BostanciCN25},
giving NanoCrypt the title of ``minimal'' world of quantum
cryptography.
Within NanoCrypt, $\EFI$ pairs, first proposed by Brakerski,
Canetti, and Qian~\cite{BCQ22}, are the primitive of interest if one wishes
to characterize minimal quantum cryptography, and thus answer the opening
question.

$\EFI$ pairs were first proposed by Brakerski, Canetti, and
Qian~\cite{BCQ22}. An $\EFI$ pair consists of a pair of mixed state families
which are efficient (E) to sample, are statistically far (F), and are
computationally indistinguishable (I). They have been shown to be implied
by almost all quantum cryptographic primitives, and to be equivalent to quantum
bit commitments and quantum multiparty computation, while being simple and
natural. In the handful of years since their introduction they have quickly
become the consensus choice for the minimal quantum cryptographic
primitive. Thus, in this work, we focus on the
following question:\\
\begin{center}
\emph{Can we find a meta-complexity problem whose hardness is equivalent to the existence of the minimal primitive in quantum cryptography, namely $\EFI$ pairs?}
\end{center}

\subsection{Our contributions}
Meta-complexity characterizations are made up of a pair of implications.
First, that the existence of a cryptographic primitive implies that some
meta-complexity problem is hard; second, that the meta-complexity
problem being hard implies the existence of that primitive.  
The first kind of implication is typically proven by
taking advantage of meta-complexity as a measure of randomness. In a
paradigmatic example of this, Ilango et
al.~\cite{DBLP:journals/eccc/IlangoRS21} achieve a meta-complexity
characterization of $\OWF$ in terms of a Kolmogorov complexity estimation
problem. The key is to show that a Kolmogorov complexity estimation
algorithm can distinguish the outputs of a $\PRG$ from uniformly random
strings~\cite{DBLP:journals/eccc/IlangoRS21}. 
Since $\PRG$s are equivalent to
$\OWF$ by the well-known work of H{\aa}stad et
al.~\cite{hastadPseudorandomGeneratorAny1999},
breaking a $\PRG$ is sufficient to break $\OWF$s. 
In the previous quantum characterizations of
meta-complexity~\cite{DBLP:conf/eurocrypt/CavalarGGH25,DBLP:conf/crypto/HirokaM25},
the latter step is achieved by showing that a quantum Kolmogorov complexity
estimation algorithm can distinguish between the outputs of a
pseudo-entropy generator (which can be built from one-way puzzles), and the
fully entropic distribution it should be indistinguishable~from.

So if we are following this roadmap, our first step should be to build some
kind of pseudo-entropy generator from $\EFI$ pairs (which we will refer to as
$\EFI$ from here on, for short). Before our work, it was
known that single-copy secure pseudorandom state generators ($\OnePRS$) imply
$\EFI$, and that they are a NanoCrypt primitive. But showing that $\EFI$ imply
any kind of pseudo-entropy has been an open problem. We spend the first
section of this paper showing our first main result: while $\EFI$ may not
imply $\OnePRS$, they do imply a (only slightly) non-uniform version.

\begin{theorem}\label{thm:intro-EFI-1PRS}
$\EFI$ exist if and only if non-uniform $\OnePRS$ with advice size
  $O(\log \lambda)$ exist.
\end{theorem}

Our first result immediately implies that, for any notion of ``complexity'' of
states that distinguishes between ``pseudorandom states'' and Haar-random
states, estimating this notion given a single copy of a state
must be hard on non-uniform state families if EFIs exist.
Between 2000 and 2004, four quantum generalizations of Kolmogorov
complexity~\cite{DBLP:conf/coco/Vitanyi00, Gac01, mora2004algorithmic,
DBLP:conf/coco/BerthiaumeDL00} were introduced, each of which extrapolated
one of the equivalent ways of thinking about Kolmogorov complexity. All
four notions are such that the Kolmogorov complexity is much higher for
random states than for pseudorandom states. So, assuming $\EFI$ exist, all of
these notions must be hard to estimate (on some state family) given just a single copy. 

However, to prove a meta-complexity characterization, we need to prove the
converse direction: if $\EFI$ do not exist, then we can efficiently
estimate one of these complexity notions given just a single copy. None of
the above notions turn out to have all of the necessary properties for a
proof of this direction. However, G\'{a}cs' notion (which he called
``quantum algorithmic entropy'', and of the four notions above most closely
resembles a measure of entropy) has several of them, and we show that a
smoothed version of his notion, which we will refer to as $\mathsf{H}$-complexity in this introduction, turns out to have all of them. Overall,
letting $\GapH$ denote the problem of estimating this notion of quantum
Kolmogorov complexity (given some promise gap), we obtain the following
meta-complexity characterization of $\EFI$.

\begin{theorem}[Informal]
\label{thm:informal-main}$\EFI$ exist if and only if there exists a
    non-uniform family of efficiently samplable states $\{\ket{\psi_k}\}$ such that $\GapH$ is hard on average. 
\end{theorem}
In conjunction, we provide a second characterization of $\EFI$ pairs in terms of the hardness of estimating a different, but related, notion of Kolmogorov complexity. We refer to the associated problem as $\GapU$. The formal definitions of these notions and the corresponding estimation problems can be found in Sections~\ref{sec:kolmogorov}.

From the characterizations above, we get several corollaries. First, because of the structure of our
proof, one of the distributions that serves as half of our $\EFI$ pair is also a
$\OnePRS$. This means that, if our initial hard-on-average $\GapH$ instance was uniformly sampled,
this results in a uniform $\OnePRS$. Conversely, since $\OnePRS$ are, almost by definition, hard-on-average $\GapH$ instances, this gives us a meta-complexity characterization of
$\OnePRS$ as well.
\begin{corollary}
    $\OnePRS$ exist if and only if there exists a uniform family of efficiently
    samplable states $\{\ket{\psi_{k}}\}$ such that, for some uniformly
    computable gap, the problem $\GapH$ is hard on average.
\end{corollary}
Second, one direction of Theorem~\ref{thm:informal-main} does not rely on the fact that the state family $\{\ket{\psi_{k}}\}$ is efficiently samplable. In~\cref{sec:unkeyed}, we obtain characterizations of $\EFI$ from the hardness of $\GapH$ over single-copy samplable state families (i.e., state families for which one can only efficiently sample a single copy of a state).

Third, as a corollary of Theorem~\ref{thm:intro-EFI-1PRS}, we get a concrete construction of a ``universal'' $\EFI$ pair. That is, a concrete construction that is an $\EFI$ pair if and only if an $\EFI$ pair exists at all. This construction is arguably simpler than
the universal construction from~\cite{hiroka2024robust}, which requires the use of a ``combiner''.

As a final contribution, we provide a potentially unifying viewpoint on our meta-complexity characterizations of $\EFI$ pairs. A bit more precisely, let $\Pi_r$ denote the span of states with $\Knet$-complexity of size at most $r$ (where informally, $\Knet$-complexity is a notion introduced by Mora and Briegel~\cite{mora2004algorithmic} that captures the minimum description-length of a program that can output the state). We refer to the latter as the ``span of easy states''. Then, we show the following characterization.
\begin{theorem}[Informal]
    $\EFI$ exist if and only if there exists an efficiently computable $r$ and a non-uniform family of efficiently samplable states $\{\ket{\psi_{k}}\}$ such that it is hard on average to decide whether a state from the family is in $\Pi_{r}$, or has small overlap with $\Pi_{r+\omega(\log n)}$, given a single copy.
\end{theorem}
The notion of ``span of easy states'' is unifying in the following sense. There exists an appropriate ``robust'' version of the latter such that the $\mathsf{H}$-complexity of a state is tightly related to its overlap onto this robust span. Roughly speaking, a state has low $\mathsf{H}$-complexity if and only if it has a large overlap onto the ``robust span of easy states''. We refer the reader to~\cref{sec:24} in the technical overview and~\cref{sec:relating-entropy-span} in the main text for more~details.

\subsection{Connections}

\paragraph{$\OWP \Rightarrow \EFI$.}

Combining this result with the meta-complexity characterizations of $\OWP$ gives an alternate proof that $\OWP$ imply $\EFI$. 
All the notions of meta-complexity considered here 
reduce to (resource-unbounded) Kolmogorov complexity when restricted to classical
strings. 
In~\cite{DBLP:conf/eurocrypt/CavalarGGH25,DBLP:conf/crypto/HirokaM25} they show that $\OWP$ exists if and only if there exists a
distribution
$D$ samplable in quantum polynomial-time
such that $\GapK$ is weakly hard. To show this they first show that that $\OWP$ exists if and only if there exists a 
distribution
$D$ samplable in non-uniform quantum polynomial-time
such that $\GapK$ is hard~\cite[Thm 5.1]{DBLP:conf/eurocrypt/CavalarGGH25}, then they combine over the choice of advice. By interpreting this non-uniform distribution as a mixed
state, we obtain that all the notions discussed in this paper must also be non-uniformly 
hard on average over that mixed state, which in turn implies that $\EFI$
exists.

\paragraph{Unitary Synthesis and Quantum Complexity.} A recent work shows
that there exists an oracle relative to which there exists
an $\EFI$ secure against a
single query to any classical oracle~\cite{DBLP:conf/stoc/LombardiMW24}. 
This has been interpreted (for
instance by~\cite{DBLP:conf/eurocrypt/CavalarGGH25}) as a weak but general
barrier to showing any complexity-theoretic consequences of the existence
of $\EFI$. This work in no way breaks the arguments of Lombardi et al. but
it does show that the existence of $\EFI$ implies that these specific
complexity theoretic problems must be hard on some distribution.

This further motivates the study of a complexity theory focused on
inherently quantum tasks. This inherently or ``properly'' quantum
complexity theory should help us better understand the relationships
between tasks with either quantum inputs, outputs, or both. 

\subsection{Open problems} This work leaves open several problems.

\begin{enumerate}
    \item \textbf{Uniform $\OnePRS$.} We are able to show that $\EFI$ are equivalent to non-uniform $\OnePRS$. 
        However, the amount of non-uniformity is quite small and corresponds
        to the amount of entropy of some efficiently preparable mixed
        state. So, in addition to being quite small, this advice is also
        (inefficiently) computable,  contrasting with the more fundamentally
        non-uniform advice used for instance to show that unary halting
        is in $\mathsf{P/poly}$. Proving that $\EFI$ implies $\OnePRS$
        would significantly clean up NanoCrypt, resulting in a broad
        equivalence that more closely resembles Minicrypt. The most obvious
        route to showing this would be to create a robust combiner for
        single-copy PRS; however,  this seems challenging,  and it is unclear
        which other approaches might work.
    \item\textbf{Other characterizations of $\EFI$.} $\OWF$s have been shown to be
        equivalent to the hardness of a wide variety of distinct
        meta-complexity problems. It seems very possible that $\EFI$ may have
        more equivalences than just those shown in this paper. Unlike
        $\OWP$, for which Kretschmer's oracle separations provide a barrier
        to showing an equivalence with time-bounded Kolmogorov complexity
        or $\MCSP$, there does not seem to be any such barriers for $\EFI$.
    \item\textbf{Applications of the non-existence of $\EFI$.} Using
        meta-complexity, Hirahara and Nanashima have shown that the
        non-existence of $\OWF$s has powerful applications for
        learning~\cite{DBLP:conf/focs/HiraharaN23}. Are there similar
        applications for the non-existence of $\EFI$ for quantum circuit
        learning, state synthesis, or other problems?
    \item\textbf{A uniform characterization.} The non-uniformity which
        remains in our final theorem statement is somewhat unsatisfying and
        does not feel fundamental. Even if an equivalence with $\OnePRS$ stays
        out of reach,  it is possible that we could prove that the existence
        of $\EFI$ is equivalent to some form of hardness of some
        meta-complexity notion on a uniformly samplable distribution.
    \item\textbf{Characterizing one-way state generators.} While we focus on $\EFI$ pairs in this work, as they capture ``minimal'' quantum cryptography, there are other quantum cryptographic primitives that still do not have a meta-complexity characterization, for example one-way
        state generators. Previous work has mentioned this as an open
        problem~\cite{DBLP:conf/eurocrypt/CavalarGGH25}. 
        Either
        giving a characterization or providing a formal barrier to doing so
        would be valuable. In~\cite{batra2024commitments}, it is shown that one-way state generators with inefficient verifiers are equivalent to $\EFI$, so the meta-complexity characterization that we prove in this paper is also a characterization of one-way state generators with inefficient verifiers. But one-way state generators with efficient verifiers are plausibly stronger primitives~\cite{DBLP:conf/eurocrypt/BostanciCN25, behera2025new}.
\end{enumerate}

\section{Technical overview}

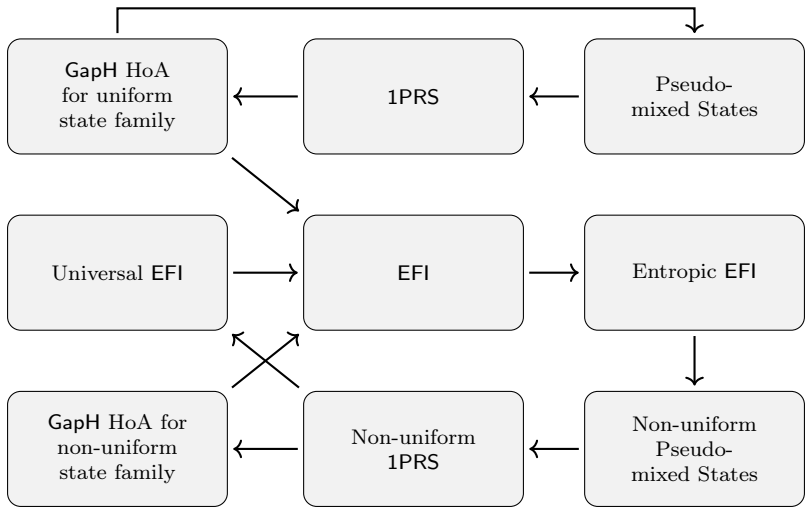
\begin{figure}[ht]
  \centering
  \begin{tikzpicture}[node distance=.8cm,
    arr/.style={->, thick, shorten >= 2pt, shorten <= 2pt},
    arrs/.style={<->, thick, shorten >= 2pt, shorten <= 2pt},
    box/.style={draw, rounded corners=5, fill=gray!10, align=center, font=\scriptsize,
      minimum height=4em, text width=2.5cm, inner sep=6}]

    \node[box] (uEFI) {Universal $\EFI$};
    \node[box, right=1cm of uEFI] (EFI) {$\EFI$};
    \node[box, below=of uEFI] (GapH) {$\GapH$ HoA for non-uniform state family};
    \node[box, right=1cm of GapH] (1PRS) {Non-uniform\\ $\OnePRS$};
    \node[box, right=of 1PRS] (PMS) {Non-uniform\\ Pseudo-mixed States};
    \node[box, right=of EFI] (EEFI) {Entropic $\EFI$};
    \node[box, above=of uEFI] (uGapH) {$\GapH$ HoA for uniform state family};
    \node[box, right=1cm of uGapH] (u1PRS) {$\OnePRS$};
    \node[box, right=of u1PRS] (uPMS) {Pseudo-mixed States};

    \draw[arr] (EFI) -- (EEFI);
    \draw[arr] (EEFI) -- (PMS);
    \draw[arr] (PMS) -- (1PRS);
    \draw[arr] (1PRS) -- (GapH);
    \draw[arr] (GapH.north east) -- (EFI.south west);
    \draw[arr] (uGapH.south east) -- (EFI.north west);

    \draw[arr] (uEFI) -- (EFI);
    \draw[arr] (1PRS.north west) -- (uEFI.south east);
    \draw[arr] (uPMS) -- (u1PRS);
    \draw[arr] (u1PRS) -- (uGapH);

    \draw[arr] (uGapH.north) -- ([yshift=.4cm]uGapH.north) -- ([yshift=.4cm]uPMS.north) -- (uPMS.north);
  \end{tikzpicture}
  \caption{Outline of reductions. ``HoA'' stands for ``hard on average''. The
  six problems at the bottom are all equivalent, as the reductions form
  cycles. The top three are also equivalent and imply the bottom six. }
  \label{fig:reductions}
\end{figure}
In this section, we give a high-level overview of how we prove our results. The starting point for our results is to establish a novel equivalence between $\EFI$ pairs and (non-uniform) single-copy pseudorandom states ($\OnePRS$). We then leverage this equivalence to obtain concrete characterizations of $\EFI$ pairs. %
Recall that a $\OnePRS$ is a family of efficiently generatable quantum states
$\{\ket{\psi_k}\}$ that is  ``stretching'' (i.e., the key size is smaller
than the number of qubits), and computationally indistinguishable, but
statistically far, from a uniformly (Haar) random state given a
\emph{single copy} of the state.\footnote{The expert reader will already
know that the ``multi-copy'' variant of a $\PRS$ is qualitatively stronger
than $\EFI$ pairs (formally, a black-box separation is known), so there is no
hope of proving an equivalence.} 

\subsection{Equivalence of \texorpdfstring{$\EFI$}{EFI} pairs and (non-uniform) \texorpdfstring{$\OnePRS$}{1PRS}} One direction of this equivalence is easy. It is well-known that a $\OnePRS$ implies an $\EFI$ pair, where $\rho_0 = \mathbb{E}_k \ketbra{\psi_k}{\psi_k}$ (where here $\{\ket{\psi_k}\}$ is the family of pseudorandom states) and $\rho_1 = \frac{I}{2^n}$, where $n$ is the number of qubits. 

The converse is the crux. Consider an arbitrary $\EFI$ pair $(\rho_0, \rho_1)$. At a high level, there are two obstacles to turning this into a $\OnePRS$: the first is that $\rho_1$ may not equal $\frac{I}{2^n}$ for $\EFI$ pairs. The second is that instead of a mixed state $\rho_0$, we require a family of \emph{pure} states $\{\ket{\psi_k}\}$ such that $\mathbb{E}_k \ketbra{\psi_k}{\psi_k}$ is computationally indistinguishable from $\frac{I}{2^n}$. So the question becomes: can an arbitrary $\EFI$ pair (where $\rho_1$ may not be maximally mixed) always be ``massaged'' into one where $\rho_1$ is maximally mixed? 

Our approach takes inspiration from Goldreich's construction of a $\PRG$ from an $\EFID$ pair (the classical analogue of an $\EFI$ pair)~\cite{goldreichNoteComputationalIndistinguishability1990}, but requires some uniquely quantum insights\footnote{As pointed out also by Brakerksi, Canetti, and Qian~\cite{BCQ22}, Goldreich's proof relies crucially on the fact that for a $\BPP$ algorithm it is possible to separate the randomness from the rest of the computation. Such techniques cannot work for quantum algorithms, as observed also in the work of Aaronson, Ingram, and Kretschmer comparing $\BPP$ and $\BQP$~\cite{aaronsonAcrobaticsBQP2021}.}. The high-level outline is the following:
\begin{itemize}
\item First, show that, starting from $(\rho_0, \rho_1)$, one can construct a new $\EFI$ pair $(\sigma_0, \sigma_1)$ where $\sigma_0$ and $\sigma_1$ have \emph{noticeably different} von Neumann entropies. The key insight is that, if $\rho_0$ and $\rho_1$ are statistically far (and hence approximately orthogonal), then the states $\sigma_0 = \frac12 \ket{0}\bra{0}\otimes \rho_0 + \frac12 \ket{1}\bra{1}\otimes \rho_1 $ and $\sigma_1 =  \frac{I}{2} \otimes (\frac12 \rho_0 + \frac12 \rho_1)$ have von Neumann entropies that differ approximately by one bit! It is also not too difficult to show that $\sigma_0$ and $\sigma_1$ remain computationally indistinguishable. We refer to $(\sigma_0, \sigma_1)$ as a ``pseudo-entropy''~pair, or more formally as an entropic $\EFI$ pair.
\item Second, ``upgrade'' the pseudo-entropy pair to one where the second state is maximally mixed. This is possible via a combination of parallel repetition (to amplify the entropy gap), and the use of a ``strong quantum randomness extractor'', such as the one from~\cite{dupuisDecouplingApproachQuantum2010}.
\item Finally, given an $\EFI$ pair $(\sigma_0, \sigma_1)$, where $\sigma_0$ has noticeably less than full entropy, and $\sigma_1 = \frac{I}{2^n}$, the idea to obtain a $\OnePRS$ is the following. Since $\sigma_0$ is efficiently generatable, there exists an efficiently generatable purification $\ket{\psi_0}_{\sf AB}$ such that $\tr_{\sf A}[\ket{\psi_0}] = \sigma_0$. One should then choose an appropriate family of ``twirling'' unitaries $U_k$ on $\sf A$ with the property that $\mathbb{E}_k (U_k \otimes I) \ket{\psi_0}\bra{\psi_0} (U_k^{\dagger} \otimes I) \approx \frac{I}{|\sf A|} \otimes \sigma_0$. Why is this useful? This is because, by hypothesis, the latter state is computationally indistinguishable from $\frac{I}{|\sf A|} \otimes \frac{I}{|\sf B|}$. Thus, defining the $\OnePRS$ to be $\ket{\psi_k} = U_k \ket{\psi_0}$ would yield the desired guarantee. Now, such a family of efficiently implementable twirling unitaries always exists, but, to obtain a $\OnePRS$, we crucially need the length of the seed $k$ to be smaller than the number of qubits of the output state. This is possible to achieve thanks to the fact that the von Neumann entropy of $\sigma_0$ (and hence the entanglement entropy of $\ket{\psi_0}$) is less than full (one can again leverage a strong quantum randomness extractor~\cite{dupuisDecouplingApproachQuantum2010}).
\end{itemize}

Note one important point here: in order to pick the appropriate family $\{U_k\}$ (i.e., the extractor), one needs to know the von Neumann entropy of $\sigma_0$. It is not yet clear how this can be circumvented, and this is why this approach only shows that an $\EFI$ pair implies a \emph{non-uniform} $\OnePRS$, where the advice is of logarithmic size (since the entropy is at most $n$).

\subsection{From a (non-uniform) \texorpdfstring{$\OnePRS$}{1PRS} to a ``universal'' \texorpdfstring{$\EFI$}{EFI} pair.} 
Leveraging the above equivalence, we can obtain a concrete construction of
a ``universal'' $\EFI$ pair, reminiscent of Levin's universal $\OWF$
construction. This is a concrete $\EFI$ pair construction that is secure if
and only if an $\EFI$ pair exists at all. We note that our construction will
differ from the universal $\EFI$ construction of~\cite{hiroka2024robust} in
that we do not need combiners. 

Informally, the outline is the following. Assume a non-uniform $\OnePRS$ (with
logarithmic-length advice) exists. Denote the seed length by $\lambda$, and the length of the output state by $n(\lambda)$. Then, if the advice is logarithmic size, for each 1PRS state $\ket{\psi_k}$, there exists a Turing machine with description length $r = \lambda + O(\log{\lambda}) + C$, where $C$ is a universal constant, that there exists a Turing machine of size $C$, takes in $1^\lambda$ and $k$ as input, in some polynomial time $T(\lambda)$, outputs a quantum circuit whose output (when run on the all zero state) is $\ket{\psi_k}$. Then, we argue that the following concrete universal construction must be an $\EFI$ pair. Consider the pair of state families $\left(\{\rho_{r, T}\}, \{\frac{I}{2^n}\}\right)$, implicitly indexed by $\lambda$, where $\rho_{r, T} = \frac{1}{2^r}\sum_{|P|\leq r}
    \ketbra{\psi_P^T}{\psi_P^T}$,
    and $\ket{\psi_P^T}$ is the $n$-qubit state obtained by running Turing machine $P$ for time $T$ to get a quantum circuit $Q$ with an $n$-qubit state output, and then running $Q$ on the all-zero state. The key point is that the states $\rho_{r,T}$ and $\frac{1}{2^n}I$ are statistically far because the rank of $\rho_{r,T}$ is far from being full (since $n = r + \omega(\log \lambda)$). At the same time, any computational distinguisher has at least an inverse-polynomial failure probability due to the fact that an inverse-polynomial fraction of the programs $P$ outputs a state from the $\OnePRS$~family. This kind of ``weak'' $\EFI$ pair can be upgraded to a standard $\EFI$ pair, based on known results.

\subsection{A meta-complexity characterization of \texorpdfstring{$\EFI$}{EFI} pairs}    
While the construction above is a desirable step forward, as it gives a
valid and concrete characterization of minimal quantum cryptography, it is
still not as ``natural'' as one could hope for. Finding characterizations
that are of independent interest, with natural connections to other
well-studied problems in complexity theory is an important step towards
relating (or decoupling) the existence of quantum cryptography to (or from)
other complexity-theoretic statements. Here, we describe how we obtain
a characterization of $\EFI$ pairs in terms of the hardness of a natural
meta-complexity problem.

Recall that the Kolmogorov complexity of a string $x$ captures the description length of the \emph{shortest} program (to be run on some fixed universal Turing machine $U$) that outputs $x$. In the quantum setting, different variants of Kolmogorov complexity of \emph{states} have been proposed (differing mainly in how they account for the program outputting an approximation of the state)~\cite{Gac01, mora2006quantum}. Let $\sf \Hbar$ denote a suitable such notion (which will describe in more detail below).
Informally, here is the problem we propose to consider: let $\{\ket{\psi_k}\}$ be an efficiently sampleable family of states with the promise that each state either has $\sf \Hbar$-complexity above some threshold $r(n)$ or below $r(n)-\omega(\log n)$; given a \emph{single copy} of a state from the family, decide if it has low or high $\sf \Hbar$-complexity. Due to the fact that the equivalence we can prove is between $\EFI$ and \emph{non-uniform} $\OnePRS$, our approach eventually leads to a version of the latter problem with respect to a \emph{non-uniform} family of states. A bit informally, the notion of $\sf \Hbar$-complexity that we are looking for should satisfy the following properties:
\begin{itemize}[leftmargin=1.1em]
\setlength\itemsep{0em}
\item It should be such that pseudorandom states have low complexity, while
    a uniformly random string has high complexity with overwhelming
    probability. If this is the case, then we can build on the previous
    equivalence: assuming an $\EFI$ pair exists, a (non-uniform) $\OnePRS$ also
    exists, and thus we can consider the following family of states. Sample
    a state from the $\OnePRS$ family with probability $\frac12$, and sample a
    uniformly random standard basis state with probability $\frac12$. Thus,
    assuming an $\EFI$ pair exists, this would be a (non-uniform) family of
    states for which the problem of estimating $\sf \Hbar$-complexity is
    hard.
\item It should be such that low complexity states have low von Neumann entropy while high complexity states have high von Neumann (or min-)entropy. This would allow us to argue the converse direction: if there exists a family $\{\ket{\psi_k}\}$ on which the problem of estimating $\sf \Hbar$-complexity is hard, then an $\EFI$ pair exists. The $\EFI$ pair would be $(\rho, \frac{I}{2^n})$, where $\rho$ is the state obtained by applying a suitable strong quantum randomness extractor~\cite{dupuisDecouplingApproachQuantum2010} to the state $\mathbb{E}_k \ketbra{\psi_k}{\psi_k}$. The point is that $\rho$ is efficiently preparable, and, if the appropriate parameters are chosen for the extractor, it is far from the maximally mixed state since half of the states in $\{\ket{\psi_k}\}$ have low von Neumann entropy.
\end{itemize}
We show that a notion of $\sf \Hbar$-complexity that satisfies the two properties above is a ``smooth'' version of (a variant of) G\'{a}cs' complexity~\cite{Gac01} (defined precisely in Definition~\ref{def:hbar}, and \ref{def:hbar-eps} for its smoothed version). We denote the smooth version by $\Hbar^{\eps}$ for $\eps \in [0,1]$. Ultimately, a bit more precisely, the promise problem that we arrive at is: let $\{\ket{\psi_k}\}$ be a non-uniform family of states with the promise that, for almost all states, either $\Hbar^{1-\eps} < r(n)-\omega(\log n)$ or $\Hbar^{0}> r(n)$, given a \emph{single copy} of a state from the family, decide which is the case. The hardness of this problem for some family of states is equivalent to the existence of $\EFI$ pairs. Just like for our $\OnePRS$ construction from an $\EFI$ pair, the length of the advice needed to efficiently sample the family of states is logarithmic.

We note that our proof of the equivalence above does not actually rely on the fact that $\{\ket{\psi_k}\}$ is a keyed family. In fact, it also works for any single-copy samplable state family, i.e., a family such that there exists a QPT algorithm that can sample \emph{a single copy} of a state from the family.  %
Thus as a result, we can also show an equivalence between the existence of
$\EFI$ and hardness of $\Hbar$-complexity estimation over single-copy
samplable state families.

\subsection{A unifying viewpoint on various notions of Kolmogorov complexity of~states}
\label{sec:24}
While so far we have primarily focused on $\Hbar$, we also consider other complexity notions (e.g., $\Umin$) and their robust versions, and prove characterizations of $\EFI$ from the hardness of estimating these complexity notions. As a final contribution, we provide a potentially unifying viewpoint on several of these notions in terms of the following problem.

Let $\Pi_r$ be the span of all the states with $\Knet$-complexity at most $r$ (where informally, $\Knet$-complexity is a notion introduced by Mora and Briegel~\cite{mora2004algorithmic} that captures the minimum-length of a program that can output the state). Let $\{\ket{\psi_k}\}$ be an efficiently samplable family of states with the promise that each state either lies (almost) entirely in $\Pi_{r}$ or it is (almost) orthogonal to $\Pi_{r+\omega(\log n)}$; given a single copy of the state, decide which is the case. 

Note that deciding whether the state is in the ``span of easy states'' or not is the best that one can hope for given only a single copy. One cannot distinguish the easy states from the (potentially hard) states in the span of the easy states given a single copy since their density matrices might be statistically very close. In~\cref{thm:EFI-span}, we prove that the existence of $\EFI$ is equivalent to the existence of a non-uniform family of states on which the above problem is hard.

An important remark is that the ``span of easy states'' is a notion that is not robust against the choice of the universal gate set used to define $\Knet$-complexity. Although, from Solovay-Kitaev, any two universal gate sets can approximate each other to arbitrary precision, the span of easy states can be very different even if we only perturb the easy states a little bit: for example, the span of $\{\ket{0}, \sqrt{1-2^{-200}}\ket 0 + 2^{-100} \ket 1\}$ is exactly $\mathsf{span}\{\ket 0, \ket1\}$, while the span of $\{\ket{0}, \sqrt{1-2^{-200}}\ket 0 + 2^{-100} \ket 2\}$ is exactly $\mathsf{span}\{\ket 0, \ket 2\}$. Two almost identical state families might have very different spans.

In order to relate other notions of complexity to the ``overlap'' on the span of easy states, we need to introduce a notion of ``robust'' span, which characterizes the significant components of the state family. A bit more formally, the ``robust'' span of a state family $\{\ket{\psi_k}\}$ is the subspace spanned by all the significant eigenvalues of $\E_k \ketbra{\psi_k}{\psi_k}$, where the expectation is taken over uniformly random $k$. This subspace is much more robust against perturbations, and, in~\cref{sec:relating-entropy-span}, we are able to tightly relate $\Hbar$ to the ``robust'' span: a state has low $\mathsf{H}$-complexity if and only if it has a large overlap onto the robust span of easy states. This relation allows us to give an alternative proof of the characterization of $\EFI$ from the hardness of $\GapH$. With this relation in hand, we can summarize the proof of the latter characterization as follows: 
\begin{itemize}
\item If $\EFI$ exist, then a non-uniform state family with a $\Hbar$-complexity gap exists (as our~\cref{thm:intro-EFI-1PRS} says that $\EFI$ implies non-uniform $\OnePRS$).
\item For the converse direction, our proof goes through the following steps: we show that a state family with a $\Hbar$-complexity gap also has a gap in overlap with the ``robust span of easy states''; the latter gap implies a gap in von Neumann entropy; finally, if $\EFI$ do not exist, such an entropy gap can be detected using quantum randomness extractors.
\end{itemize}

\section{Preliminaries}

In this section, we introduce the fundamentals of quantum information, quantum
extractors, quantum cryptography, and the necessary tools.

\subsection{Quantum information}

Our notation for quantum information mainly follows~\cite{NC02}.
We refer the reader to~\cite{NC02} for a more detailed discussion.
For a given Hilbert space $\H$, we use $\Lin(\H)$ and $\Pos(\H)$ to denote the
set of linear operators and positive semi-definite operators on $\H$.
For $A, B \in \Lin(\H)$, we write $A \ge B$ or $B \le A$ if
$A - B \in \Pos(\H)$.
The Schatten $1$-norm of a linear operator $A$ is defined as
\begin{equation}\label{eq:1-norm}
  \norm{A}_1 = \Tr(\sqrt{A^\dagger A}).
\end{equation}

A pure quantum state is a unit vector in a Hilbert space $\H$.
A mixed quantum state is a density matrix $\rho$ in $\Pos(\H)$ with unit trace.
We use $\Density(\H)$ to denote the set of density matrices on a Hilbert space
$\H$.
When considering multiple quantum systems, we use labels such as $A$ and $B$ to
refer to different systems and use $\H_A$ and $\H_B$ to denote the corresponding
Hilbert spaces.
For example, $\rho_{AB} \in \Density(\H_A \otimes \H_B)$ is a density matrix
describing a mixed state on the joint $AB$ system.
We also consider sub-normalized states $\rho$ where $\rho \in \Pos(\H)$ and
$\Tr(\rho) \le 1$.
In this case, the matrix $\rho$ is called a semi-density matrix.
We use $\SemiD(\H)$ to denote the set of semi-density matrices.
Every density matrix $\rho_A \in \Density(\H)$ has a purification
$\ket{\psi}_{AB}$ such that $\rho_A = \Tr_B (\ketbra{\psi}{\psi}_{AB})$, where
$\Tr_B$ is the partial trace over $B$.

Various distance measures between two quantum states are needed.
For density matrices $\rho$ and $\sigma$, we use $D(\rho, \sigma)$ to denote
their trace distance:
\begin{equation*}
  D(\rho, \sigma) = \frac{1}{2} \norm{\rho - \sigma}_1.
\end{equation*}
The trace distance generalizes the total variation distance between probability
distributions.
For pure states $\ket{\psi}$ and $\ket{\phi}$, the trace distance can be
computed as $\sqrt{1 - \abs{\braket{\psi | \phi}}^2}$.

Another commonly used quantity measuring the closeness of two quantum states is
the fidelity:
\begin{equation*}
  F(\rho, \sigma)= \norm{\sqrt{\rho} \sqrt{\sigma}}_1.
\end{equation*}
For pure states $\ket{\psi}, \ket{\phi} \in \H$, we have
$F(\ket{\psi}, \ket{\phi}) = \abs{\braket{\psi | \phi}}$.
Fidelity can be seen as a quantum generalization of the Bhattacharyya
coefficient for two probability distributions,
$\mathrm{BC}(p,q) = \sum_i \sqrt{p_i q_i}$.
Uhlmann's theorem characterizes the fidelity of two mixed states
$\rho, \sigma \in \Density(\H_A)$ as the maximum overlap of their purifications:
$F(\rho, \sigma) = \max_{\ket{\psi}, \ket{\phi}} \abs{\braket{\psi|\phi}}$,
where $\ket{\psi}, \ket{\phi} \in \H_A \otimes \H_B$ are purifications of $\rho$
and $\sigma$, respectively, and system $B$ is a copy of system $A$ with the same
dimension.

\begin{lemma}[Fuchs-van de Graaf inequalities, cf., for example, Theorem 3.33
  of~\cite{Wat18}]\label{lem:ineq-fuchs-vdgraff}
  For any mixed states $\rho$ and $\sigma$, we have
  \begin{equation*}
    D (\rho, \sigma) \geq 1 - F(\rho, \sigma).
  \end{equation*}
\end{lemma}

\begin{lemma}\label{lem:distance-mix}
  Let $\{\rho_{k}\}$ be a family of quantum states.
  If for all $k$, there is a state $\rho'_{k}$ such that
  $D(\rho_{k}, \rho'_{k}) \le \delta$, we have
  $D \bigl(\E_{k} \rho_{k}, \E_{k} \rho'_{k} \bigr) \le \delta$.
\end{lemma}

\begin{proof}
  This follows directly from the joint convexity of the trace distance.
\end{proof}

The trace distance and fidelity can be generalized to sub-normalized
states~\cite{TCR10}.
For two semi-density matrices $\rho, \sigma \in \SemiD(\H)$,
\begin{align*}
  D (\rho, \sigma) & = \frac{1}{2}\norm{\rho - \sigma}_1 +
                     \frac{1}{2}\abs{\Tr (\rho) - \Tr (\sigma)},\\
  F (\rho, \sigma) & = \norm{\sqrt{\rho}\sqrt{\sigma}}_1 +
                     \sqrt{1 - \Tr(\rho)}\sqrt{1 - \Tr(\sigma)}.
\end{align*}
These generalizations are obtained by extending the system and considering the
normalized states $\rho \oplus (1-\Tr(\rho))$ and $\sigma \oplus (1 - \Tr(\sigma))$,
so many properties of the trace distance and fidelity carry over to the
generalized versions.

For two semi-density matrices $\rho, \sigma \in \SemiD(\H)$, the purified
distance~\cite{TCR10,Tom12} is defined as
\begin{equation*}
  P(\rho, \sigma) = \sqrt{1 - F^2(\rho, \sigma)}.
\end{equation*}
It is known that $P$ is a metric on $\SemiD(\H)$ and that
$P(\rho, \sigma) \ge D(\rho, \sigma)$.
The purified distance is used to define the smoothed versions of min-entropy
considered below.

For our purposes, we require the following definitions of quantum entropies.
The von Neumann entropy of a state $\rho \in \Density(\H)$ is defined as
$S(\rho) = - \Tr (\rho \log \rho)$.
If $\rho \in \Density(\H_A)$ is a density matrix of system $A$, we may also
write $S(A)_{\rho}$ to denote the von Neumann entropy, and may omit the
subscript $\rho$ and simply write $S(A)$ when no confusion arises.
The quantum conditional entropy is defined as
$S(A|B)_\rho = S(AB)_\rho - S(B)_\rho$.
Unlike the classical case, the quantum conditional entropy can be negative.
The relative entropy of $\rho$ with respect to $\sigma$ is
$S(\rho \Vert \sigma) = \Tr(\rho \log \rho - \rho \log\sigma)$.

The quantum min-entropy of a sub-normalized quantum state $\rho \in \SemiD(\H)$
is $\Hmin(\rho) = \sup\{ \lambda \in \real \mid \rho \le 2^{-\lambda}I\}$.
The $\eps$-smooth min-entropy of a sub-normalized quantum state
$\rho \in \SemiD(\H)$ is
\begin{equation*}
  \Hmin^\eps (\rho) = \sup_{\rho' \in \SemiD(\H),
    P(\rho, \rho') \le \eps} \Hmin(\rho').
\end{equation*}

The quantum conditional min-entropy of a semi-density matrix
$\rho_{AB} \in \SemiD(\H_{AB})$ is defined as
\begin{equation*}
  \Hmin(A | B)_{\rho} = \sup_{\sigma_B \in \Density(\H_B)}
  \sup \{ \lambda \in \real \mid \rho \le
  2^{-\lambda}I_A \otimes \sigma_B\}.
\end{equation*}
Its smoothed version is given by
\begin{equation*}
  \Hmin^\eps (A | B)_{\rho} = \sup_{\rho'_{AB} \in \SemiD(\H_{AB}),
      P(\rho, \rho') \le \eps} \Hmin(A | B)_{\rho'}.
\end{equation*}

The quantum max-entropy, according to~\cite{Ren08}, is defined as\footnote{The
  quantum max-entropy has definitions by Renner and Tomamichel.
  Ref~\cite{Ren08} defines the max-entropy as the R\'{e}nyi entropy of order
  $0$, while~\cite{Tom12} defines the max-entropy as the R\'{e}nyi entropy of
  order 1/2.
  We adopt the definition from~\cite{Ren08} as it has a better operational
  meaning.}
\begin{equation*}
  \Hmax(\rho) = \log (\rank(\rho)).
\end{equation*}
Its smoothed version is given by
\begin{equation*}
  \Hmax^\eps(\rho) = \inf_{\rho'\in \SemiD(\H),
    P(\rho, \rho') \le \eps}\Hmax(\rho').
\end{equation*}

The following inequality between smoothed min-entropy and von Neumann entropy is
implicit in~\cite{Tom12}.

\begin{lemma}[Result 5 and Corollary 6.5
  of~\cite{Tom12}]\label{lem:min-entropy-bound}
  Let $\rho$ be an arbitrary quantum state on systems $A$ and $B$ and let $R$ be
  the quantum system of its purification. Let $n$ be the number of qubits of system A.
  For $g(\eps) = - \log(1-\sqrt{1-\eps^2})$, $0<\eps<1$, and
  $m \ge \frac{8}{5} g(\eps)$, we have
  \begin{align*}
    \frac{1}{m} \Hmin^{\eps} (A^m|B^m)_{\rho^{ \otimes m}} & \ge
    S(A|B)_\rho - \frac{2(n + 4)\sqrt{g(\eps)}}{\sqrt{m}},\\
    \frac{1}{m} \Hmax^{\eps} (A^m|B^m)_{\rho^{\otimes m}} & \le
    S(A|B)_\rho + \frac{2(n + 4)\sqrt{g(\eps)}}{\sqrt{m}}.
  \end{align*}
\end{lemma}

\begin{corollary}\label{cor:min-entropy-explicit-bound}
  Let $\rho$ be an arbitrary quantum systems on systems A and B.
  Let $n$ be the number of qubits of system A.
  For any $m \geq 5 \log \frac 1\eps$, $\eps < 1/2$ and $n>10$, we have
  \begin{align*}
    \frac{1}{m}  \Hmin^\eps(A^m|B^m)_{\rho^{\otimes m}} &\ge
    S(A|B)_\rho - 6n \sqrt{\frac{\log (1/\eps)}{m}},\\
    \frac{1}{m}  \Hmax^\eps(A^m|B^m)_{\rho^{\otimes m}} &\le
    S(A|B)_\rho + 6n \sqrt{\frac{\log (1/\eps)}{m}}.\\
  \end{align*}
\end{corollary}

\begin{proof}
  This is just a direct application of~\cref{lem:min-entropy-bound} and the
  observation that
  $g(\eps) = - \log(1-\sqrt{1-\eps^2}) \leq 3\log \frac{1}{\eps}$.
\end{proof}

\begin{remark}
    When we take $B$ as the trivial system in~\cref{lem:min-entropy-bound} and~\cref{cor:min-entropy-explicit-bound}, then we get the bound of $\Hmin^\eps(\rho^{\otimes m})$ and $\Hmax^\eps(\rho^{\otimes m})$ with $S(\rho)$.
\end{remark}

The following helper lemma will be useful in entropy estimation.

\begin{definition}[Almost Orthogonality]\label{def:AO}
  We call two $n$-qubit quantum mixed states $\rho,\sigma$ are almost orthogonal
  if there exists a projector $\Pi$ such that $\Tr((I-\Pi)\rho)\leq \negl(n)$,
  and $\Tr(\Pi\sigma)\leq\negl(n) $.
\end{definition}

\begin{lemma}\label{lem:TDtoAO}
  If two $n$-qubit states $\rho,\sigma$ have trace distance
  $D(\rho,\sigma)\geq 1-\negl(n)$, $\rho$ and $\sigma$ are almost orthogonal.
\end{lemma}

\begin{proof}
  The proof of the lemma follows from the operational meaning of trace distance
  via setting the projector $\Pi$ as the pretty good measurement.
\end{proof}

\begin{lemma}\label{lem:AOadd}
  For two almost orthogonal $n$-qubit states $\rho,\sigma$,
  \begin{equation*}
    \abs{S\Bigl(\frac{\rho+\sigma}{2}\Bigr) -
      \frac{S(\rho)+S(\sigma)}{2} - 1 } \leq \negl(n).
  \end{equation*}
\end{lemma}

\begin{proof}
  By~\cref{def:AO}, there exists a projector $\Pi$ such that
  $\Tr((I-\Pi)\rho) \leq \negl(n)$, and $\Tr(\Pi\sigma) \leq \negl(n)$.
  We denote $\Tilde{\rho} = \frac{\Pi\rho\Pi}{\Tr(\Pi\rho\Pi)}$ and
  $\Tilde{\sigma} = \frac{(I-\Pi) \sigma(I-\Pi)}{\Tr((I-\Pi) \sigma(I-\Pi))}$.

  By the gentle measurement lemma~\cite{Win99}, we have that
  $D(\rho, \Tilde{\rho}) \leq \negl(n)$ and
  $D(\sigma, \Tilde{\sigma}) \leq \negl(n)$, and thus
  $D((\rho+\sigma)/2, (\Tilde{\rho} + \Tilde{\sigma})/2) \leq \negl(n)$.
  By Fannes inequality~\cite{Fan73}, we have that $\abs{S(\Tilde{\rho}) - S(\rho)}$,
  $\abs{S(\Tilde{\sigma}) - S(\sigma)}$, and
  $\abs{S((\rho+\sigma)/2) - S((\Tilde{\rho} + \Tilde{\sigma})/2)}$ are negligible.
  Since $\Tilde{\rho}$ and $\Tilde{\sigma}$ are orthogonal, we have that
  \begin{equation*}
    S \left(\frac{\Tilde{\rho} + \Tilde{\sigma}}{2}\right)
    = \frac{S(\Tilde{\rho}) + S(\Tilde{\sigma})}{2} + 1.
  \end{equation*}
  The lemma follows from a triangle inequality.
\end{proof}

\begin{definition}
Let $n(\cdot)$ be a polynomial function.  We say that ${\{\ket{\psi_k}\}}_{k \in \bit^{n(\lambda)},\lambda \in \natural}$ is an \emph{efficiently
    samplable distribution of keyed states} if there exists a QPT quantum
  algorithm $G$ such that, for all $k$ of length $n(\lambda)$, we have $G(1^{\lambda}, k) = \ket{\psi_{k}}$.
  We say that ${\{\ket{\psi_{k}}\}}_{{k \in \bit^{n(\lambda)}, \lambda \in \natural}}$ is a
  \emph{non-uniform} efficiently samplable distribution of keyed states if there
  exists a \emph{non-uniform} quantum polynomial-time algorithm $G$ such that, for all $k$ of length $n(\lambda)$, $G(1^{\lambda}, k) = \ket{\psi_{k}}$. We sometimes denote the advice string as $a$, and refer to the length of $a$ as the advice size of the family.
\end{definition}

\begin{definition}\label{def:single-copy-samplable-family}
    We say that $\{\ket{\psi_k}\}_{k\in \{0,1\}^*}$ together with a family distribution $\{\Dist_n\}_{n \in \natural}$ is a single-copy samplable family if $\Dist_n$ is a distribution on $\ket{\psi_k}$ with $k \in \{0,1\}^n$ and there exists a QPT algorithm $\adv$ that takes in $1^n$ as input and outputs a pair $(k, \ket{\psi_k})$ according to the distribution $\Dist_n$.
\end{definition}

\subsection{Quantum cryptographic primitives}

We recall several existing quantum cryptographic primitives used in this paper.
$\EFI$ pairs were first proposed by Brakerski, Canetti, and
Qian~\cite{BCQ22}.

\begin{definition}[$\EFI$]\label{def:EFI}
  We call two families of mixed states $\{\rho_{0, \lambda}\}_\lambda$ and
  $\{\rho_{1, \lambda}\}_\lambda$ an $\EFI$ pair if the following conditions hold:
  \begin{description}
    \item[Efficient Generation:] There exists a QPT algorithm $G$ that takes
          input $(1^\lambda, b)$ for integer $\lambda$ and $b \in \bit$, and
          outputs the mixed state $\rho_{b, \lambda}$.
    \item[Statistically Far:]
          $D(\rho_{0, \lambda}, \rho_{1, \lambda}) \ge 1-\negl(\lambda)$.
    \item[Computational Indistinguishability:] For any QPT adversary $\adv$, we
          have
          \begin{equation*}
            \abs{\Pr[\adv(1^\lambda, \rho_{0, \lambda})=1] -
              \Pr[\adv(1^{\lambda}, \rho_{1, \lambda})=1]} \leq \negl(\lambda).
          \end{equation*}
  \end{description}
\end{definition}

\begin{remark}
    In the original definition of $\EFI$ in~\cite{BCQ22}, they only required $1/\poly(\lambda)$ statistical gap. It is easy to show the equivalence between their definition and the current definition by considering polynomial copies of their states (say, by~\cite{BQSY24}).
\end{remark}

We can relax the definition of $\EFI$ to allow other gaps in the statistical
distance and computational indistinguishability.
\begin{definition}[Weak $\EFI$]\label{def:weakEFI}
  For any functions $\eps(\cdot)$ and $\delta(\cdot)$, we call two families of
  mixed states
  $\{\rho_{0, \lambda}\}_\lambda, \{\rho_{1, \lambda}\}_\lambda$ an
  $(\eps, \delta)$-weak $\EFI$ pair if the following conditions hold:
  \begin{description}
    \item[Efficient Generation:] There exists a QPT algorithm $G$ that takes
          input $(1^\lambda, b)$ for integer $\lambda$ and $b \in \bit$, and
          outputs the mixed state $\rho_{b, \lambda}$.
    \item[Statistically $(1-\eps)$-Far:]
          $\TD(\rho_{0, \lambda}, \rho_{1, \lambda}) \ge 1-\eps(\lambda)$.
    \item[$\delta$-Computational Indistinguishability:] For any QPT adversary
          $\adv$, we have
          \begin{equation*}
            \abs{\Pr[\adv(1^{\lambda}, \rho_{0, \lambda}) = 1] -
              \Pr[\adv(1^{\lambda}, \rho_{1,\lambda}) = 1]} \le
            \delta(\lambda) + \negl(\lambda).
          \end{equation*}
  \end{description}
\end{definition}

In certain parameter regime, weak $\EFI$ pairs imply standard $\EFI$ pairs.
The following theorem is implicit in~\cite{BQSY24}.

\begin{theorem}[Weak $\EFI$ implies $\EFI$]\label{thm:WEFI-EFI}
  The existence of an $(\eps, \delta)$-weak $\EFI$ pair family in any of the following parameter range implies a standard $\EFI$
  \begin{itemize}
    \item  $\eps=\negl(\lambda), \delta=1 - \frac{1}{\poly(\lambda)}$.
    \item $\eps = 1 -\frac{1}{\poly(\lambda)}, \delta = \negl(\lambda)$.
    \item ${(1-\eps)}^2 - \sqrt \delta \geq C$ for some universal constant $C$ independent of $\lambda$.
  \end{itemize}
  implies the existence of a standard $\EFI$ pair family.
  
\end{theorem}

\begin{definition}[Single-copy Pseudorandom State ($\OnePRS$)]
  A state family $\{\ket{\phi_{k}}\}$ of $n(\lambda)$ qubits and key length
  $\ell(\lambda)$ is called single-copy pseudorandom if the following conditions
  hold:
  \begin{description}
    \item[Efficient Preparation:] There is a QPT algorithm $G$ that on input
          $(1^{\lambda}, k)$ prepares the state $\ket{\phi_{k}}$.
    \item[Single-copy Pseudorandom Property:] For any QPT adversary $\adv$, there
          exists a negligible function $\negl(\cdot)$ such that for all $\lambda \in \natural$ and
          $n=n(\lambda)$,
          \begin{equation}\label{eq:1prs-def}
            \abs{\Pr_{k} \bigl[ \adv(1^{\lambda}, \ket{\phi_k}) = 1 \bigr] -
              \Pr_{\ket{\phi}\gets \mu_{n} } \bigl[ \adv(1^{\lambda},
              \ket{\phi}) = 1 \bigr] } = \negl(\lambda),
          \end{equation}
          where $\mu_{n}$ is the Haar measure on $n(\lambda)$-qubit states.
    \item [Stretch Property:] The length $\ell(\lambda)$ of the key $k$ is
          strictly less than the number of qubits $n(\lambda)$.
  \end{description}
\end{definition}

\subsection{Quantum extractors}

As a natural generalization of classical extractors, there are also studies of quantum extractors in the context of quantum information~\cite{BFW14,DBWR14}. We rephrase their definition and results to show their similarity with classical extractors.

\begin{definition}[$(k,\eps,\delta)$-Quantum Strong Extractor]\label{def:strong-quantum-extractor}
  Let $\ell \in \natural$ and $A = A_1A_2$ be a quantum system with $A_1$ and $A_2$ as subsystems, where the subsystem $A_1$ consists of $\ell$ qubits. A collection of quantum unitaries $\{U_j\}_{j\in L}$ acting on system $A$ is
  called a $(k,\eps, \delta)$-quantum strong extractor that extracts $\ell$
  qubits if for any quantum state $\rho_{AE}\in \Density(\H_{A} \otimes \H_E)$
  with $H^{\delta}_{\min}{(A|E)}_\rho\geq k$,
  \begin{equation*}
    D \biggl(\frac{1}{\abs{L}} \sum_{j\in L} \ket{j}\bra{j}\otimes \Tr_{A_2} \bigl(U_j\rho_{AE} U_j^\dagger \bigr),
    \frac{I_L}{\abs{L}} \otimes \frac{I_{A_1}}{\abs{A_1}} \otimes \rho_E \biggr) \le \eps.
  \end{equation*}
\end{definition}

To construct quantum extractors, we need the unitary $t$-design.

\begin{definition}
    An ensemble of quantum unitaries $\{U_r\}_{r\in R}$ is called a unitary $t$-design if for all $M\in\Lin(A)$:
    \begin{align*}
        \frac{1}{|R|}\sum_{r\in R}{U_r^{\otimes t}M {U_r^\dag}^{\otimes t}}=\int U^{\otimes t}M {U^\dag}^{\otimes t}{\rm d}\eta(U),
    \end{align*}
    where $\eta$ is the Haar measure over $U(A)$.
\end{definition}

Unitary $t$-designs can be viewed as the quantum analog of classical $t$-wise independent hash functions.

We use the Choi-Jamio\l kowski isomorphism.

\begin{definition}
    The Choi-Jamio\l kowski map $J$ takes maps $\mathcal{T}_{A \rightarrow B}: \Lin(\H_A) \rightarrow \Lin(\H_B)$ to operators $J(\mathcal{T}_{A \rightarrow B})$ in $\Lin(\H_A \otimes \H_B)$. It is defined as
    \[J(\mathcal{T}_{A \rightarrow B}) = (\mathcal{I}_A \otimes \mathcal{T}_{A' \rightarrow B})(\ketbra{\phi}{\phi}_{AA'}),\]
    where $\ket{\phi}_{AA'} = \frac{1}{\sqrt{|A|}}\sum_i \ket{i}_A \otimes \ket{i}_{A'}$.
\end{definition}

\begin{remark}
It is well known that this map is in fact an isomorphism: The Choi-Jamio\l kowski map $J$ bijectively maps the set of completely positive maps from $\H_A$ to $\H_B$ to the set $\Pos(\H_A \otimes \H_B)$, and its inverse maps any $\gamma_{AB} \in \Pos(\H_A \otimes \H_B)$ to \[\mathcal{T}_{A \rightarrow B}:M_A \mapsto |A|\Tr_{A}\left[\gamma_{AB}M_A^T\right]\]

        Using this, $J(\mathcal{T}_{A \rightarrow B})$ is called the Choi-Jamio\l kowski representation of $\mathcal{T}_{A \rightarrow B}$.
\end{remark}

We recall the decoupling theorem proved in~\cite[Theorem 3.1]{DBWR14}.

\begin{theorem}[Decoupling Theorem]\label{thm:decoupling}
  Let $\mathcal{T}_{A \to B}$ be a completely-positive
  map with Choi-Jamio\l kowski representation $\tau_{AB} = J(\mathcal{T})$ such
  that $\Tr(\tau_{AB}) \le 1$.
  Let $E$ be a quantum system for the environment.
  Then, for $\eps > 0$, $\rho_{AE} \in \Density(\H_A\otimes \H_E)$, and any unitary
  $2$-design $\{U_j\}_{j \in L}$ on $A$, we have
  \begin{equation*}
    \frac{1}{|L|} \sum_{j \in L} \norm{ \mathcal{T} \bigl(U_j \rho_{AE}
    U_j^\dagger \bigr), \tau_B \otimes \rho_E}_1
    \le 2^{-\frac{1}{2}\Hmin^{\eps}(A | E)_{\rho} -
      \frac{1}{2}\Hmin^{\eps}(A | B)_{\tau}} + 12 \eps.
  \end{equation*}
\end{theorem}

\begin{lemma}\label{lem:extractor_bound}
  Let $k \in \interval{-n}{n}$ and $\eps \in \interval[open]{0}{1}$.
  Let $E$ be a quantum system for the environment and $A=A_1A_2$ be an $n$-qubit quantum system with $A_1$ and $A_2$ as subsystems, where the subsystem $A_1$ consists of at most $\frac{n+k}{2} - \log (1/\eps)$ qubits.
  Let $\rho_{AE} \in \Density(\H_A \otimes \H_E)$ be a density matrix on systems
  $A$ and $E$ having smoothed conditional min-entropy $\Hmin^{\eps/12}(A|E) \ge k$.
  Then, for any unitary $2$-design $\{U_j\}_{j \in L}$ on $A$,
  \begin{equation*}
    \frac{1}{2|L|} \sum_{j \in L} \norm{\Tr_{A_2} \Bigl(U_j \rho_{AE} U_j^\dagger \Bigr) -
      \frac{I_{A_1}}{\abs{A_1}} \otimes \rho_E}_1 \le \eps.
  \end{equation*}
\end{lemma}

\begin{proof}[Proof of \cref{lem:extractor_bound}]
  We apply \cref{thm:decoupling} to prove the statement and choose $B$ to be the
  subsystem $A_1$, and $\mathcal{T}$ to be the partial trace over subsystem
  $A_2$.

  By definition, we compute the Choi-Jamio\l kowski state
  \begin{equation*}
    \tau_{AB} = \frac{1}{\abs{A_1}} \sum_{x,y} {\ket{x}\bra{y}}_{A_1}
    \otimes \frac{I_{A_2}}{\abs{A_2}} \otimes \ket{x}\bra{y}_{B}.
  \end{equation*}
  Its reduced density matrix $\tau_B = \frac{I_B}{\abs{B}}$ is maximally mixed.
  The conditional min-entropy $\Hmin{(A|B)}_\tau$ of $\tau_{AB}$ can be computed
  as
  \begin{equation*}
    \begin{split}
      \Hmin{(A|B)}_\tau & = \log\abs{A_2} + \Hmin{(A_1|B)} \\
                        & = \log\abs{A_2} - \log\abs{A_1} \\
                        & = n - 2 \log\abs{A_1} \\
                        & \ge 2\log(1/\eps) - k,
    \end{split}
  \end{equation*}
  where the last step follows from the condition on the number of qubits in
  system $A_1$.
  Hence, we can bound
  \begin{equation}\label{eq:ext_bound}
    \Hmin^{\eps/12}(A|E)_\rho + \Hmin^{\eps/12}(A|B)_\tau \ge
    \Hmin^{\eps/12}(A|E)_\rho + \Hmin(A|B)_\tau \ge 2 \log (1/\eps),
  \end{equation}
  where the first inequality follows from the definition of smoothed min-entropy
  taking the maximum over close states.

  The conditions of \cref{thm:decoupling} are all met and we have
  \begin{equation*}
    \frac{1}{|L|} \sum_{j \in L} \norm{\Tr_{A_2} \Bigl(U_j \rho_{AE} U_j^\dagger \Bigr) -
      \frac{I_{A_1}}{\abs{A_1}}\otimes \rho_E }_1
    \le 2^{-\frac{1}{2}\Hmin^{\eps/12}{(A | E)}_{\rho} -
      \frac{1}{2}\Hmin^{\eps/12}{(A | B)}_{\tau}} + \eps \le 2\eps,
  \end{equation*}
  where the last step follows from \cref{eq:ext_bound}.
\end{proof}

\begin{theorem}\label{thm:extractor}
    Let $n,\ell \in \natural$, $k \in \interval{-n}{n}$, and $\eps \in \interval[open]{0}{1}$ such that $\ell \leq \frac{n+k}{2}-\log(1/\eps)$. Then any unitary $2$-design on $n$-qubit system is a $(k,\eps, \eps/12)$-quantum strong extractor that extracts $\ell$ qubits.
\end{theorem}

\begin{proof}
  By \Cref{def:strong-quantum-extractor}, the condition for the unitary $2$-design $\{U_j\}_{j\in L}$ on $n$-qubit system $A$ to be such a quantum strong extractor is that for every state $\rho$ on systems $A$ and $E$ having smoothed conditional min entropy $\Hmin^{\eps/12}{(A|E)}_\rho \ge k$,
  \begin{equation*}
    D \biggl(\frac{1}{\abs{L}} \sum_{j\in L} \ket{j}\bra{j}\otimes
    \Tr_{A_2} \bigl(U_j\rho_{AE} U_j^\dagger \bigr),
    \frac{I_L}{\abs{L}} \otimes \frac{I_{A_1}}{\abs{A_1}}
    \otimes \rho_E \biggr) \le \eps,
  \end{equation*}
  where $A_1$ consists of $\ell$ qubits, and $A_2$ consists of the other $n - \ell$ qubits.
  The left hand side is the trace distance of two block diagonal matrices
  indexed by $j$ and can hence be simplified to
  \begin{equation*}
    \frac{1}{|L|} \sum_{j\in L} D\biggl(\Tr_{A_2} \Bigl(U_j \rho_{AE} U_j^\dagger \Bigr),
      \frac{I_{A_1}}{\abs{A_1}} \otimes \rho_E\biggr).
  \end{equation*}
  The statement then follows from
  \cref{lem:extractor_bound}.
\end{proof}

For a given unitary $2$-design $\{U_j\}_{j \in L}$ on an $n$-qubit quantum
system $A = A_1A_2$, we introduce the notation $\Ext_{\ell}^{A \to A_{1}}(\cdot)$, or
simply $\Ext_\ell^{A}(\cdot)$, to denote the extractor that extracts $\ell$
qubits (on $A_{1}$) from the input system $A$.
That is,
\begin{equation*}
  \Ext^{A \to A_{1}}_\ell(\rho_{AE}) = \frac{1}{\abs{L}} \sum_{j\in L}
  \ket{j}\bra{j}\otimes \Tr_{A_2} \bigl(U_j\rho_{AE} U_j^\dagger \bigr).
\end{equation*}

The claim in the above theorem can be written as
\begin{equation*}
  D \Bigl(\Ext^{A}_\ell(\rho_{AE}), \frac{I_L}{\abs{L}} \otimes
  \frac{I_{A_1}}{\abs{A_1}} \otimes \rho_E \Bigr) \le \eps,
\end{equation*}
for $\ell \leq \bigl(n + \Hmin^{\eps/12}(A|E)\bigr)/2 - \log(1/\eps)$.

\subsection{Kolmogorov complexity}
\label{sec:kolmogorov}

In this paper we will use several notions of Kolmogorov
complexity. The most well-known of these is the standard prefix-free
classical Kolmogorov complexity of binary strings,
which we denote by $\Kc(x)$. 

\begin{definition}\label{def:K-complexity}
    Let $U$ be a universal 
    prefix-free 
    Turing machine.
    For strings $x \in \{0,1\}^*$, 
    the 
    Kolmogorov complexity $\Kc_U(x)$ of $x$ is the length of
    the shortest program $p$ such that $U(p)$ will halt and output $x$
    after a finite number of steps.
\end{definition}
We clarify that, here and in the rest of the paper, by ``length of a
program $p$'' we mean the length of the string corresponding to $p$ when
viewed as an input to the universal 
Turing machine $U$ (we do not mean the
size of the program $p$ in terms of some set of gates).

Because for any two universal Turing machines $U$ and $V$ there exists some
constant $c$ such that for all $x$, $|\Kc_U(x)-\Kc_V(x)|<c$, the choice of
universal Turing machine is unimportant to us. So, we choose to fix some
universal Turing machine $U$, and simply write $\Kc(x)$, dropping the
subscript.

In this work we are interested in the complexity of quantum states,
and we will use several generalizations of Kolmogorov complexity that allow
us to measure their complexity. The first of these was introduced by Mora
and Briegel~\cite{mora2004algorithmic,mora2006quantum} and can be thought
of as measuring the amount of classical information required to generate a
good approximation of the state of interest. It is defined relative to some
choice of a universal classical Turing machine $U$ and a universal quantum
gate set $B$. It measures the length of the shortest program on which the
universal classical Turing machine outputs a description of a quantum
circuit $C$ that, when 
given $\ket{0 \cdots 0}$ as input, outputs a state
which is $\eps$-close to the state of interest. 
\begin{definition}[\emph{$\Knet$-complexity}~\cite{mora2004algorithmic, mora2006quantum}]\label{def:Knet-complexity}
Let $U$ be a universal Turing machine, $B$ a universal set of quantum
gates, and $\mathcal{C}^B$ be the set of circuits composed of gates from $B$. Let $\eps \in [0,1]$. For a pure state $\ket{\psi}$, we define
its \emph{$\Knet$-complexity} as:
\begin{equation*}
  \Knet^{U,B,\eps} (\ket{\psi}) = \min_{p}
  \bigl\{\abs{p} : C = U(p) \in 
      \mathcal{C}^B
      \text{ and }
  \abs{\bra{\psi}C \ket{0^{m}}}^2 \geq 1 - \eps \bigr\},
\end{equation*}
where the minimum is taken over program descriptions $p$, and 
$\mathcal{C}^B$
is the set of quantum circuits of finite size consisting of gates from $B$ (here, $m$ denotes the size of inputs to $C$, which can depend on $C$).
\end{definition}

Since $B$ and $U$ are universal, this definition changes only by a just barely superconstant factor when we
change our choice of $U$ or $B$ (see
\cref{sec:robust-qk-gateset} for details). So, going forward, we will
simply fix a choice of $U$ and $B$ and write
$\Knet^{\eps}(\ket{\psi})$, dropping the superscripts.

The second notion that we will use was introduced by G\'{a}cs~\cite{Gac01} and
generalizes the definition of the classical Kolmogorov complexity $\Kc(x)$
when viewed as the negative logarithm of the probability of $x$ being
output by the ``universal distribution''.
When $U$ is prefix-free,
the (classical) universal semi-distribution\footnote{
    A semi-distribution is a distribution where the total probability adds up to
some value less than one, which here corresponds to the probability that
we sample a program which halts. 
This is possible because $U$ is prefix-free.
} $D_U$ is defined as follows: sample 
a
uniformly random program $p$,

run $U(p)$, and return its
output. Defining $\Kc(x) = -\log(\Pr[x \sim D_U])$ results in a notion
equivalent to the one from Definition \ref{def:K-complexity} 
up to an additive constant.

We can generalize $D_U$ to the universal \emph{semi-density
matrix}\footnote{
    Similar to a semi-distribution,
    a semi-density matrix is defined as
    some $\sum_i c_i\ketbra{\phi_i}{\phi_i}$ where the $c_i$'s add up to less
    than or equal to 1.
} $\udm$ in one
of several equivalent ways. G\'{a}cs chooses to take the outputs of $U(p)$ and
interpret them as vectors of complex numbers describing the amplitudes of a
state. He then takes $\udm$ to be the resulting semi-density
matrix from
picking state $\ket{\psi}$ with the probability that $D_U$ would output the
vector corresponding to $\ket{\psi}$. We can take an approach closer to
that of Mora and Briegel and equivalently define $\udm_n$ to be the
semi-density matrix 
resulting from picking state $\ket{\psi}$ over $n$ qubits with the
probability that $D_U$ outputs a classical description of a quantum circuit
$C$ such that $C\ket{0} = \ket{\psi}$ (so we only consider circuits
outputting $n$-qubit states). Towards this,
we fix some finite universal set of quantum gates, and consider circuits
consisting of gates from this set. Given this notion of universal semi-density matrix,
G\'{a}cs' notion of state complexity is the following. 
\begin{definition}[``$\Hbar$-complexity''~\cite{Gac01}]
\label{def:hbar}
Let $U$ be a universal Turing machine, and $B$ a universal set of quantum
gates. For a pure state $\ket{\psi}$ over $n$ qubits, we define its $\Hbar$-complexity as
\begin{equation*}
\Hbar^{U,B}(\ket{\psi}) = - \log \braket{\psi | \udm_n | \psi},
\end{equation*}
where $\udm_n$ is the universal semi-density matrix defined with respect to $U$ and $B$.
\end{definition}

Since this variant of the notion is new, we include proofs of its invariance
with respect to $U$ and $B$ and its equivalence with the notion introduced
by G\'{a}cs in \Cref{sec:robust-qk-gateset}. Given its invariance we will
fix a choice of $U$ and $B$ and write $\Hbar(\ket{\psi})$, dropping the
superscripts. 
Furthermore, whenever $n$ is clear from context we will omit it and
simply write $\udm$.
While the notion described here is equivalent to the notion
defined by G\'{a}cs, our notion is more natural in a setting like ours
where we are interested in quantum algorithms.

We also introduce a robust version of G\'{a}cs' complexity $\Hbar$.

\begin{definition}
\label{def:hbar-eps}
  For any $\eps \in [0,1]$, we define
  \begin{equation*}
    \Hbar^\eps(\ket{\psi}) = \max_{\ket{\phi}: D(\ket{\psi}, \ket{\phi}) \le \eps}
    \Hbar(\ket{\phi}).
  \end{equation*}
\end{definition}

\begin{remark}
  The definition of $\Hbar^\eps$ takes the maximum of $\Hbar$ in the
  $\eps$-neighborhood of $\ket{\psi}$, analogous to $\eps$-smoothed
  min-entropies, since the purified distance becomes the trace distance for pure
  states.
  For a state to have small $\Hbar^{\eps}$, all nearby states must have small
  $\Hbar$.
  We note that only taking the maximum is meaningful here, as $\Hbar$ is
  always small when taking the minimum in the following sense: any state is
  negligibly close to a state with $\Hbar$ less than $O(\log^2 n)$.
  Specifically, for any $n$-qubit state $\ket{\psi}$, it is
  $2^{-\log^{2} n}$-close in purified distance to a state of the form
  $\ket{\psi'} = a e^{\ii \theta} \ket{0^{n}} + \sqrt{1 - a^{2}} \ket{\phi}$,
  where $\braket{0^n | \phi} = 0$ and $a \ge 2^{-\log^{2} n}$.
  According to the definition of $\udm$, we have
  $\frac{1}{cn} \ket{0^n} \bra{0^n} \leq \udm$ for some constant $c$; thus,
  \begin{equation*}
    \braket{\psi'|\udm|\psi'} \geq \frac{1}{cn} \abs{\braket{\psi'|0^n}}^2
    \geq 2^{-2 \log^2 n - \log(cn)},
  \end{equation*}
  which implies $\Hbar(\ket{\psi'}) \leq 2\log^2 n + \log(cn)$.
\end{remark}

G\'{a}cs' also considered a dual state complexity measure
$\Hol(\ket{\psi}) = -\braket{\psi | \log \udm | \psi}$ which we do not use.
What we will instead use is a new complexity measure $\Umin$, which is overlooked by G\'{a}cs'
work. As we will see, this notion is closely related to the measures $\Hbar$ and $\Hol$.
Moreover, we find that this new measure is a much better dual of $\Hbar$ as the their properties
demonstrate.

Consider the relative min-entropy, which is defined as
\begin{align*}
  D_\infty(\rho \| \sigma) = \min\{\lambda \mid \rho \leq 2^\lambda \sigma\}.
\end{align*}
The new measure $\Umin(\ket{\psi})$ is the relative min-entropy of $\ket{\psi}$
with respect to the universal density matrix $\udm$:
\begin{equation*}
  \Umin(\ket{\psi}) = D_\infty(\rho \| \udm).
\end{equation*}
We can define the smoothed version of it as:
\begin{equation*}
  \Umin^{\eps}(\ket{\psi}) = \min_{D(\ket{\psi}, \ket{\phi}) \le \eps} \Umin(\ket{\phi}).
\end{equation*}
We prove some simple properties of $\Umin$.

\begin{lemma}
  $\Umin (\ket{\psi}) = \log \braket{\psi | \udm^{-1} | \psi}$.
\end{lemma}

\begin{proof}
  $\ket{\psi}\bra{\psi} \leq 2^{r} \udm$ is equivalent to
  $\udm^{-1/2}\ket{\psi}\bra{\psi}\udm^{-1/2} \leq 2^r I$.
  (note that $\udm$ is invertible as $\udm$ is a full-rank Hermitian
  matrix). Then, the inequality holds if and only if
  $\|\udm^{-1/2}\ket{\psi}\|^2 \leq 2^{r}$, which can be reformulated as
  $\braket{\psi|\udm^{-1}|\psi}\leq 2^{r}$, so we are done.
\end{proof}

\begin{lemma}\label{lem:U-bound}
  For any quantum pure state $\ket{\psi}$ satisfying
  $\udm \ge 2^{-\kappa} \ketbra{\psi}{\psi}$, we have $\Umin(\ket{\psi}) \le \kappa$.
\end{lemma}

\begin{proof}
  Write $\udm = 2^{-\kappa} \ketbra{\psi}{\psi} + \udmr$.
  We can assume without loss of generality that $\udmr$ is strictly positive;
  otherwise, we can consider a small perturbation of it.
  Using the Sherman-Morrison formula
  \begin{equation*}
    {(A + uv^{\dagger})}^{-1} = A^{-1} -
    \frac{A^{-1} uv^{\dagger} A^{-1}}{1 + v^{\dagger} A^{-1} u}
  \end{equation*}
  with $A = \udmr$, $u = v = 2^{-\kappa/2} \ket{\psi}$, we obtain
  \begin{equation*}
    \udm^{-1} = \udmr^{-1} - \frac{2^{-\kappa} \udmr^{-1}
      \ketbra{\psi}{\psi} \udmr^{-1}}{1 + 2^{-\kappa}
      \braket{\psi | \udmr^{-1} | \psi}}.
  \end{equation*}
  Define $w = \braket{\psi | \udmr^{-1} | \psi}$ and take the expectation value
  of $\ket{\psi}$ on both sides:
  \begin{equation*}
    \begin{split}
      \braket{\psi | \udm^{-1} | \psi}
      &= w - \frac{2^{-\kappa} w^{2}}{1 + 2^{-\kappa} w}\\
      &= \frac{w}{1 + 2^{-\kappa} w}\\
      &\le 2^{\kappa}.
    \end{split}
  \end{equation*}
  This completes the proof.
\end{proof}

\begin{lemma}\label{lem:U-upper-bound}
  For any pure quantum state $\ket{\psi}$, we have
  \begin{equation*}
    \Umin(\ket{\psi}) \le \Knet(\ket{\psi}).
  \end{equation*}
\end{lemma}

\begin{proof}
  By the definition of $\udm$, we have
  \begin{equation*}
    \udm = 2^{-\Knet(\ket{\psi})} \ketbra{\psi}{\psi} + \udmr
  \end{equation*}
  for some positive semi-definite $\udmr$.
  \Cref{lem:U-bound} then completes the proof.
\end{proof}

\begin{lemma}\label{lem:U-lower-bound}
  For all pure quantum states $\ket{\psi}$ and
  $\eps \in \interval[open left]{0}{1}$, the following two bounds hold:
  \begin{align*}
    \Umin^{1-\eps}(\ket{\psi}) & \ge \Hbar(\ket{\psi}) + \log\eps,\\
    \Hbar^{1-\eps}(\ket{\psi}) & \le \Umin(\ket{\psi}) - \log\eps.
  \end{align*}
\end{lemma}

\begin{proof}
  By definition,
  $\Umin^{1-\eps}(\ket{\psi}) = \min_{\ket{\phi} : D(\ket{\psi},
    \ket{\phi}) \le 1 - \eps} \Umin(\ket{\phi})$.
  From the condition of the minimization, we have
  \begin{equation*}
    D(\ket{\psi}, \ket{\phi}) = \sqrt{1 - \abs{\braket{\psi | \phi}}^{2}} \le 1 - \eps,
  \end{equation*}
  and
  \begin{equation*}
    \abs{\braket{\psi | \phi}}^{2} \ge 1 - {(1-\eps)}^{2} = 2\eps - \eps^{2} \ge \eps.
  \end{equation*}
  Using the Cauchy-Schwarz inequality, we have
  \begin{equation*}
    \braket{\psi | \udm | \psi} \braket{\phi | \udm^{-1} | \phi}
    \ge \abs{\braket{\psi | \phi}}^{2} \ge \eps.
  \end{equation*}
  Taking the logarithm on both sides completes the proof.
  The other inequality follows from a similar reasoning.
\end{proof}

\begin{definition}[The $\GapH$ problem]
    \label{def:gaph}
    Let $r, \Delta, n \in \natural$. Let $\eps \in [0,1]$. We define~$\GapH^\eps(r, r+~\Delta)$ as the following (promise) problem: given a \emph{single} copy of a state $\ket{\psi}$ on some number $n$ of qubits, decide whether 
\begin{itemize}
    \item $\Hbar^{1-\eps}(\ket{\psi}) \leq r$, or
    \item $\Hbar(\ket{\psi}) \geq r+\Delta$.
\end{itemize}
\end{definition}

\begin{definition}[Hardness of $\GapH$ over a ``promise'' family]\label{def:strong-hardness-gaph}
  Let $r, \Delta, n \in \natural$ be functions of $\lambda$. Let $\eps \in [0,1]$.
  We say that $\GapH^\eps(r, r+\Delta)$ is hard over a family of states
  ${\{\ket{\psi_k}: k \in \bit^{n(\lambda)}\}}_{\lambda \in \natural}$ if the following hold: 
  \begin{itemize}
  \item (\emph{promise}) There exists a negligible function $\negl$ such that, for all $\lambda \in \natural$,\\ $\Pr_k[\Hbar^{1 - \eps} (\ket{\psi_{k}})\le r] \geq \frac12 - \negl(\lambda)$ and $\Pr_k[\Hbar(\ket{\psi_{k}}) \ge r + \Delta] \geq \frac12 - \negl(\lambda)$\,.
  \item (\emph{hardness of distinguishing}) For any QPT adversary
  $\adv$, there exists a negligible function $\negl'$ such that, for all
  $\lambda \in \natural$,
  \begin{equation*}
    \left| \Pr_{k} [\adv(1^{\lambda}, \ket{\psi_{k}}) = 0|  C_{\high}] -
    \Pr_{k} [\adv(1^{\lambda}, \ket{\psi_{k}}) = 0 | C_{\low}] \right| \leq \negl(\lambda) \,,
  \end{equation*}
  where $C_{\high}$ and $C_{\low}$ are events standing for
  $\Hbar(\ket{\psi_{k}}) \ge r + \Delta$ and
  $\Hbar^{1 - \eps} (\ket{\psi_{k}})\le r$, respectively.
  \end{itemize} 
If the hardness is against \emph{non-uniform} quantum polynomial-time adversaries, then we say that  $\GapH^\eps(r, r+\Delta)$ is non-uniformly hard over the family of states.
\end{definition}

\begin{definition}[The $\GapU$ problem]
    \label{def:gapu}
    Let $r, \Delta, n \in \natural$. Let $\eps \in [0,1]$. We define $\GapU^\eps(r, r+~\Delta)$ as the following (promise) problem: given a \emph{single} copy of a state $\ket{\psi}$ on some number $n$ of qubits, decide whether
\begin{itemize}
    \item $\Umin(\ket{\psi}) \leq r$, or
    \item $\Umin^{1-\eps}(\ket{\psi}) \geq r+\Delta$.
\end{itemize}
\end{definition}

We then define the hardness of $\GapU$ similarly as we did for $\GapH$ in~\cref{def:strong-hardness-gaph}.

\section{Entropic \texorpdfstring{$\EFI$}{EFI} and pseudo-mixed states}

In this section, we introduce two variants of the $\EFI$ primitive called
entropic $\EFI$ and (non-uniform) pseudo-mixed states.
As the main results of this section, we show that $\EFI$ implies both of these
variants.

\subsection{Entropic \texorpdfstring{$\EFI$}{EFI} from \texorpdfstring{$\EFI$}{EFI}}

Entropic $\EFI$ uses the entropy difference as a measure of distance between the
state pair, rather than the trace distance.
The more formal definition is given in \cref{def:EEFI}.
\begin{definition}[Entropic $\EFI$]\label{def:EEFI}
  We call two families of mixed states $\{\sigma_{0, \lambda}\}_{\lambda}$,
  $\{\sigma_{1, \lambda}\}_\lambda$ an entropic $\EFI$ pair, if the following
  condition holds:
  \begin{description}
    \item[Efficient Generation:] There exists a QPT algorithm $G$ that takes
          input $(1^\lambda, b)$ for security parameter $\lambda$ and
          $b \in \bit$, and outputs the mixed state $\sigma_{b, \lambda}$.
    \item[Entropy Gap:]
          $S(\sigma_{1, \lambda}) > S(\sigma_{0, \lambda}) + 1/\poly(\lambda)$.
    \item[Computational Indistinguishability:] For any QPT adversary algorithm
          $\adv$, we have that
          \begin{equation*}
            \abs{\Pr[\adv(1^{\lambda}, \sigma_{0, \lambda})=1]
              - \Pr[\adv(1^{\lambda}, \sigma_{1, \lambda})=1]}
            \le \negl(\lambda).
          \end{equation*}
    \end{description}
\end{definition}

We remark that by Fannes' inequality, every entropic $\EFI$ is automatically an
$\EFI$, as for any two states $\sigma_{0, \lambda}$ and $\sigma_{1, \lambda}$
satisfying $S(\sigma_{1, \lambda}) > S(\sigma_{0, \lambda}) + 1/\poly(\lambda)$,
we have $D(\sigma_{1, \lambda}, \sigma_{0, \lambda}) \ge 1/\poly(\lambda)$.
An $\EFI$ is not necessarily an entropic $\EFI$, as there are states with large
trace distance but no entropy difference.
However, we can show that the existence of an $\EFI$ implies that of an entropic
$\EFI$ by slightly modifying the state generation procedure.

\begin{theorem}\label{thm:EEFI}
  The existence of $\EFI$ implies the existence of entropic $\EFI$.
\end{theorem}

\begin{proof}

  Let $\{\rho_{0, \lambda}\}_\lambda$ and $\{\rho_{1, \lambda}\}_\lambda$ be an
  $\EFI$ pair that can be generated by a QPT algorithm $G^*$.
  We consider the quantum states
  \begin{align*}
    \sigma_{0, \lambda} & = \frac{1}{2} \ketbra{0}{0} \otimes \rho_{0, \lambda} +
                          \frac{1}{2} \ketbra{1}{1} \otimes \rho_{1, \lambda},\\
    \sigma_{1, \lambda} & = \frac{I}{2} \otimes \frac{\rho_{0, \lambda} + \rho_{1, \lambda}}{2}.
  \end{align*}
  We claim that $\{\sigma_{0, \lambda}\}_{\lambda}$ and
  $\{\sigma_{1, \lambda}\}_{\lambda}$ form an entropic $\EFI$ family.

  \paragraph{Efficient Generation} It is not difficult to see that the following
  QPT algorithm $G$ in \cref{alg:EEFI} outputs the mixed state
  $\sigma_{b, \lambda}$ on input $(1^\lambda, b)$.

  \begin{algorithm}[htbp!]
    \caption{Construction of the entropic $\EFI$ state generation algorithm
      $G$}\label{alg:EEFI}
    \begin{algorithmic}[1]
      \Require Inputs $1^\lambda$ and a state generation algorithm $G^{*}$
      for $\EFI$
      \State If $b = 0$, initialize the registers $A$ and $B$ as the mixed state
      $\frac{\ketbra{00}{00}_{AB}+\ketbra{11}{11}_{AB}}{2}$.
      \State If $b = 1$, initialize the registers $A$ and $B$ as the mixed state
      $\frac{I_{A}}{2} \otimes \frac{I_{B}}{2}$.
      \State Run $G^*(1^\lambda, \cdot)$ on the register $B$ and store the
      obtained state in register $C$.
      \State Output registers $A$ and $C$.
    \end{algorithmic}
  \end{algorithm}

  \paragraph{Entropy Gap}
  We can directly calculate the entropy of $\sigma_{0, \lambda}$, since it can
  be block diagonalized:
  \begin{equation*}
    S(\sigma_{0, \lambda}) = \frac{S(\rho_{0, \lambda}) +
      S(\rho_{1, \lambda})}{2} + 1.
  \end{equation*}
  Note that by~\cref{def:EFI},
  $\TD(\rho_{0,\lambda}, \rho_{1,\lambda}) \geq 1 - \negl(\lambda)$.
  Thus by~\cref{lem:TDtoAO,lem:AOadd}, we have that
  \begin{equation*}
    \begin{split}
      S(\sigma_{1,\lambda})
      & = 1+S\left(\frac{\rho_{0,\lambda} + \rho_{1,\lambda}}{2}\right)\\
      & \geq 2+ \frac{S(\rho_{0,\lambda}) +
        S(\rho_{1,\lambda})}{2} - \negl(\lambda).
    \end{split}
  \end{equation*}
  This implies an $1-\negl(\lambda)$ entropy gap between $\sigma_{0,\lambda}$
  and $\sigma_{1,\lambda}$.

  \paragraph{Computational Indistinguishability}
  By~\cref{def:EFI}, $\rho_{0,\lambda}$ and $\rho_{1,\lambda}$ are
  computationally indistinguishable.
  A standard hybrid argument shows that
  $\sigma_{0, \lambda} \approx_c \frac{I}{2} \otimes \rho_{0,\lambda} \approx_c \sigma_{1,\lambda}$,
  which concludes the proof of \Cref{thm:EEFI}.

\end{proof}

\subsection{Non-uniform pseudo-mixed states from entropic \texorpdfstring{$\EFI$}{EFI}}\label{sec:PMS}

Next, we study a variant of $\EFI$ pair called pseudo-mixed states.
Informally, a pseudo-mixed state is an efficiently preparable state $\rho$ that,
together with the maximally mixed state, forms an $\EFI$.
That is, $\rho$ is far from $I/2^n$, yet no QPT algorithm can distinguish them
with non-negligible advantage.
We require a non-uniform version of $\PMS$ for which the generation algorithm
uses a classical advice string.
The formal definition of (non-uniform) $\PMS$ is given in \cref{def:PMS}.

\begin{algorithm}[htbp!]
  \caption{Construction of non-uniform $\PMS$ from entropic $\EFI$ }\label{alg:pms}
  \begin{algorithmic}[1]
    \Require Inputs $1^\lambda$, a classical advice $a(\lambda) \in [\lambda^2n^2(\lambda)p^2(\lambda)]$, and a state generation algorithm $G$ for the entropic $\EFI$.
    \State Let $\rho_{0,\lambda}$, $\rho_{1,\lambda}$ be the $n(\lambda)$-qubit entropic $\EFI$ family pair with entropy gap $1/p(\lambda)$ for polynomial $p$, namely $S(\rho_{1,\lambda}) - S(\rho_{0,\lambda}) \geq 1/p(\lambda)$.
    \State Let $m(\lambda) = \lambda^2n^2(\lambda)p^2(\lambda)$, and $\rho_{0,\lambda}' = \rho_{0, \lambda}^{\otimes m}, \rho_{1, \lambda}' = \rho_{1, \lambda}^{\otimes m}$.
    \State Let $\eps(\lambda) = 2^{-\lambda} k(\lambda) = a(\lambda) - \lambda^2n^2(\lambda)p(\lambda)/2$ and $\ell(\lambda) = (n(\lambda)m(\lambda)+k(\lambda))/2 - \log(1/\eps)$.
    \State Let $\{C_j\}_{j\in L}$ be the Clifford family over $n(\lambda)m(\lambda)$ qubits. By~\cref{thm:extractor}, it's a $(k, \eps, \eps/12)$ quantum strong extractor $\Ext_\ell^{A \to A_1}$ that extracts $\ell$ qubits. Let $A = A_1A_2$ where subsystem $A_1$ consists of the first $\ell$ qubits of system $A$ and subsystem $A_2$ consists of the last $n(\lambda)m(\lambda)-\ell(\lambda)$ qubits.
    \State Output state
  \begin{equation*}
    \tau_{0,\lambda}  = \Ext^{A \to A_{1}}_{\ell} (\rho'_{0,\lambda}),\quad
    \tau_{1, \lambda} = \Ext^{A \to A_{1}}_{\ell} (\rho'_{1,\lambda}).
  \end{equation*}
  \end{algorithmic}
\end{algorithm}

\begin{definition}[Pseudo-mixed States]\label{def:PMS}
  A family of mixed states $\{\rho_{\lambda}\}_\lambda$ of $n(\lambda)$ qubits
  is called a pseudo-mixed state family if the following conditions hold:
  \begin{description}
    \item[Efficient Generation:] There exists a QPT algorithm $G$ that, on input
          $1^\lambda$ for integer $\lambda$, outputs the mixed state
          $\rho_{\lambda}$.
    \item[Entropy Gap:] $S(\rho_{\lambda}) < n(\lambda) - 1/\poly(\lambda)$.
    \item[Computational Indistinguishability:] For any QPT adversary $\adv$,
          \begin{equation*}
            \abs{\Pr \bigl[\adv(1^{\lambda}, \rho_{\lambda}) = 1 \bigr] -
            \Pr \bigl[\adv(1^{\lambda}, I / 2^{n(\lambda)}) = 1 \bigr]}
            \leq \negl(\lambda).
          \end{equation*}
  \end{description}
\end{definition}

\begin{definition}[Non-uniform pseudo-mixed states]\label{def:nu-PMS}
  A family of mixed states $\{\rho_{\lambda, a}\}_\lambda$ of $n(\lambda)$
  qubits is a non-uniform pseudo-mixed state family if it satisfies the
  condition of \cref{def:PMS} with the only change that the state generation
  algorithm $G$ takes an additional advice string $a$ as input.
  The length of the string $a$ is called the advice size of the pseudo-mixed
  state.
\end{definition}

Our main theorem concerning pseudo-mixed states is the following:
\begin{theorem}\label{thm:EFI-PMS}
  The existence of $\EFI$ implies the existence of non-uniform pseudo-mixed states
  with advice size $O(\log \lambda)$.
\end{theorem}

As established by \cref{thm:EEFI}, $\EFI$ implies entropic $\EFI$, so it
suffices to prove \cref{thm:EEFI-PMS}.
\begin{theorem}\label{thm:EEFI-PMS}
  The construction in~\cref{alg:pms} is a secure non-uniform $\PMS$ with advice size $O(\log \lambda)$.
\end{theorem}

\begin{proof}[Proof of \Cref{thm:EEFI-PMS}]

  We will show that the construction in~\cref{alg:pms} is a secure $\PMS$ in case that $a(\lambda)$ is a $1/m(\lambda)$ of $S(\rho_{1,\lambda})$, i.e., $(a(\lambda)-1)/m(\lambda) < S(\rho_{1,\lambda}) \leq a(\lambda)/m(\lambda)$

  We consider a large number of copies of the state both to amplify the entropy
  gap and to better approximate the min-entropy with von Neumann entropy. According to
  \cref{cor:min-entropy-explicit-bound}, we have that
  \begin{equation*}
    S(\rho'_{0,\lambda}) = m(\lambda)S(\rho_{0,\lambda}) \leq S(\rho'_{1,\lambda})
    - \lambda^2 n^2(\lambda) p(\lambda).
  \end{equation*}

  From~\cref{cor:min-entropy-explicit-bound}, we have that
  \begin{equation*}
    \begin{split}
      \Hmin^{\eps/12}(\rho_{1,\lambda}')
      & \geq S(\rho'_{1,\lambda}) - O\left(n(\lambda)\sqrt{\log(1/\eps) m(\lambda)}\right) \\
      & = S(\rho'_{1,\lambda}) - O\left(\lambda^{1.5}n^2(\lambda)p(\lambda)\right).
    \end{split}
  \end{equation*}

  We proceed to show that $\tau_{0,\lambda}$ forms a pseudo-mixed state of
  $\log\abs{L}+\ell$ qubits.

  \paragraph{Efficient Generation} Since we can prepare $m(\lambda)$ copies of
  $\rho_{0,\lambda}$ and apply the unitary $U_r$ in polynomial time,
  $\tau_{0,\lambda}$ can be prepared efficiently.

  \paragraph{Entropy Gap}
  Define states $\rho^{(j)} = U_j \rho'_{0, \lambda} U_j^\dagger$.
  By the subadditivity of von Neumann entropy, we have
  \begin{equation*}
    \begin{split}
      S(A_{1})_{\rho^{(j)}}
      & \le S(A)_{\rho^{(j)}} + S(A_{2})_{\rho^{(j)}}\\
      & \le S(\rho'_{0, \lambda}) + \log \abs{A_{2}}\\
      & \le S(\rho'_{1, \lambda}) - \lambda^{2} n^{2}(\lambda) p(\lambda) +
        (n'(\lambda) - \ell(\lambda))\\
      & \le a(\lambda)  - \lambda^{2} n^{2}(\lambda) p(\lambda) +
        (n'(\lambda) - \ell(\lambda)).
    \end{split}
  \end{equation*}
  Together with the definition of $\ell$, this proves that for all $j$
  \begin{equation*}
    S(A_{1})_{\rho^{(j)}} \le \ell(\lambda) - \Omega\bigl(\lambda^{2} n^{2}(\lambda)
    p(\lambda)\bigr).
  \end{equation*}

  We can compute the entropy of $\tau_{0,\lambda}$ as
  \begin{equation*}
    \begin{split}
      S(\tau_{0,\lambda})
      & = \log\abs{L} + \frac{1}{\abs{L}}\sum_{j \in L}
        S\left(\Tr_{A_2} \left[ U_j \rho'_{0, \lambda} U_j^\dagger \right] \right)\\
      & \le \log\abs{L} + \frac{1}{\abs{L}} \sum_{j} S(A_{1})_{\rho^{(j)}}\\
      & \le \log\abs{L} + \ell(\lambda) - \Omega\bigl(\lambda^{2} n^{2}(\lambda)
        p(\lambda)\bigr),
    \end{split}
  \end{equation*}
  which has a non-negligible gap with the entropy of the
  $(\log\abs{L}+\ell)$-qubit maximally mixed state.

  \paragraph{Computational Indistinguishability}
  Notice that the smoothed min entropy of $\rho_{1,\lambda}'$ satisfies
  \begin{equation*}
    \Hmin^{\eps/12}(\rho_{1,\lambda}') \geq S(\rho_{1,\lambda}') -
    O\left(\lambda^{1.5}n^2(\lambda)p(\lambda)\right) \geq k(\lambda).
  \end{equation*}

  By the definition of $(k,\eps,\eps/12)$-strong extractor, we have
  $D\left(\tau_{1, \lambda}, \frac{I_L}{\abs{L}} \otimes \frac{I_{A_1}}{\abs{A_{1}}}\right)\leq \eps$,
  where $\log\abs{A_{1}}=\ell$.
  By our parameter choice $\eps = 2^{-\lambda}$, it follows that
  $\tau_{1, \lambda}$ and
  $\frac{I_L}{\abs{L}} \otimes \frac{I_{A_1}}{\abs{A_{1}}}$ are statistically
  indistinguishable.

  By~\cref{def:EFI}, $\rho_{0,\lambda}$ and $\rho_{1,\lambda}$ are
  computationally indistinguishable.
  A standard hybrid argument shows that
  $\rho_{0,\lambda}'\approx_c\rho_{1,\lambda}'$ and thus
  $\tau_{0,\lambda}\approx_c\tau_{1, \lambda}$.
  Therefore,
  $\tau_{0,\lambda} \approx_c \frac{I_L}{\abs{L}} \otimes \frac{I_{A_1}}{\abs{A_{1}}}$,
  yielding our pseudo-mixed state.

  Our construction relies on knowing an estimate $a(\lambda)/m(\lambda)$ of the
  von Neumann entropy of our entropic $\EFI$ state $\rho_{1, \lambda}$.
  To address this, we introduce non-uniformity and take $a(\lambda)$ as advice,
  which can be represented by a bit string of length $O(\log \lambda)$.
\end{proof}

\begin{remark}
  Actually, we have constructed a special type of imbalanced $\EFI$ defined
  in~\cite{KT24}.
  We can show that if
  $\log \abs{A_{1}} \leq \frac{n+\Hmin^{\eps/12}(\rho_{1,\lambda}')}{2}-\lambda$,
  the state $\tau_{0, \lambda}$ should be computationally indistinguishable from
  $\frac{I_L}{\abs{L}} \otimes \frac{I_{A_1}}{\abs{A_{1}}}$; while if
  $\log \abs{A_{1}} \geq \frac{n+S(\rho_{0,\lambda}')}{2} + 1/\poly(\lambda)$,
  we can show that
  $S(\tau_\lambda)\leq \log\abs{L}+\log\abs{A_{1}}-1/\poly(\lambda)$.
\end{remark}

\section{Single-copy pseudorandom states from pseudo-mixed states}\label{sec:1prs}

In this section, we show how to construct $\OnePRS$ from a pseudo-mixed state.
If the pseudo-mixed state is non-uniform, then so is the resulting $\OnePRS$.

Assume $\{\rho_\lambda\}_{\lambda}$ is a family of pseudo-mixed states on system $A$.
By applying the tensoring method to amplify the gap if necessary, we can assume
without loss of generality that $\rho_\lambda$ is an $n(\lambda)$-qubit state
with entropy $S(\rho_\lambda) < n(\lambda) - 1$.

To construct single-copy pseudorandom states that are pure, a natural approach
is to consider the purification $\ket{\Psi_\lambda}_{AB}$ of the state
$\rho_\lambda$, such that
$\Tr_B (\ketbra{\Psi_\lambda}{\Psi_\lambda}) = \rho_\lambda$.
Assume without loss of generality that system $B$ consists of
$n'(\lambda) \ge n(\lambda)$ qubits.
Then the state $\ket{\Psi_\lambda}_{AB}$ is computationally indistinguishable
from the maximally mixed state on system $A$, but there is no guarantee
regarding system $B$.

To make the system $B$ also indistinguishable from the maximally mixed state, we
apply the quantum extractor to the system $B = B_1B_2$, treating the system $A$
as the environment.
We use the Clifford group, a unitary 2-design, as the quantum extractor, and use
a quantum one-time pad on the subsystem $B_2$ to effectively trace out $B_2$
when the keys are sampled uniformly at random:
\begin{equation*}
  \ket{\phi_{k}} = \Bigl(I_A \otimes \bigl(\bigl(I_{B_1} \otimes
  (X^{\alpha}Z^{\beta})_{B_2} \bigr) C_B \bigr)\ket{\Psi_\lambda}_{AB} \Bigr)
  \otimes \ket{C}_L
\end{equation*}
where the key $k = C \parallel \alpha \parallel \beta$, $C$ ranges over all Clifford
gates on $B$, $B_1$ consists of $\ell$ qubits, and $\alpha, \beta \in \bit^{n'-\ell}$ are
the quantum one-time pad keys on the $n'-\ell$ qubits of $B_2$ for some $\ell$
to be chosen later.

The goal is to show that our extractor makes the subsystem $B$ indistinguishable from the maximally mixed state while using a short key.
However, notice that the quantum strong extractor works on states with low min-entropy, but a pseudo-mixed state only has a low von Neumann entropy, rather than a low min-entropy.
Thus, we must consider a sufficiently large number $m$ of copies of
$\rho_\lambda$ so that the min-entropy and von Neumann entropy are close
asymptomatically.
The complete construction is provided in \cref{alg:one-prs}.

\begin{algorithm}[htbp!]
  \caption{Construction of $\OnePRS$ from pseudo-mixed states}\label{alg:one-prs}
  \begin{algorithmic}[1]
    \Require Inputs $1^\lambda$, $k \in \{0, 1\}^{r(\lambda)}$.
    \State Let $G$ be the generation algorithm for a pseudo-mixed states family $\{\rho_{\lambda}\}_{\lambda}$ of entropy gap at least 1.
    \State Obtain from $G$ a circuit $V_{\lambda}$ that prepares a purification $\ket{\Psi_{\lambda}}_{AB}$ of $\rho_{\lambda}$. Let the number of qubits of $A$ be $n$ and the number of qubits of $B$ be $n'$. Without loss of generality, we assume that $n' \geq n$. %
    \State Pick $m = 50(n')^2\lambda$, $\ell = (n'-n) m/2 + 1$.
    \State Parse $k$ as $C \parallel \alpha \parallel \beta$ where $C$ is a
    Clifford gate over $B^m$, and $\alpha, \beta \in \bit^{n'm-\ell}$. %
    \State Let $B^m=B_1B_2$ where $B_1$ consists of the first $\ell$ qubits of $B^m$. %
    \State Output state
    \begin{equation*}
      \ket{\phi_{k}} = \Bigl(I_{A^m} \otimes \bigl(\bigl(I_{B_1} \otimes
      (X^\alpha Z^\beta)_{B_2} \bigr) C_{B^m} \bigr) {\bigl( V_{\lambda}\ket{0^{n+n'}}_{AB}
        \bigr)}^{\otimes m} \Bigr) \otimes \ket{C}_L.
    \end{equation*}
  \end{algorithmic}
\end{algorithm}

\begin{theorem}\label{thm:pms-1prs-uniform}
  Assuming that pseudo-mixed states family exists, then $\OnePRS$ exists.
\end{theorem}

\begin{proof}
  We prove that the procedure in \cref{alg:one-prs} constructs
  $\OnePRS$.

  We first establish that the construction in \cref{alg:one-prs} possesses a
  non-trivial stretch property.
  Let $L$ denote the set of Clifford gates over $B^m$ of $mn'$ qubits.
  The key $k$ has length $r = 2(n'm - \ell) + \log \abs{L}$.
  The state $\ket{\phi_k}$ consists of $(n + n')m + \log \abs{L}$ qubits.
  Therefore, the stretch is
  \begin{equation*}
    (n - n')m + 2\ell = 2.
  \end{equation*}

  Next, we show that $\E_k \left[\ketbra{\phi_k}{\phi_k}\right]$ is computationally
  indistinguishable from the maximally mixed state on the system $A^mB^mL$.
  By the properties of the quantum one-time pad, we have
  \begin{equation*}
    \E_k \left[\ketbra{\phi_k}{\phi_k}\right] = \E_C\, \left[\ketbra{C}{C}_L \otimes \Tr_{B_2}
    \Bigl( C_{B^m} {\ketbra{\Psi_{\lambda}}{\Psi_{\lambda}}}^{\otimes m}
    C^\dagger_{B^m} \Bigr)\right] \otimes \frac{I_{B_2}}{\abs{B_2}}.
  \end{equation*}
  It therefore suffices to prove that
  \begin{equation}\label{eq:1prs-1}
    \E_C\, \left[\ketbra{C}{C}_L \otimes \Tr_{B_2} \Bigl( C_{B^m}
    {\ketbra{\Psi_{\lambda}}{\Psi_{\lambda}}}^{\otimes m} C^\dagger_{B^m} \Bigr)\right]
  \end{equation}
  is computationally indistinguishable from the maximally mixed state on
  $LA^mB_1$.

  Define $\sigma = {\ketbra{\Psi_\lambda}{\Psi_\lambda}}^{\otimes m}$.
  \Cref{eq:1prs-1} can be written as $\Ext^{B^m}_\ell \bigl(\sigma \bigr)$.
  
  Recall that $\Tr_{B^m}(\sigma) = \left(\rho_{\lambda}^{\otimes m}\right)_{A^m}$, which, by standard hybrid argument, is computationally indistinguishable with the maximally mixed state on $A^m$. By \cref{thm:extractor}, we have that
  \begin{equation*}
   \Ext^{B^m}_\ell \bigl(\sigma \bigr) \approx_{s} \frac{I_L}{\abs{L}} \otimes \frac{I_{B_1}}{\abs{B_1}} \otimes (\rho_{\lambda}^{\otimes m})_{A^m} \approx_c \frac{I_L}{\abs{L}} \otimes \frac{I_{B_1}}{\abs{B_1}} \otimes \frac{I_{A^m}}{\abs{A^m}},
  \end{equation*}
  as long as there exists $\eps = \negl(\lambda)$ such that 
  \begin{equation}\label{eq:1prs-2}
      \ell \leq \frac{mn' + \Hmin^{\eps/12}{(B^m|A^m)}_{\sigma}}{2} - \log(1/\eps).
  \end{equation}

  It remains to prove inequality \ref{eq:1prs-2}.

  Let $\eps = 2^{-\lambda}$. We apply \cref{cor:min-entropy-explicit-bound} to bound the smoothed conditional min entropy of the state $\sigma$. Since $m \geq 5\log \frac{1}{\eps}$, we have that
  \begin{equation*}
  \begin{split}
    \Hmin^{\eps/12}{(B^m|A^m)}_{\sigma}
    \ge&\; m S{(B|A)}_{\ket{\Psi_{\lambda}}} - 6n'\sqrt{(\lambda + 4)m}\\
    =& - m S{(\rho_{\lambda})}- 6n'\sqrt{(\lambda + 4)m}\\
    \ge&\; m(1 - n) - 6n'\sqrt{(\lambda + 4)m}.
  \end{split}
  \end{equation*}

  A direct calculation shows that
  \begin{equation*}
    \begin{split}
      & \frac{mn' + \Hmin^{\eps/12}{(B^m|A^m)}_{\sigma}}{2} - \log(1/\eps)\\
      \ge \;& \frac{m(n' - n)}{2} + \frac{m - 6n'\sqrt{(\lambda + 4)m} - 2\lambda}{2}\\
      \ge \;& \ell
    \end{split}
  \end{equation*}
  for large $\lambda$, which concludes the proof.
\end{proof}

\begin{theorem}\label{thm:pms-1prs-no-uniform}
  Assuming that non-uniform pseudo-mixed states family with advice size $s(\lambda)$ exists, then non-uniform $\OnePRS$ with advice size $s(\lambda)$ exists.
\end{theorem}
\begin{proof}
In the non-uniform setting, we provide \cref{alg:one-prs} with the same advice as that of the non-uniform pseudo-mixed states family. Using this advice, the algorithm can generate the corresponding circuits for producing the pseudo-mixed states. The rest of the proof then follows without modification.
\end{proof}

\begin{theorem}\label{thm:efi-1prs}
  Assuming that $\EFI$ exist, then non-uniform $\OnePRS$ with advice size
  $O(\log \lambda)$ exists.
\end{theorem}

\begin{proof}
Assuming that $\EFI$ exist, by \cref{thm:EFI-PMS}, there exist non-uniform pseudo-mixed states
  with advice size $O(\log \lambda)$. Then by \cref{thm:pms-1prs-no-uniform}, there exists non-uniform $\OnePRS$ with advice size $O(\log \lambda)$.
\end{proof}

\begin{corollary}\label{cor:1prs-stretch}
  Assuming the existence of $\EFI$, there exists non-uniform $\OnePRS$ of $n(\lambda)$ qubits with advice size $O(\log \lambda)$ such that the stretch is at least $\sqrt{n(\lambda)}$.
\end{corollary}

\begin{proof}
  Assuming the existence of $\EFI$, by \Cref{thm:efi-1prs}, there exists
  non-uniform $\OnePRS$ $\{\ket{\phi_{k}}\}$ with advice size $O(\log \lambda)$.
  Denote the number of qubits as $n'(\lambda)$.

  We can construct another non-uniform $\OnePRS$ with the same advice size by
  parallel repetition:
  $\ket{\Phi_K} = \ket{\phi_{k_1}}\ket{\phi_{k_2}}\cdots \ket{\phi_{k_{n'(\lambda)}}}$,
  where $K = k_1 \parallel k_2 \parallel \cdots \parallel k_{n'(\lambda)}$.
  Let $n(\lambda) = n'^{2}(\lambda)$.
  It is clear that $\{\ket{\Phi_K}\}$ is a $n$-qubit state that can be generated
  with an advice of size $O(\log \lambda)$ and that it is a non-uniform
  $\OnePRS$ of $n(\lambda)$ qubits whose stretch is at least
  $\sqrt{n(\lambda)}$, which concludes the proof.
\end{proof}

\section{A natural universal \texorpdfstring{$\EFI$}{EFI}}

In cryptography, it is sometimes possible to have universal constructions of a
cryptographic primitive, meaning that it is secure as long as such primitives
exist.
For example, a universal one-way function is known to exist as an early result
in the field~\cite{Levin87}.
Such universal constructions are of theoretical importance and often reveal the
essential reason and understanding of the corresponding primitive.
In this section, we prove the existence of a weak form of universal $\EFI$ as an
interesting corollary of the construction of non-uniform $\OnePRS$ from $\EFI$
proved in the previous section.

\begin{theorem}
  The following are equivalent:
  \begin{enumerate}
    \item $\EFI$ exist.
    \item There exist (efficiently computable) functions $T(\lambda) = \poly(\lambda)$, $n(\lambda) = \poly(\lambda)$,
          $r(\lambda) = n - \omega(\log n)$, $\eps = \negl(\lambda)$,
          $\delta = 1 - \Omega(\frac{1}{\lambda})$ such that the pair
          $\left(\{\rho_{r}^T\}, \{\frac{1}{2^n}I\}\right)$ is an
          $(\eps, \delta)$-weak $\EFI$, where
          $\rho_{r}^T = \frac{1}{2^r}\sum_{\abs{P} \leq r} \ket{\psi_P^T}\bra{\psi_P^T}$,
          and $\ket{\psi_P^T}$ is the $n$-qubit state output by program $P$ in
          time at most $T$ (if the number of output qubits exceeds $n$, abort;
          if the output state $\ket{\psi_P}$ has $k < n$ qubits, replace it with
          $\ket{\psi_P}\ket{0^{n-k}}$).\footnote{Some care needs to be taken to
          ensure that this output is an $n$-qubit pure state.
          If a program, outputs a state that is longer than $n$ qubits we can
          output some canonical state such as $\ket{0}^{\otimes n}$.
         If a program outputs a state that may be mixed, there are ways to
          check if the output is far from pure, for instance by doing two runs of
          the program, and performing a swap test on their outputs.
          If a program output a state that is noticeably mixed, we can
          again output a canonical state.
          Finally, programs that output fewer than $n$ qubits can be padded with
          $\ket{0}$'s.
          }

  \end{enumerate}
\end{theorem}

\begin{proof}
    Suppose there exists $r$, $T$, $\eps$, and $\delta$ as in 2 such that $\rho_r^T$ and
    $\frac{1}{2^n}I$ is an $(\eps, \delta)$-weak $\EFI$ pair, then, by Theorem~\ref{thm:WEFI-EFI}, an $\EFI$ pair also exists.

    Assume that $\EFI$ exist, then, by Theorem~\ref{thm:efi-1prs}, there exists a non-uniform $\OnePRS$ with
    $\log(n)$-size advice, say $m$-to-$n$ $\Gen_{\lambda}$ where $\lambda
    \in [n]$ is the advice. 
    For any $k \in \{0,1\}^m$ and $\lambda \in
    [n]$, $\Gen_\lambda\ket{k}$ can be generated by a Turing machine with
    size $m+\log n + C$, where $C$ is a constant. Thus for $r=m+\log n+C$,
    and $T$ be the running time of $\Gen$, we can view $\rho_{r}^T$ as
    sampling a state from the non-uniform $\OnePRS$ with probability
    $\Omega(\frac 1n)$, and sampling from some other state with the
    remaining probability. Thus, any adversary can distinguish $\rho_r^T$ from $\frac{1}{2^n}I$ with advantage at most $1-\Omega(\frac 1n) + \negl$, and they are $(1-2^{r-n})$-statistically far (since the entropy of $\rho_r^T$ is bounded by $r$, while the entropy of $\frac{1}{2^n}I$ is $n$). Thus, the pair $\left(\{\rho_{r}^T\}, \{\frac{1}{2^n}I\}\right)$ is an $(\eps, \delta)$-weak $\EFI$ pair, for $\eps = \negl(n)$, and $\delta = 1 - \Omega(\frac{1}{n})$.
\end{proof}

\section{Equivalence with \texorpdfstring{$\GapH$}{GapH} hardness}
\label{sec:gaph}

In this section, we characterize $\EFI$ with the hardness of estimating the robust G\'{a}cs' complexity. The proof goes in three steps. In~\cref{sec:quantum-alg-info}, we show the entropy gap for the mixture of high and low complexity states. In~\cref{sec:extraction}, we show how to extract randomness from the mixture of high complexity states. And then in~\cref{sec:pieces-together}, we put pieces together and propose an algorithm for estimating the non-uniform $\GapH$ problem.

\subsection{Quantum algorithmic information}\label{sec:quantum-alg-info}

\begin{lemma}[Mixtures over low-complexity states are approximately
  low-entropy]\label{lem:H-low-entropy}
  For any family of states $\{\ket{\psi_k}\}$ such that for all $k$,
  $\Hbar^{1-\eps}(\ket{\psi_k}) \leq r$, the mixed state
  $\E_k \ketbra{\psi_k}{\psi_k}$ is $\sqrt{2\eps}$-close to a state with von
  Neumann entropy at most $r$.
\end{lemma}

\begin{proof}
  Let $\udm = \sum_{i} \mu_{i} \ketbra{\phi_{i}}{\phi_{i}}$ be the spectral
  decomposition of the universal density matrix $\udm$.
  Define $\Pi_{\low}$ to be the projection onto the span of low complexity
  eigenstates $\ket{\phi_{i}}$ of $\udm$ with eigenvalue $\mu_{i}$ at least
  $2^{-r}$, and $\Pi_{\high}$ to be the projection onto the span of high
  complexity eigenstates, i.e.\ those with eigenvalue less than $2^{-r}$.
  Since $\udm$ has trace at most $1$, the number of $\mu_{i}$'s at least
  $2^{-r}$ is at most $2^{r}$, and therefore the low-complexity space defined by
  $\Pi_{\low}$ has dimension at most $2^{r}$.

  For any state $\ket{\psi}$, define
  $\ket{\psi_{\high}} = \Pi_{\high} \ket{\psi} / \norm{\Pi_{\high} \ket{\psi}}$
  and
  $\ket{\psi_{\low}} = \Pi_{\low} \ket{\psi} / \norm{\Pi_{\low} \ket{\psi}}$.

  We first show that for any state $\ket{\psi}$ satisfying
  $\Hbar^{1-\eps}(\ket{\psi}) \le r$, the trace distance between $\ket{\psi}$
  and $\ket{\psi_{\low}}$ is at most $\sqrt{2\eps}$.
  To establish this, observe that
  \begin{equation*}
    D(\ket{\psi}, \ket{\psi_{\high}}) =
    \sqrt{1 - \abs{\braket{\psi | \psi_{\high}}}^{2}} = \norm{\Pi_{\low} \ket{\psi}}.
  \end{equation*}
  If $\norm{\Pi_{\low} \ket{\psi}} \le 1-\eps$, then by the definition of
  $\Hbar^{1-\eps}(\ket{\psi})$, we would have
  \[
    \Hbar^{1-\eps}(\ket{\psi}) \ge \Hbar(\ket{\psi_{\high}}) > r,
  \]
  since $\ket{\psi_{\high}}$ is $(1-\eps)$-close to $\ket{\psi}$ and has
  $\Hbar$-complexity greater than $r$.
  This contradicts the assumption that $\Hbar^{1-\eps}(\ket{\psi}) \le r$.

  Therefore, it must be that $\norm{\Pi_{\low} \ket{\psi}} > 1 - \eps$.
  Consequently,
  \begin{equation*}
    D(\ket{\psi}, \ket{\psi_{\low}}) = \sqrt{1 - \norm{\Pi_{\low}\ket{\psi}}^{2}}
    < \sqrt{1 - {(1 - \eps)}^{2}} < \sqrt{2\eps}.
  \end{equation*}

  Next, we apply the result from the first step to $\ket{\psi_{k}}$, since it
  holds that $\Hbar^{1-\eps}(\ket{\psi_{k}}) \le r$.
  This implies that, for all $k$,
  $D(\ket{\psi_{k}}, \ket{\psi_{k, \low}}) < \sqrt{2 \eps}$.
  By \cref{lem:distance-mix}, $\E_{k} \ketbra{\psi_{k}}{\psi_{k}}$ is
  $\sqrt{2\eps}$-close to $\E_{k} \ketbra{\psi_{k, \low}}{\psi_{k, \low}}$.
  Since $\E_{k} \ketbra{\psi_{k, \low}}{\psi_{k, \low}}$ is supported on the
  low-complexity subspace $\Pi_{\low}$, a linear subspace of dimension at most
  $2^{r}$, its von Neumann entropy is at most $r$.
\end{proof}

\begin{lemma}\label{lem:min-entropy-eigenstate}
    Let $\rho = \sum \alpha_i \ketbra{\varphi_i}{\varphi_i}$ be a density matrix with its eigenvalue decomposition. Let $\tilde \Pi_\low$ be the projector of subspace spanned by $\ket{\varphi_i}$ with $\alpha_i\geq 2^{-s}$. If $\Tr (\tilde \Pi_\low \rho) \leq \eps$, then we have $\Hmin^{2\eps}(\rho) \geq s+\log(1-\eps)$.
\end{lemma}
\begin{proof}
    Let $p \coloneq \Tr (\tilde \Pi_\low \rho) \leq \eps$, and $\tilde \rho \coloneq (I-\tilde \Pi_\low)\rho(I-\tilde \Pi_\low), \hat \rho \coloneqq \frac{\tilde \rho}{1-p}$, then we have
    \begin{align*}
        \|\hat \rho \| \leq \frac {2^{-s}}{1-p} \leq \frac{2^{-s}}{1-\eps},
    \end{align*}
    so
    \begin{align*}
        \Hmin(\hat \rho) \geq s + \log (1-\eps)
    \end{align*}
    Also
    \begin{align*}
        \|\rho - \hat\rho\| \leq \| \rho - \tilde \rho \| + \| \tilde \rho - \hat \rho \| = p + p = 2p \leq 2\eps,
    \end{align*}
    so $\hat \rho$ lies within trace distance $2\eps$ of $\rho$, so we have $\Hmin^{2\eps}(\rho) \geq s + \log(1-\eps)$
\end{proof}

\begin{lemma}[Mixtures over high complexity states are approximately high min-entropy]\label{lem:H-high-entropy}
  For any (not necessarily efficiently) samplable state family $\ket{\psi_k}$
  such that $\forall k: \Hbar^0(\ket{\psi_k}) \geq s$, there exists a constant
  $C$ such that for any $\Gamma > \log n + C$, the mixed state
  $\rho = \E_k \ketbra{\psi_k}{\psi_k}$ is $2^{-\Gamma+\log n + C}$-close to a
  state with min-entropy at least $s-\Gamma$.
\end{lemma}

\begin{proof}
  Let $\rho = \sum_{i} \beta_{i} \ketbra{\varphi_{i}}{\varphi_{i}}$ be the
  spectral decomposition of $\rho$.
  We prove the lemma by showing that removing the large eigenvalues does not
  significantly change the state.
  Let $\tilde{\Pi}_\low$ be the projection onto the subspace spanned by
  $\ket{\varphi_{i}}$ with $\beta_{i} \ge 2^{\Gamma - s -1}$.
  Each of these eigenvectors of $\rho$ has a short program of length at most
  $s - \Gamma + \log n + c$ for some constant $c$.
  That is, for all such $\ket{\varphi_{i}}$,
  $\Knet(\varphi_{i}) \le s - \Gamma + \log n + c$.
  Thus, we can write $\udm  = \sum_{i} 2^{-\Knet(\ket{\varphi_{i}})}
             \ketbra{\varphi_{i}}{\varphi_{i}} + \udm'$, where $\udm'$ is a semi-density matrix representing the remaining part of $\udm$, and the sum is only over the eigenvectors of $\rho$ corresponding to $\tilde\Pi_{\low}$.  Then, we have
  \begin{equation}\label{eq:high-entropy-1}
    \begin{split}
      \udm & = \sum_{i} 2^{-\Knet(\ket{\varphi_{i}})}
             \ketbra{\varphi_{i}}{\varphi_{i}} + \udm'\\
           & \ge \sum_{i} 2^{-\Knet(\ket{\varphi_{i}})}
             \ketbra{\varphi_{i}}{\varphi_{i}}\\
           & \ge 2^{-(s - \Gamma + \log n + c)}\sum_{i}
             \ketbra{\varphi_{i}}{\varphi_{i}}\\
           & = 2^{-(s - \Gamma + \log n + c)}\, \tilde\Pi_{\low}\,.
    \end{split}
  \end{equation}
  Now, for any state $\ket{\psi}$ with $\Hbar^0(\ket{\psi}) \ge s$, we have
  $\braket{\psi | \udm | \psi} \le 2^{-s}$.
  So, it follows from \cref{eq:high-entropy-1} that
  \begin{equation*}
    \braket{\psi | \tilde \Pi_{\low} | \psi} \le 2^{s - \Gamma + \log n + c}
    \braket{\psi | \udm | \psi} \le 2^{- \Gamma + \log n + c},
  \end{equation*}
  and, as a result,
  \begin{equation}
    \label{eq:high-entropy-2}
    \Tr (\tilde \Pi_\low \rho) = \E_k \braket{\psi_k | \tilde \Pi_\low | \psi_k}
    \le 2^{-\Gamma + \log n+c}.
  \end{equation}
  According to~\cref{lem:min-entropy-eigenstate}, set $\eps = 2^{-\Gamma+\log n + c}$, we have
  \begin{align*}
      \Hmin^{2\eps}(\rho) \geq s - \Gamma + 1 - \log (1-\eps) \geq s-\Gamma
  \end{align*}
  As long as $\eps \leq 1/2$ (this is the case in case that $\Gamma > \log n + C$). Choose $C = c + 1$ to be the constant in the lemma statement, then we have $\Hmin^{2^{-\Gamma+\log n +C}}(\rho) \geq s-\Gamma$.
\end{proof}

\begin{lemma}\label{lem:knet0-Heps}
  For any state $\ket{\psi}$ and $\eps \in \interval[open right]{0}{1}$, we have
  $H^{1-\eps}(\ket{\psi}) \leq \Knet(\ket{\psi}) + \log\frac{1}{\eps}$.
\end{lemma}

\begin{proof}
    By the definition of $\udm$ we know that for all $\ket{\psi}$, 
    \begin{equation*}
        \udm = 2^{-\Knet(\ket{\psi})} \ketbra{\psi}{\psi} + \udm' \geq 2^{-\Knet(\ket{\psi})}\ket{\psi}\bra{\psi},
    \end{equation*}
    where $\udm'$ is the residual part of the universal density matrix.
    
    Recall that $\Hbar^0$ is defined as $\Hbar^0(\ket{\psi'}) = \braket{\psi' | \udm |\psi'}$. So in case that $|\braket{\psi|\psi'}|^2 \geq \eps$ then we have 
    \begin{equation*}
    \braket{\psi' | \udm |\psi'} \geq \eps \braket{\psi | \udm |\psi} \geq \eps 2^{\Knet(\ket{\psi})},
    \end{equation*}
    and $\Hbar^0(\ket{\psi'}) \leq -\log(\eps 2^{\Knet(\ket{\psi})}) = \Knet(\ket{\psi}) + \log 1/\eps$. %
\end{proof}

\subsection{Extraction}\label{sec:extraction}

\begin{lemma}\label{lem:extractor-channel}
  Let $\rho$ be a mixed state over an $n$-qubit system $A$. Let $A_{1}$ be a
  subsystem of $A$ with $\ell$ qubits, and let $A_2$ be the remaining $n-l$ qubits.
  Let $\Ext^{A \to A_{1}}_{\ell}$ be an extractor on $A$ mapping states on $A$
  to $A_{1}$ by applying a unitary $2$-design $\{U_{j}\}_{j\in L}$.
  Define
  $\rho' = \Ext^{A \to A_{1}}_{\ell}(\rho) \otimes \frac{I_{A_{2}}}{\abs{A_{2}}}$.
  Then for any $\Delta$, the following holds:
  \begin{itemize}
    \item If $S(\rho) \leq 2\ell-n-\Delta$, then
          $S(\rho') \le n + \log\abs{L} -\Delta$.
    \item If $\Hmin(\rho) \ge 2\ell-n+\Delta$, then $\rho'$ is
          $2^{-\Delta/2}$-close to $\frac{I_{A}}{\abs{A}} \otimes \frac{I_{L}}{\abs{L}}$.
  \end{itemize}
\end{lemma}

\begin{proof}
  For the first statement, we write
  \begin{equation*}
    \begin{split}
      \rho' & = \Ext^{A \to A_{1}}_{\ell}(\rho) \otimes
              \frac{I_{A_{2}}}{\abs{A_{2}}}\\
            & = \frac{1}{L}\sum_{j\in L} \ketbra{j}{j}_{L} \otimes \Tr_{A_{2}}
              \bigl( U_{j} \rho U_{j}^{\dagger} \bigr) \otimes
              \frac{I_{A_{2}}}{\abs{A_{2}}}.
    \end{split}
  \end{equation*}
  Define $\rho^{(j)}_{A} =  U_{j} \rho U_{j}^{\dagger}$.
  Noticing that $\rho'$ is a cq-state, we have
  \begin{equation}\label{eq:extractor-channel-1}
    S(\rho') \le \log \abs{L} + n - \ell +
    \max_j S \bigl(\rho^{(j)}_{A_{1}}\bigr),
  \end{equation}
  where $\rho^{(j)}_{A_{1}} = \Tr_{A_{2}} \rho^{(j)}_{A} $ is the reduced
  density matrix of $\rho^{(j)}_{A}$ on $A_{1}$.
  Using subadditivity of the von Neumann entropy,
  $S(A_{1}) \le S(A_{2}) + S(A)$, we have for all $j$
  \begin{equation*}
    \begin{split}
      S \bigl( \rho^{(j)}_{A_{1}} \bigr)
      & \le S \bigl( \rho^{(j)}_{A_{2}} \bigr) + S \bigl( \rho^{(j)}_{A} \bigr)\\
      & \le (n - \ell) + (2\ell - n - \Delta)\\
      & = \ell - \Delta.
    \end{split}
  \end{equation*}
  Together with the bound in \cref{eq:extractor-channel-1}, this proves the
  first statement.

  The second statement is a direct application of \cref{thm:extractor} by taking
  $\eps = 2^{-\Delta/2}$.
  Since $\Hmin^{\eps/12}(\rho) \ge \Hmin(\rho) \ge 2\ell - n + \Delta$, the theorem guarantees that
  the number of qubits that one can extract is at least
  \begin{equation*}
    \frac{n + (2\ell - n + \Delta)}{2} - \log(1/\eps) = \ell.
  \end{equation*}
  Thus, by \cref{thm:extractor}, we have
  \begin{equation*}
    \begin{split}
      D \biggl(\rho', \frac{I_{A}}{\abs{A}} \otimes \frac{I_{L}}{\abs{L}} \biggr)
      & = D \biggl(\Ext_{\ell}^{A}(\rho) \otimes \frac{I_{A_{2}}}{\abs{A_{2}}},
        \frac{I_{A}}{\abs{A}} \otimes \frac{I_{L}}{\abs{L}} \biggr) \\
      & \le \eps = 2^{-\Delta/2}.
    \end{split}
  \end{equation*}
\end{proof}

\subsection{Putting the pieces together}\label{sec:pieces-together}

\begin{algorithm}[htbp!]
  \caption{Algorithm for solving $\GapH[r,r+\Delta]$}\label{alg:estimate-gaph}
  \begin{algorithmic}[1]
    \Require A single-copy input state $\ket{\psi_k} \in \H_{A}$.
    \State Let $\rho = \E_k \ket{\psi_k}\bra{\psi_k}$, and let $\ell = (n+r+\Delta/2)/2$.
    \State Let $\Pi$ be the projector that distinguishes $\Ext^{A \to A_1}_{\ell}(\rho)$
    from the maximally mixed state, where $A_1$ is the first $\ell$ qubits of system $A$.
    \State Test the projector $\Pi$ on $\Ext^{A \to A_1}_{\ell}(\ket{\psi_k}\bra{\psi_k})$.
    \State If the test passes report $\low$, otherwise report $\high$.
  \end{algorithmic}
\end{algorithm}

\begin{theorem}\label{thm:efi-nu-gaph}
  The following two statements are equivalent:
  \begin{itemize}
    \item $\EFI$ exist.
    \item There exists a non-uniform family of efficiently samplable states with advice of size $O(\log \lambda)$
          $\{\ket{\psi_{k}}\}$, efficiently computable functions
          $r \in [n(\lambda)], \Delta = \omega(\log\lambda)$, and a universal
          constant $\eps < \frac{1}{100}$, such that $\GapH^\eps(r, r+\Delta)$
          is non-uniformly hard on average over $\{\ket{\psi_{k}}\}$.
    \end{itemize}
\end{theorem}

\begin{proof}

  First, assuming $\EFI$ exist, we prove that the $\GapH$ problem for some
  non-uniform state family is hard on average.
  By the existence of $\EFI$ and \cref{cor:1prs-stretch}, there exists a non-uniform $n(\lambda)$-qubit $\OnePRS$ family $\{\ket{\phi_{k'}}\}$ where the advice $a$ has size $O(\log \lambda)$ and the stretch is at least $\sqrt{n(\lambda)}$ (i.e.\ the seed is of length at most $n(\lambda)- \sqrt{n(\lambda)}$)

  Define
  $k = b \parallel k' \parallel j$ with $j \in \bit^{n}$ and $b \in \bit$.
  Consider the non-uniform state family $\ket{\psi_{k}}$ defined as
  \begin{equation}\label{eq:efi-gaph-states}
    \ket{\psi_{k}} =
    \begin{cases*}
      \ket{\phi_{k'}} & \text{ if } $b = 0$,\\
      \ket{j} & \text{ otherwise.}
    \end{cases*}
  \end{equation}
  We will show that this state family is a hard instance of $\GapH$.

  When $b = 0$, the state $\ket{\psi_{k}}$ is $\ket{\phi_{k'}}$, and thus it can be
  described by a program of size at most
  $n(\lambda) - \sqrt{n(\lambda)} + O(\log \lambda) +C$ for some constant $C$, i.e.\ $\Knet^0(\ket{\psi_{k}}) \leq  n(\lambda) - \sqrt{n(\lambda)} + O(\log \lambda) +C$.
  Hence, by \cref{lem:knet0-Heps}, the state has low $\Hbar^{1-\epsilon}$ complexity:
  $\Hbar^{1-\eps}(\ket{\psi_{k}}) \le \Knet^{0} (\ket{\psi_{k}}) +
  \log (1/\eps)\le n(\lambda) - \sqrt{n(\lambda)} + O(\log \lambda) +
  C + \log(1/\eps)$.
  On the other hand, when $b = 1$, the state $\ket{\psi_{k}}$ is $\ket{j}$
  for a uniformly random $j \in \bit^{n}$.
  Since $\sum_{j} \braket{j | \udm |j} \le \Tr(\udm) \le 1$, we have for all $\delta > 0$
  \begin{equation*}
    \Pr_{j \in \bit^{n}} \bigl[ \braket{j | \udm | j}
    \ge 2^{-n(\lambda)+ \delta} \bigr] \le 2^{-\delta}.
  \end{equation*}
  Taking, for example, $\delta = \sqrt{n(\lambda)} / 2$,
  the state has high $\Hbar$ complexity with high probability:
  \begin{equation}\label{eq:efi-nu-gaph-1}
    \Pr_{j \in \bit^{n}} \bigl[\Hbar(\ket{j}) > n(\lambda) - \delta \bigr]
    \ge 1 - 2^{-\delta} = 1 - \negl(\lambda).
  \end{equation}

  Choose
  $r = n(\lambda) - \sqrt{n(\lambda)} + O(\log \lambda) + C + \log(1/\eps)$ and
  $r + \Delta = n(\lambda) - \delta$.
  We have
  $\Delta = \sqrt{n(\lambda)} / 2 - O(\log \lambda) - C - \log(1/\eps)
  = \omega(\log \lambda)$.
  We claim that if there is an adversary $\adv$ that solves the $\GapH$ for this
  non-uniform family, the same adversary breaks the $\OnePRS$
  $\{\ket{\phi_{k'}}\}$.

  Define events $C_{\low}$ and $C_{\high}$
  $\Hbar^{1 - \eps}(\ket{\psi_{k}}) \le r$ and
  $\Hbar(\ket{\psi_{k}}) \ge r + \Delta$ over the random key $k$,
  representing the complexity of the state being low and high respectively.
  Define events $A_{\low}$ and $A_{\high}$ for $\adv$ outputting $0$ and $1$ on
  input $\ket{\psi_{k ,a}}$.
  The discussion above indicates that $\Pr[C_{\low} | b = 0] = 1$ and
  $\Pr[C_{\high} | b = 0] > 1 - \negl(\lambda)$.
  Under these two condition, we can easily show that for all possible events
  $A$,
  \begin{align*}
    \abs{\Pr[A \land b = 0] - \Pr [A \land C_{\low}]} &\le \negl(\lambda),\\
    \abs{\Pr[A \land b = 1] - \Pr [A \land C_{\high}]} &\le \negl(\lambda).
  \end{align*}

  By the linearity of quantum operations,
  \begin{equation*}
    \begin{split}
      \Pr\bigl[ \adv \bigl(1^{\lambda}, \E\nolimits_{k'}
      \ketbra{\phi_{k'}}{\phi_{k'}}\bigr) = 0\bigr]
      & = \Pr_{k'}\bigl[ \adv \bigl(1^{\lambda},
        \ketbra{\phi_{k'}}{\phi_{k'}}\bigr) = 0\bigr]\\
      & = \Pr_{k: b=0}\bigl[ \adv \bigl(1^{\lambda},
        \ketbra{\phi_{k}}{\phi_{k}}\bigr) = 0\bigr]\\
      & = \Pr_{k} \bigl[A_{\low} | b = 0\bigr]\\
      & = 1 - \Pr_{k} \bigl[A_{\high} | b = 0\bigr]\\
      & = 1 - 2 \Pr_{k} \bigl[A_{\high} \land b = 0\bigr]\\
    \end{split}
  \end{equation*}
  On the other hand,
  \begin{equation*}
    \begin{split}
      \Pr\biggl[ \adv \Bigl(1^{\lambda}, \frac{I}{2^{n}}\Bigr) = 0\biggr]
      & = \Pr_{k: b=1}\bigl[ \adv \bigl(1^{\lambda},
        \ketbra{\phi_{k}}{\phi_{k}}\bigr) = 0\bigr]\\
      & = \Pr_{k} \bigl[A_{\low} | b = 1\bigr]\\
      & = 2 \Pr_{k} \bigl[A_{\low} \land b = 1\bigr]\\
    \end{split}
  \end{equation*}

  So when the failure probability of $\adv$ for the $\GapH$ problem is at most $
  \frac{1}{2} - \frac{1}{\lambda^{c}}$ for some constant $c$,
  the advantage of $\adv$ for the $\OnePRS$ is bounded by
  \begin{equation*}
    \begin{split}
      & \abs{1 - 2 \Bigl(\Pr_{k} \bigl[A_{\low} \land b = 1\bigr] +
        \Pr_{k} \bigl[A_{\high} \land b = 0\bigr]\Bigr)}\\
      \ge\; & \abs{1 - 2 \Bigl(\Pr_{k} \bigl[A_{\low} \land C_{\high}\bigr] +
        \Pr_{k} \bigl[A_{\high} \land C_{\low}\bigr]\Bigr)} - \negl(\lambda)\\
      \ge\; & \frac{2}{\lambda^{c}} - \negl(\lambda).
    \end{split}
  \end{equation*}
  That is, $\adv$ breaks $\OnePRS$ and therefore also $\EFI$.
  This completes the proof for the direction that $\EFI$ implies the average
  hardness of $\GapH$.

  Next, we establish the converse direction by demonstrating that if no $\EFI$
  exists, then non-uniform $\GapH$ can be solved efficiently.\
  Define the subsets $K_{\high}$ and $K_{\low}$ of keys $k$ as follows:
  \begin{align*}
    K_{\high} & = \{k : \Hbar (\ket{\psi_k}) \ge r+\Delta\},\\
    K_{\low} & = \{k : \Hbar^{1-\eps}(\ket{\psi_{k}}) \le r\}.
  \end{align*}
  Define the states
  \begin{align*}
    \rho_{\high} & = \E_{k \in K_{\high}} \ketbra{\psi_{k}}{\psi_{k}},\\
    \rho_{\low} & = \E_{k \in K_{\low}} \ketbra{\psi_{k}}{\psi_{k}},\\
    \rho_{\midd} & = \E_{k \not\in K_{\low} \cup K_{\high}} \ketbra{\psi_{k}}{\psi_{k}}.
  \end{align*}
  Then we can express
  \begin{equation*}
    \rho = p_{\low} \rho_{\low} + p_{\midd} \rho_{\midd} + p_{\high} \rho_{\high}\,,
  \end{equation*}
where $p_{\low}$, $p_{\midd}$, and $p_{\high}$ are the respective probabilities.
  We begin by examining a few simple cases.
  First, if
  \begin{equation*}
    \Pr_{k} [\Hbar(\ket{\psi_{k}}) \ge r + \Delta]
    \le \frac{1}{2} - \textnormal{non-negl}(\lambda)\,
  \end{equation*}
where $\textnormal{non-negl}(\lambda)$ is some non-negligible function, then one can always guess that the state has low
  complexity and obtain a non-negligible advantage.
  Thus, we assume without loss of generality that
  \begin{equation}\label{eq:efi-nu-gaph-2}
    p_{\high} = \Pr_{k} [\Hbar(\ket{\psi_{k}}) \ge r + \Delta]
    \ge \frac{1}{2} - \negl(\lambda).
  \end{equation}
  Similarly, we may also assume
  \begin{equation}\label{eq:efi-nu-gaph-3}
    p_{\low} = \Pr_{k} [\Hbar^{1-\eps}(\ket{\psi_k}) \le r]
    \ge \frac{1}{2} - \negl(\lambda).
  \end{equation}
  A direct consequence is that $p_{\midd} = \negl(\lambda)$.

  Now, we use \cref{alg:estimate-gaph} and prove that it will provide a
  non-negligible advantage in solving non-uniform $\GapH$ assuming $\EFI$'s do not exist. Note that all the parameters in~\cref{alg:estimate-gaph} is efficiently computable: $\ell = (n+r+\Delta/2)/2$ is efficiently computable as $r(\lambda)$ and $\Delta(\lambda)$ are efficiently computable functions. (\cref{alg:estimate-gaph} is a uniform $\GapH$ estimator, but it can be also used for a non-uniform $\GapH$ estimation by replacing the $\EFI$ distinguisher with the non-uniform distinguisher)
  Define $\ell = (n + r + \Delta/2)/2$ as in the algorithm.
  Applying the extractor $\Ext_{\ell}$ to $\rho$, we have
  \begin{equation*}
    \Ext_{\ell}(\rho) = p_{\low} \Ext_{\ell}(\rho_{\low}) +
    p_{\midd} \Ext_{\ell}(\rho_{\midd}) + p_{\high} \Ext_{\ell}(\rho_{\high}).
  \end{equation*}
  By \cref{lem:H-low-entropy}, we know that $\rho_{\low}$ is
  $\sqrt{2\eps}$-close to a state of von Neumann entropy at most
  $r = 2\ell - n - \Delta/2$.
  We denote this state by $\rho'_{\low}$.
  By the first part of \cref{lem:extractor-channel}, we have
  $S(\Ext_{\ell}(\rho_{\low}')) \leq n + \log \abs{L} - \Delta/2$.
  This implies that
  \begin{equation*}
    D \bigl(\Ext_{\ell}(\rho_{\low}'), \pi \bigr) \geq 1 - 2^{-\Delta/2}.
  \end{equation*}
  By the triangle inequality, it follows that
  \begin{equation*}
    D \bigl(\Ext_{\ell}(\rho_{\low}), \pi \bigr)
    \geq 1 - 2^{-\Delta/2} - \sqrt{2\eps},
  \end{equation*}
  where $\pi = \frac{I_{L}}{\abs{L}}\otimes \frac{I_{A_{1}}}{\abs{A_{1}}}$
  is the maximally mixed state.

  Taking $\Gamma = \Delta/4$ in \cref{lem:H-high-entropy}, we have that
  $\rho_{\high}$ is $2^{-\Delta/4 + \log n + C}$-close to a state with
  min-entropy at least $r + \Delta - \Gamma = r + 3\Delta/4$.
  Let this state be $\rho_{\high}'$.
  By \cref{lem:extractor-channel}, $\Ext_{\ell}(\rho_{\high}')$ is
  $2^{-\Delta/8}$-close to the maximally mixed state $\pi$.
  By the triangle inequality, we have
  \begin{equation}\label{eq:efi-nu-gaph-4}
    D \bigl(\Ext_{\ell}(\rho_{\high}), \pi \bigr)
    \le 2^{-\Delta/4+\log n+C} + 2^{-\Delta/8} = \negl(\lambda).
  \end{equation}

  So, we can conclude that
  \begin{equation*}
    D \bigl(\Ext_{\ell}(\rho), \pi \bigr) > 1/2-\sqrt{2\eps}-\negl(\lambda) > 1/4
  \end{equation*}
  when $\eps < 1/100$.
  As a result, the two states $\Ext_{\ell}(\rho)$ and $\pi$ are statistically
  far and efficiently samplable pairs.
  According to the non-existence of $\EFI$, there exists a projector $\Pi$ that
  can distinguish these two with non-negligible probability.
  In more detail,
  \begin{equation*}
    \Tr (\Pi\, \Ext_{\ell}(\rho)) -
    \Tr ( \Pi\, \pi ) = p(\lambda)
  \end{equation*}
  for some non-negligible probability $p(\lambda)$.
  From \cref{eq:efi-nu-gaph-4}, it follows that
  \begin{equation}\label{eq:efi-nu-gaph-5}
    \Tr (\Pi\,  \Ext_{\ell}(\rho)) -
    \Tr ( \Pi\, \Ext_{\ell}(\rho_{\high})) \ge p(\lambda) - \negl(\lambda).
  \end{equation}
  By \cref{eq:efi-nu-gaph-2,eq:efi-nu-gaph-3}, we have
  \begin{equation*}
    D \Bigl( \Ext_{\ell}(\rho), \frac{1}{2} \Ext_{\ell}(\rho_{\low}) +
    \frac{1}{2}\Ext_{\ell}(\rho_{\high}) \Bigr) \le \negl(\lambda).
  \end{equation*}
  Using this in \cref{eq:efi-nu-gaph-5}, we have
  \begin{equation*}
    \Tr (\Pi\,  \Ext_{\ell}(\rho_{\low})) -
    \Tr ( \Pi\, \Ext_{\ell}(\rho_{\high})) \ge 2 p(\lambda) - \negl(\lambda),
  \end{equation*}
  which is also non-negligible.
  Thus $\Pi$ provides non-negligible advantage in estimating $\GapH$ given only
  a single copy of $\ket{\psi_k}$.
\end{proof} 

In the previous theorem we build the equivalence of non-uniform hardness of
$\GapH$. The uniform harness of $\GapH$ in turn is a characterization of
$\OnePRS$.

\begin{corollary}\label{cor:gaph-1prs}
    The following two statements are equivalent:
    \begin{itemize}
        \item $\OnePRS$ exists.
        \item There exists a uniform family of efficiently samplable states
            $\{\ket{\psi_k}\}$, a universal constant $\eps < 1/100$, and
            $\Delta = \omega(\log \lambda)$ such that $\GapH$ is 
            hard over $\set{\ket{\psi_k}}$.
    \end{itemize}
\end{corollary}
\begin{proof}
    In case that $\OnePRS$ exist, then define the same state family as in~\cref{thm:efi-nu-gaph}. Then we can define a uniform state family on which $\GapH^\eps[r, r+\Delta]$ for some function $r$ and $\Delta = \omega (\log n)$.

    On the other hand, in case that $\OnePRS$ do not exist, then as the proof of~\cref{thm:efi-1prs} implicitly shows that pseudo-mixed states are equivalent to $\OnePRS$, we can conclude that pseudo-mixed states do not exist. Then we can apply~\cref{alg:estimate-gaph} (as pseudo-mixed states do not exist, we can distinguish mixed states from the maximally mixed state efficiently), which provides a non-negligible advantage for estimating $\GapH$.
\end{proof}

\section{Equivalence with \texorpdfstring{$\GapU$}{GapU} hardness}
\label{sec:gapu}

In this section, we obtain a characterization of $\EFI$ analogous to that of
\Cref{sec:gaph},
but using $\Umin$ instead of $\Hbar$. In particular, we build up equivalence of $\EFI$ with the hardness of $\GapU$ (\cref{def:gapu}). The proof goes in the same way as~\cref{sec:gaph}: first we show that the complexity gap implies the entropy gap. Then we extract the entropy from the state with quantum extractors and apply~\cref{alg:estimate-gaph} to estimate the complexity.

\begin{lemma}[Mixture over low-complexity states is approximately low-entropy
  for $\Umin$]\label{lem:U-low-entropy}
  For any family of states $\{\ket{\psi_k}\}$ such that for all $k$,
  $\Umin(\ket{\psi_k}) \leq r$, the mixed state
  $\E_k \ketbra{\psi_k}{\psi_k}$ is $2^{-\Gamma/2}$-close to a state with von
  Neumann entropy at most $r + \Gamma$.
\end{lemma}

\begin{proof}
  Consider the spectral decomposition
  $\udm = \sum_{i} \mu_{i} \ketbra{\phi_{i}}{\phi_{i}}$ of $\udm$.
  Let $\Pi_{\low}$ be the projection onto the subspace spanned by the
  low-complexity eigenvectors $\ket{\phi_{i}}$ with eigenvalues
  $\mu_{i} \ge 2^{-r-\Gamma}$ and let $\Pi_{\high} = I - \Pi_{\low}$.
  For any state $\ket{\psi}$ satisfying $\Umin(\ket{\psi}) \le r$, we write
  \begin{equation*}
    \ket{\psi} = \Pi_{\low} \ket{\psi} + \Pi_{\high} \ket{\psi}.
  \end{equation*}

  By the condition $\Umin(\ket{\psi}) \le r$, we have
  $\braket{\psi | \udm^{-1} | \psi} \le 2^{r}$.
  Using $\braket{\psi | \Pi_{\low} \udm^{-1} \Pi_{\high} | \psi} = 0$ and
  $\braket{\psi | \Pi_{\high} \udm^{-1} \Pi_{\low} | \psi} = 0$, we obtain
  \begin{equation*}
    \begin{split}
      2^{r + \Gamma} \braket{\psi | \Pi_{\high} | \psi}
      & \le \braket{\psi | \Pi_{\high} \udm^{-1} \Pi_{\high} | \psi}\\
      & \le \braket{\psi | \Pi_{\low} \udm^{-1} \Pi_{\low} | \psi} +
        \braket{\psi | \Pi_{\high} \udm^{-1} \Pi_{\high} | \psi}\\
      & = \braket{\psi | \udm^{-1} | \psi}\\
      & \le 2^{r}.
    \end{split}
  \end{equation*}
  Thus, for all $k$,
  $\braket{\psi_{k} | \Pi_{\high} | \psi_{k}} \le 2^{-\Gamma}$.
  Define
  $\ket{\psi_{k, \low}} = \Pi_{\low} \ket{\psi_{k}}/\norm{\Pi_{\low} \ket{\psi_{k}}}$,
  we have
  \begin{equation*}
    D(\ket{\psi_{k}}, \ket{\psi_{k, \low}})
    \le \norm{\Pi_{\high}\ket{\psi_{k}}} \le 2^{-\Gamma/2}.
  \end{equation*}
  It follows from \cref{lem:distance-mix} that
  \begin{equation*}
    D \Bigl(\E_{k} \ketbra{\psi_{k}}{\psi_{k}},
    \E_{k} \ketbra{\psi_{k,\low}}{\psi_{k, \low}}\Bigr) \le 2^{-\Gamma/2},
  \end{equation*}
  which completes the proof by noting that
  $\E_{k} \ketbra{\psi_{k,\low}}{\psi_{k, \low}}$ is supported on $\Pi_{\low}$
  of dimension $2^{r + \Gamma}$ and has entropy at most $r+\Gamma$.
\end{proof}

\begin{lemma}[Mixture over high complexity states is approximately high
  min-entropy for $\Umin$]\label{lem:U-high-entropy}
  For any (not necessarily efficiently) samplable state family $\ket{\psi_k}$
  such that $\forall k: \Umin^{1-\eps}(\ket{\psi_k}) > s$, then there is
  constant $C$ such that the mixed state $\rho = \E_k \ketbra{\psi_k}{\psi_k}$
  is $2\eps$-close to a state with min-entropy at least $s - \log n - C$.
\end{lemma}

\begin{proof}
  Consider the spectral decomposition of
  $\rho = \sum_{i} \beta_{i} \ketbra{\varphi_{i}}{\varphi_{i}}$.
  Let $c$ be the length of the program describing the sampling algorithm for
  $\ket{\psi_{k}}$.
  Let $\tilde\Pi_{\low}$ be the projection onto the span of eigenvectors
  $\ket{\varphi_{i}}$ with corresponding eigenvalue
  $\beta_{i} \ge 2^{-(s - \log n - c)}$.
  All such vectors can be indexed by programs of length at most
  $(s - \log n - c) + \log n + c = s$ and therefore have
  $\Knet(\ket{\varphi_{i}}) \le s$.
  Thus, we can write
  $\udm = \sum_{i} 2^{-\Knet(\ket{\varphi_{i}})} \ketbra{\varphi_{i}}{\varphi_{i}} + \udm'$
  for some positive semi-definite $\udm'$.
  From this, we have $\udm \ge 2^{-s} \tilde\Pi_{\low}$ and, by \cref{lem:U-bound},
  for any state $\ket{\psi}$ in the space that $\Pi_{\good}$ projects onto,
  $\Umin(\ket{\psi}) \le s$.

  For any state $\ket{\psi}$, define
  $\ket{\psi_{\low}} = \tilde\Pi_{\low}\ket{\psi} / \norm{\tilde\Pi_{\low} \ket{\psi}}$.
  When $\norm{\tilde\Pi_{\low} \ket{\psi}} \le 1 - \eps$, we have
  $D(\ket{\psi}, \ket{\psi_{\low}}) = \norm{\tilde\Pi_{\low} \ket{\psi}} \le 1 - \eps$.
  So we have $\Umin^{1-\eps}(\ket{\psi}) \le s$ as $\ket{\psi_{\low}}$ is in
  the space $\tilde\Pi_{\low}$ projects onto.
  Thus for any $\Umin^{1-\eps}(\ket{\psi}) > s$, we have
  $\norm{\tilde\Pi_{\low} \ket{\psi}} > 1 - \eps$ and consequently
  \begin{equation*}
    \braket{\psi |\tilde \Pi_{\low} | \psi} < 1 - (1-\eps)^{2} \le 2\eps.
  \end{equation*}
  Using the condition $\Umin^{1-\eps}(\ket{\psi_{k}}) > s$ for all $k$, we have
  $\Tr(\Pi_{\good} \E_{k} \ketbra{\psi_{k}}{\psi_{k}}) < 2\eps$.

  Define projection $\tilde\Pi_\high = I - \tilde\Pi_\low$ and state
  $\rho' = \frac{\tilde\Pi_\high \rho \tilde\Pi_\high}{\Tr(\tilde\Pi_\high \rho \tilde\Pi_\high)}$.
  We have $D(\rho, \rho') = \Tr(\tilde\Pi_{\low} \rho) < 2\eps$.
  Any eigenvalue of $\rho'$ is bounded above by
  $2^{- s + \log n + c}/(1 - 2^{- s + \log n + c}) \le 2^{-s + \log n + c + 1}$.
  Taking $C = c+1$ completes the proof.
\end{proof}

We are now ready to prove that the average-case hardness of estimating $\Umin$
is equivalent to the existence of $\EFI$.

\begin{theorem}\label{thm:efi-un-gapu}
  The following two statements are equivalent:
  \begin{itemize}
    \item $\EFI$ exist.
    \item There exists a non-uniform family of states $\{\ket{\psi_{k}}\}$ with advice size $O(\log \lambda)$
          such that $\GapU[r,r+\Delta]$ is hard on average.
  \end{itemize}
\end{theorem}

\begin{proof}
  The proof is similar to that of \cref{thm:efi-nu-gaph} and we only outline the
  differences.

  We first prove that the existence of $\EFI$ implies the hardness of $\GapU$
  for a non-uniform family of states.
  We use the same construction in \cref{eq:efi-gaph-states}.
  Now define events $C_{\low}$ and $C_{\high}$ to be
  $\Umin(\ket{\psi_{k}}) \le r$ and
  $\Umin^{1-\eps}(\ket{\psi_{k}}) \ge r + \Delta$ over the random keys.
  Define events $A_{\low}$ and $A_{\high}$ be the events that the algorithm
  $\adv$ for $\GapU$ outputs $0$ and $1$ respectively.
  As in the proof of \cref{thm:efi-nu-gaph}, we prove that
  $\Pr[C_{\low} | b = 0 ] = 1$.
  That is, when the state is sampled from the $\OnePRS$ family
  $\{\ket{\phi_{k'}}\}$, we have by \cref{lem:U-upper-bound},
  $\Umin(\ket{\psi_{k}}) \le \Knet(\ket{\psi_{k}})
  \le n(\lambda) - \sqrt{n(\lambda)} + O(\log \lambda) + C$.
  Similarly, we can prove $\Pr[C_{\high} | b = 1] = 1 - \negl(\lambda)$.
  When $b = 1$, the state is a random computational basis state $\ket{j}$ of $n$
  qubits, and by \cref{eq:efi-nu-gaph-1} and \cref{lem:U-lower-bound}, we have
  \begin{equation*}
    \Pr_{j \in \bit^{n}} \bigl[ \Umin^{1-\eps}(\ket{j}) \ge n(\lambda) - \delta + \log\eps \bigr]
    \ge \Pr_{j \in \bit^{n}} \bigl[ \Hbar(\ket{j}) \ge n(\lambda) - \delta \bigr] \ge 1 - 2^{-\delta}.
  \end{equation*}
  Choosing $r = n(\lambda) - \sqrt{n(\lambda)} + O(\log \lambda) + C$ and
  $r + \Delta = n(\lambda) - \delta + \log \eps$, a similar argument as in the
  rest of the proof of \cref{thm:efi-nu-gaph} shows that $\adv$ also breaks
  $\EFI$.

  Next, assume that $\EFI$ do not exist.
  We can prove, for any state family $\{\ket{\psi_{k}}\}$, there exists an
  algorithm that solves $\GapU(r, r+\Delta)$ efficiently.
  The proof is essentially the same as that for $\Umin^{1-\eps}$ and $\Umin^{0}$
  in place of $\Hbar^{0}$ and $\Hbar^{1-\eps}$, respectively, and using
  \cref{lem:U-high-entropy,lem:U-low-entropy} in place of
  \cref{lem:H-high-entropy,lem:H-low-entropy}.
  We omit the details.
\end{proof}

\section{Equivalence with the hardness of identifying the ``span of easy~states''}

We now give a different characterization of EFIs, showing that they are equivalent to the hardness of deciding if a state is in the span of states of low $\Knet$ complexity.
We then introduce a notion of ``robust span'' of states, and briefly sketch how this can give a unified perspective into the proofs of \Cref{sec:gaph,sec:gapu}.

\subsection{Characterization of \texorpdfstring{$\EFI$}{EFI} from the hardness of identifying the ``span of easy states''}
In this section, we characterize $\EFI$ with the hardness of learning whether a state in the span of easy states or not. We will denote $\Pi_r$ as the span of states with $\Knet$ at most $r$.
\begin{remark}
    The definition of $\Pi_r$ is not robust over choice of the universal gate set: different choice of the gate set would correspond to different sets of easy states, and different sets of easy states (even with a very little deviation) have different spans. According to Solovay-Kitaev, universal gate set can simulate any state up to arbitrary precision, but the span of two states might differ significantly even if two state families are very close to each other. For example, span of $\{\ket{0}, \ket{0}\}$ is one-dimensional, and span of $\{\ket{0}, \sqrt{1-2^{-2n}}\ket{0}+2^{-n}\ket{1}\}$ is two-dimensional, albeit these two states are almost identical. But we will see that the non-robustness does not matter a lot for our arguments.
\end{remark}

We'll need lemma on the entropy bound related to $\Pi_r$.

\begin{lemma}\label{lem:span-low-entropy}
    For any family of states $\{\ket{\psi_k}\}$ such that 
  $\ket{\psi_k}$ lies in $\Pi_r$
    for all $k$,
  the mixed state
  $\E_k \ketbra{\psi_k}{\psi_k}$ 
  has von
  Neumann entropy at most $r+1$.
\end{lemma}
\begin{proof}
    As $\Pi_r$ is spanned by states with $\Knet$ at most $r$, and there are at most $2^{r+1}-1$ such states,
    the dimension of $\Pi_r$ is at most $2^{r+1}-1$. 
    Since $\E_k \ketbra{\psi_k}{\psi_k}$ is a state supported in $\Pi_{r}$, its entropy is at most $\log (2^{r+1}-1) < r+1$.
\end{proof}

\begin{lemma}\label{lem:span-high-entropy}
    For any efficiently family of states $\{\ket{\psi_k}\}$ such that
    $\braket{\psi_k|\Pi_s|\psi_k} \leq \eps$
    for all $k$,
    the mixed state $\E_k \ketbra{\psi_k}{\psi_k}$ is $2\eps$-close to a state with min-entropy at least $s-\Gamma$ for any $\Gamma = \omega(\log n)$.
\end{lemma}
\begin{proof}
    According to the same argument as in~\cref{lem:H-high-entropy}, we can define $\tilde \Pi_\low$ as the eigenstates $\ket{\varphi_i}$ of
    $\rho = \E_k \ketbra{\psi_k}{\psi_k}$ 
    with eigenvalues $\beta_i\geq 2^{\Gamma - s-1}$. Then all such states
    can be encoded with the index of eigenstates in the decreasing order of
    eigenvalues, so $\Knet(\ket{\varphi_i}) \leq s-\Gamma +1 + O(\log n) <
    s$, and thus $\tilde \Pi_\low $ is a subspace contained in $\Pi_s$. 
    As a
    result, we have $\braket{\psi_k|\tilde \Pi_\low|\psi_k} \leq
    \braket{\psi_k|\Pi_s|\psi_k} \leq \eps$ for all $k$, so
    $\Tr (\tilde \Pi_\low \rho) = \E_k \braket{\psi_k | \tilde
    \Pi_\low|\psi_k} \leq \eps$. Now, according
    to~\cref{lem:min-entropy-eigenstate}, we have
    \begin{align*}
        \Hmin^{2\eps}(\rho) \geq s-\Gamma +1 + \log(1-\eps) \geq s-\Gamma.
    \end{align*}
    So we can conclude that $\rho$ is $2\eps$-close to a state with min-entropy at least $s-\Gamma$.
\end{proof}

\begin{theorem}\label{thm:EFI-span}
    $\EFI$ exist if and only if there exists a non-uniform family of efficiently samplable states $\{\ket{\psi_{k}}\}$ with advice size $O(\log \lambda)$, efficiently computable functions $r \in [n(\lambda)], \Delta = \omega(\log \lambda)$, and $\eps < 1/100$ such that it is non-uniformly hard to distinguish the following two cases given a single copy of a state from the family (sampled uniformly at random over $k$):
    \begin{itemize}
        \item The states lies in $\Pi_r$
        \item The overlap of the state with $\Pi_{r+\Delta}$ is at most
            $\eps$.
    \end{itemize}
\end{theorem}
\begin{proof}
    First, assuming $\EFI$ exist, we consider the same state family as
    in~\cref{thm:efi-nu-gaph}. When $b=0$, the state can be described by a
    program of size at most $n(\lambda)-\sqrt{n(\lambda)} + O(\log \lambda)
    +C$ for some constant $C$, i.e.. $\Knet(\ket{\psi_{k}}) \leq
    n(\lambda) - \sqrt{n(\lambda)} + O(\log \lambda) +C$. Hence, the state
    has low $\Knet$ complexity and lies in $\Pi_{n-\Delta}$ for $r =
    n-\sqrt n + O(\log n) + C$. On the other hand, when $b=1$, the state
    $\ket{\psi_{k}}$ is $\ket{j}$ for a uniformly random $j\in
    \{0,1\}^n$. Let $\Delta = n-\sqrt{n}/2-r$, then $r+\Delta =
    n-\sqrt{n}/2$. Since $\sum_j \braket{j|\Pi_{n-\sqrt{n}/2}|j} = \Tr
    \Pi_{n-2\sqrt{n}/3} \leq 2^{n-2\sqrt{n}/3}$,
    we obtain from
    Markov's inequality:
    \begin{align*}
        \Pr_{j\in \{0,1\}^n} [\braket{j|\Pi|j} \geq 2^{-\sqrt{n}/3}] \leq 2^{-\sqrt n/3} = \negl(\lambda).
    \end{align*}
    This means that the state has little overlap with $\Pi_{r+\Delta} =
    \Pi_{n-2\sqrt{n}/3}$ with high probability. Thus, with the same argument
    in~\cref{thm:efi-nu-gaph}, we conclude that $\ket{\psi_{k}}$ is a non-uniformly hard
    instance to decide whether it's in the span or not.

    Second, assuming $\EFI$ do not exist, we can apply the algorithm as
    in~\cref{thm:efi-nu-gaph}. Note that the same randomness extractor
    works as we have the entropy bound from~\cref{lem:span-high-entropy}
    and~\cref{lem:span-low-entropy}.
\end{proof}
\begin{remark}
    If we replace the non-uniform family of states with a uniform family, we will build up a characterization of $\OnePRS$. The argument goes in the same way as~\cref{cor:gaph-1prs}.
\end{remark}
\begin{remark}
    Here we adapt the strong version of promise: we require that the state lies
    exactly in $\Pi_r$. We can also adapt the robust version, modifying
    the requirement to be so that the state almost lies in the span of states.
\end{remark}

\subsection{Relating algorithmic entropy with the ``\emph{robust} span of easy states''}\label{sec:relating-entropy-span}

The definitions of $\Hbar$ and its robust version $\Hbar^\eps$ are
presented in~\cref{sec:kolmogorov}, but arguably the intuition behind the definition
of $\Hbar^\eps$ is still vague. In this section, we relate $\Hbar$ with the
robust span (we'll give the definition of the robust span later) of the
easy states: if $\Hbar$ is low, then the state almost lies in the robust
span of easy states; if $\Hbar$ is high, then the state almost lies in the
complement of the robust span of easy states.

\begin{definition}
    Let $\{\ket{\psi_k}\}_{k\in[L]}$ be a family of quantum sates. Then the
    \emph{$\eps$-robust span} of $\{\ket{\psi_k}\}$ is defined as the subspace
    spanned by eigenstates of
    $\frac{1}{L}\sum_{k\in [L]}\ketbra{\psi_k}{\psi_k}$
    whose corresponding eigenvalues are at least $\frac \eps L$.
\end{definition}

\begin{theorem}[High $\Hbar$ complexity implies low overlap with robust span of easy states]\label{thm:high-complexity-span}
    If a state $\ket{\psi}$ satisfies $\Hbar^0(\ket{\psi}) \geq r$ with $r \leq n$, and
    $\Pi_{r,\gamma}$ is the projector on the $\gamma$-robust span of states with $\Knet$ at most $r-\Delta$, then we have $\braket{\psi|\Pi_{r,\gamma}|\psi} \leq \poly(n) \gamma^{-1} 2^{-\Delta}$.
\end{theorem}

\begin{proof}
  We can bound $r$ by the number of qubits of $\ket{\psi}$: as there is a
  universal upper bound on $\Hbar$ for any states:
  $\Hbar(\ket{\psi}) \leq n+O(\log n)$ for any $n$-qubit state $\ket{\psi}$, it
  follows that $r\leq n+O(\log n)$.

  First, from the condition $\Hbar^0(\ket{\psi}) \geq r$ and the definition of
  $\Hbar$, we have
  $\Hbar^0(\ket{\psi}) = \Hbar(\ket{\psi}) = -\log \braket{\psi|\udm|\psi} \geq r$,
  which implies $\braket{\psi|\udm|\psi} \leq 2^{-r}$.

  Then we relate $\udm$ with the robust span $\Pi_{r,\gamma}$.
  By the definition of the universal semi-density matrix $\udm$, we have that
  $\udm \geq \sum_{\ket{\phi}} 2^{-\Knet(\ket{\phi})} \ketbra{\phi}{\phi} \geq 2^{-r+\Delta} \sum_{\Knet(\ket{\phi}) \leq r-\Delta} \ketbra{\phi}{\phi}$.
  Although this appears to sum over uncountably many quantum states, it is
  actually a sum over a countable family: there are only countably many quantum
  states with finite $\Knet$.

  Let $\rho = \E_{\Knet(\ket{\phi}) \leq r-\Delta} \ketbra{\phi}{\phi}$, where
  the expectation is taken uniformly at random over all the quantum states
  $\ket{\phi}$ with $\Knet$ at most $r-\Delta$.
  Let $L$ be the number of states with $\Knet$ at most $r-\Delta$.
  Note that we have a lower bound of $L$: any states
  $\ket{k\|0^{n-r+\Delta+O(\log n)}}_{k \in \{0,1\}^{r-\Delta-O(\log n)}}$ can
  be encoded by a prefix-free Turing machine that outputs $n$ and $k$ (note that
  $r\leq n+O(\log n)$ so the argument holds), so there are at least
  $2^{r-\Delta-O(\log n)}$ different states with $\Knet$ at most $r-\Delta$.
  As a result, we have
  \begin{equation*}
      \udm \geq 2^{-r+\Delta} \sum_{\Knet(\ket{\phi})\leq r-\Delta} \ketbra{\phi}{\phi} \geq 2^{-O(\log n)} \frac{1}{L} \sum_{\Knet(\ket{\phi})\leq r-\Delta} \ketbra{\phi}{\phi} = \frac{1}{\poly(n)} \rho.
  \end{equation*}
  Thus we can conclude that
  $\braket{\psi|\rho|\psi} \leq \poly(n)\braket{\psi|\udm|\psi} \leq \poly(n)\cdot 2^{-r}$.

  Let $\sum_i \lambda_i \ketbra{\psi_i}{\psi_i}$ be the
  eigendecomposition of $\rho$. The $\gamma$-robust span of $\rho$ can be
  expressed as $\Pi_{r,\gamma} = \sum_{\lambda_i \geq \gamma/L} \ketbra{\psi_i}{\psi_i}$.
  Thus $\frac \gamma L \Pi_{r,\gamma} \leq \rho$, and we have
  $\braket{\psi|\Pi_{r,\gamma}|\psi} \leq \frac L\gamma \braket{\psi|\rho|\psi} \leq \frac{L}{\gamma} \poly(n) 2^{-r} = \poly(n) \gamma^{-1} 2^{-\Delta}$.
\end{proof}

\begin{theorem}[Low $\Hbar$ complexity implies high overlap with robust span of easy states]\label{thm:low-complexity-span}
    If an $n$-qubit state $\ket{\psi}$ satisfies
    $\Hbar^{1-\eps}(\ket{\psi}) \leq r$, and $\Pi_{r,\gamma}$ is
    the projector of the $\gamma$-robust span of states with $\Knet$ at
    least $r+O(\log n)$, then we have $\braket{\psi|\Pi_{r,\gamma}|\psi}
    \geq 1 - \sqrt{2\eps} - \gamma \poly(n)$.
\end{theorem}
\begin{proof}
      Let $\sum_{i} \mu_{i} \ketbra{\phi_{i}}{\phi_{i}} = \udm$ be the spectral
  decomposition of the universal density matrix $\udm$.
  Define $\Pi_{\low}$ to be the projection onto the span of
  eigenstates $\ket{\phi_{i}}$ of $\udm$ with eigenvalue $\mu_{i} \geq
  2^{-r}$,
  and $\Pi_{\high}$ to be the projection onto the span of high
  complexity eigenstates, i.e.\ those with eigenvalue less than $2^{-r}$.
  Since $\udm$ has trace at most $1$, the number of $\mu_{i}$'s at least
  $2^{-r}$ is at most $2^{r}$, and therefore the low-complexity space defined by
  $\Pi_{\low}$ has dimension at most $2^{r}$.

  For any state $\ket{\psi}$, define
  $\ket{\psi_{\high}} = \Pi_{\high} \ket{\psi} / \norm{\Pi_{\high} \ket{\psi}}$
  and
  $\ket{\psi_{\low}} = \Pi_{\low} \ket{\psi} / \norm{\Pi_{\low} \ket{\psi}}$.

  We first show that for any state $\ket{\psi}$ satisfying
  $\Hbar^{1-\eps}(\ket{\psi}) \le r$, the trace distance between $\ket{\psi}$
  and $\ket{\psi_{\low}}$ is at most $\sqrt{2\eps}$.
  To establish this, observe that
  \begin{equation*}
    D(\ket{\psi}, \ket{\psi_{\high}}) =
    \sqrt{1 - \abs{\braket{\psi | \psi_{\high}}}^{2}} = \norm{\Pi_{\low} \ket{\psi}}.
  \end{equation*}
  If $\norm{\Pi_{\low} \ket{\psi}} \le 1-\eps$, then by the definition of
  $\Hbar^{1-\eps}(\ket{\psi})$, we would have
  \[
    \Hbar^{1-\eps}(\ket{\psi}) \ge \Hbar(\ket{\psi_{\high}}) > r,
  \]
  since $\ket{\psi_{\high}}$ is $(1-\eps)$-close to $\ket{\psi}$ and has
  $\Hbar$-complexity greater than $r$.
  This contradicts the assumption that $\Hbar^{1-\eps}(\ket{\psi}) \le r$.

  Therefore, it must be that $\norm{\Pi_{\low} \ket{\psi}} > 1 - \eps$.
  Consequently,
  \begin{equation*}
    D(\ket{\psi}, \ket{\psi_{\low}}) = \sqrt{1 - \norm{\Pi_{\low}\ket{\psi}}^{2}}
    < \sqrt{1 - {(1 - \eps)}^{2}} < \sqrt{2\eps}.
  \end{equation*}

    Thus, $\ket{\psi}$ is $\sqrt{2\eps}$-close to a state $\ket{\psi_\low}$ lies in the projector $\Pi_\low$ where $\Pi_\low$ is the projector onto eigenstates of $\mu$ whose corresponding eigenvalues are at least $2^{-r}$. These states can be encoded with Turing machine of size $r+O(\log n)$.\footnote{More precisely, we can encode the eigenstates of $\Pi_\low$ with a Turing machine of size $r+O(\log n)$ that can output a circuit that is $2^{-2^{n}}$-close to the eigenstate. Such encoding can be implemented as, for example, the number of qubits and the index of the eigenstates. The double exponential error does not affect our result so we will ignore the error afterwards.}

    Let $\rho = \E_{\Knet(\ket{\phi}) \leq r+O(\log n)} \ketbra{\phi}{\phi}$ be the uniform mixture of states with $\Knet$ at most $r+O(\log n)$. Then as all the eigenstates of $\Pi_\low$ have $\Knet$ at most $r+O(\log n)$, we can conclude that $\rho \geq \frac 1{\poly(n)} 2^{-r} \Pi_\low \geq \frac 1{\poly(n)}2^{-r} \ketbra{\psi_\low}{\psi_\low}$.

    Let $\rho = \sum_i \lambda_i \ketbra{\psi_i}{\psi_i}$ be the
    eigendecomposition of $\rho$, and $L$ be the number of states with
    $\Knet$ at most $r+O(\log n)$. Then the $\gamma$-robust span of $\rho$
    can be expressed as $\Pi_{r,\gamma} = \sum_{\lambda_i \geq
    \gamma/L}\ketbra{\psi_i}{\psi_i}$. Let $\ket{\tilde\psi_{\low}} =
    \frac{(I-\Pi)\ket{\psi_\low}}{\norm{(I-\Pi)\ket{\psi_\low}}}$, then as
    $\ket{\tilde \psi_\low}$ lies in the subspace with spectrum bounded by
    $\frac \gamma L$, so $\braket{\tilde\psi_\low|\rho|\tilde \psi_\low}
    \leq \gamma/L$.
    But on the other hand, we have $\braket{\tilde\psi_\low | \rho | \tilde
    \psi_\low} \geq \frac{1}{\poly(n)}2^{-r} |\braket{\tilde \psi_\low|
    \psi_\low}|^2$, thus we can deduce that $|\braket{\tilde \psi_\low |
    \psi_\low}|^2 \leq \frac \gamma L \poly(n)2^r \leq \gamma \poly(n)$,
    where the last line follows from the fact that $L = 2^{r+O(\log n)}$.

    So we have $\braket{\psi_\low | \Pi_{r,\gamma} | \psi_\low} = 1 -
    |(I-\Pi_{r,\gamma})\ket{\psi_\low}|^2 = 1 - |\braket{\psi_\low|\tilde
    \psi_\low}|^2\geq 1-\gamma \poly(n)$. 
    Combined with $\TD(\ketbra{\psi}{\psi}, \ketbra{\psi_\low}{\psi_\low})
    \leq \sqrt{2\eps}$,
    we get that $\braket{\psi|\Pi_{r,\gamma}|\psi} \geq 1-\gamma \poly(n) -
    \sqrt{2\eps}$.
\end{proof}

\begin{theorem}\label{thm:robust-span-low-entropy}
    For any family of states $\{\ket{\psi_k}\}$ such that
    $\braket{\psi_k|\Pi_{r,\gamma}|\psi_k} \geq 1-\eps$ for all $k$, where
    $\Pi_{r,\gamma}$ is the $\gamma$-robust span of the states with $\Knet$
    at most $r$, then the mixed state $\E_k \ketbra{\psi_k}{\psi_k}$ is
    $\sqrt \eps$-close to a state with von Neumann entropy at most $r$.
\end{theorem}
\begin{proof}
    We first show that any state $\ket{\psi}$ satisfying
    $\braket{\psi|\Pi_{r,\gamma}|\psi}\geq 1-\eps$ is $\sqrt \eps$-close to a
    pure state $\ket{\tilde \psi}$ which lies in the support of $\Pi_{r,\gamma}$.
    Indeed,
    let $\ket{\tilde \psi} = \frac{\Pi_{r,\gamma}
    \ket{\psi}}{\|\Pi_{r,\gamma}\ket{\psi}\|}$, then $\braket{\psi|\tilde
    \psi} =
    \frac{\|\Pi_{r,\gamma}\ket{\psi}^2\|}{\|\Pi_{r,\gamma}\ket{\psi}\|} =
    \sqrt{\braket{\psi|\Pi_{r,\gamma}|\psi}} \geq \sqrt{1-\eps}$, and thus
    \begin{align*} 
        \| \ketbra{\psi}{\psi} - \ketbra{\tilde \psi}{\tilde
    \psi} \| = \sqrt{1 - |\braket{\psi|\tilde \psi}|^2} \leq \sqrt{\eps}.
    \end{align*}
    Thus $\E_k\ketbra{\psi_k}{\psi_k}$ is $\sqrt{\eps}$-close to $\E_k
    \ketbra{\tilde \psi_k}{\tilde \psi_k}$ with $\ket{\tilde \psi_k}$ lies
    in the support of $\Pi_{r,\gamma}$. But we know that $\Pi_{r,\gamma}$
    is of dimension at most $2^{r}$ (because the rank of
    $\E_{\Knet(\ket{\psi})\leq r}\ketbra{\psi}{\psi} $ is bounded by
$2^r$), so $\E_k \ketbra{\tilde \psi_k}{\tilde \psi_k}$ is of von Neumann
entropy at most $r$.
\end{proof}

\begin{theorem}\label{thm:robust-span-high-entropy}
    For any samplable family of states $\{\ket{\psi_k}\}$ such
    that $\braket{\psi_k|\Pi_{r,\gamma}|\psi_k} \leq \eps$ where
    $\Pi_{r,\gamma}$ is the $\gamma$-robust span of the states with $\Knet$
    at most $s$, the mixed state $\E_k \ketbra{\psi_k}{\psi_k}$ is
    $\eps$-close to a state with min-entropy at least $s-\Gamma$ for some
    $\Gamma = \omega(\log n)$.
\end{theorem}
\begin{proof}
     Let $\rho = \sum_{i} \beta_{i} \ketbra{\varphi_{i}}{\varphi_{i}}$ be the
  spectral decomposition of $\rho$.
  We prove the lemma by showing that removing the large eigenvalues does not
  significantly change the state.
  Let $\tilde{\Pi}_\low$ be the projection onto the subspace spanned by
  $\ket{\varphi_{i}}$ with $\beta_{i} \ge 2^{\Gamma - s -1}$.
  Each of these eigenvectors of $\rho$ has a short program of length at most
  $s - \Gamma + O(\log n)$.
  That is, for all such $\ket{\varphi_{i}}$,
  we have
  $\Knet(\ket{\varphi_{i}}) \le s - \Gamma + O(\log n) $.
  
  Now, for any state $\ket{\psi}$ with $\braket{\psi|\Pi_{r,\gamma} |\psi} \le \eps$, note that $\frac{1}{L} \tilde\Pi_{\low} \leq \frac 1L \sum_{\Knet(\ket{\psi})\leq s} \ketbra{\psi}{\psi}\coloneqq \rho$, where $L$ is the number of states with $\Knet$ at most $s$. So 
  \begin{align*}
  \braket{\psi|\tilde \Pi_\low | \psi} &= \braket{\psi|(I- \Pi_{r,\gamma}) \tilde \Pi_\low (I-\Pi_{r.\gamma}) |\psi } + \braket{\psi|(I- \Pi_{r,\gamma}) \tilde \Pi_\low \Pi_{r.\gamma} |\psi } \\&+ \braket{\psi|\Pi_{r,\gamma} \tilde \Pi_\low (I-\Pi_{r.\gamma}) |\psi } + \braket{\psi| \Pi_{r,\gamma} \tilde \Pi_\low \Pi_{r.\gamma} |\psi } \\
  &\leq 3\sqrt \eps + \braket{\psi|(I-\Pi_{r,\gamma}) \tilde \Pi_\low (I-\Pi_{r,\gamma}) |\psi}\\ 
  &= 3 \sqrt \eps + L\braket{\psi|(I- \Pi_{r,\gamma}) \rho (I-\Pi_{r.\gamma}) |\psi } \\
  &\leq 3\sqrt \eps + \gamma,
  \end{align*}
  where the second line is from the fact that $\|\Pi_{r,\gamma}\ket{\psi}\|\leq \sqrt \eps$, and the last inequality is from the fact that $\|(I-\Pi_{r,\gamma})\rho(I-\Pi_{r,\gamma})\|\leq \gamma/L$.
  So as a result,
  \begin{equation*}
    \Tr (\tilde \Pi_\low \rho) = \E_k \braket{\psi_k | \tilde \Pi_\low | \psi_k}
    \le 3\sqrt{\eps}+\gamma.
  \end{equation*}
  According to~\cref{lem:min-entropy-eigenstate}, we can conclude that $\Hmin^{6\sqrt{\eps}+2\gamma}(\rho) \geq s+\log (1-3\sqrt{\eps}-\gamma) -\Gamma +1\geq s-\Gamma$.
\end{proof}
So one can write a new proof for~\cref{lem:H-low-entropy}
and~\cref{lem:H-high-entropy} (probably with different bounds, but the
bound will be negligible as long as the gap is $\omega(\log n)$): any state
family with high $\Hbar$ will almost lies in the robust span according
to~\cref{thm:high-complexity-span}, whose mixture has a high smoothed
min-entropy according to~\cref{thm:robust-span-high-entropy}. On the other
hand, any state family with low $\Hbar^{1-\eps}$ almost lies in the robust
span according to~\cref{thm:low-complexity-span}, whose mixture have low
robust von Neumann entropy according to~\cref{thm:robust-span-low-entropy}.

\section{Equivalence with hardness of state complexity over unkeyed state families}
\label{sec:unkeyed}

In \cref{sec:gaph,sec:gapu}, the state families considered in the $\GapH$ and
$\GapU$ problems are keyed-samplable, meaning that it is possible to output
the state $\ket{\psi_{k}}$ given the key $k$ as input.
In this section, we show similar characterizations of $\EFI$ using the hardness
of Kolmogorov complexity for single-copy samplable state families. As defined in~\cref{def:single-copy-samplable-family},
a single-copy samplable state family is a family of key-state pairs $\{(k, \ket{\psi_{k}})\}$ and a distribution on it
which can be sampled by running a generation unitary $G$ on two systems $A$ and
$B$ and measuring $A$ in the computational basis.
System $A$ holds the key $k$, and $B$ holds the quantum state.
This is more general than a keyed state family, as it is not guaranteed that the
state $\ket{\psi_{k}}$ can be reproduced even given the key $k$.
Such unkeyed state families are relevant in quantum cryptography for quantum
money and quantum lightning.
Our main result in this section is the following \cref{thm:unkeyed}.

We need a variant of the $\GapH$ problem where in the high complexity case, the
measure is also smoothed.
\begin{definition}[The $\DGapH$ problem] Let $r, \Delta, n$ be functions of $\lambda$. We define $\DGapH^\eps(r, r+\Delta)$ as the following (promise) problem: given a \emph{single} copy of a state $\ket{\psi}$ on some number $n$ of qubits, decide whether
\begin{itemize}
    \item $\Hbar^{1-\eps}(\ket{\psi}) \leq r$, or
  \item $\Hbar^{\eps}(\ket{\psi}) \geq r+\Delta$.
\end{itemize}
\end{definition}

\begin{theorem}\label{thm:unkeyed}
  The following two statements are equivalent:
  \begin{itemize}
    \item $\EFI$ exist.
    \item There exists a single-copy samplable family of states $(\{\ket{\psi_k}\}, \{\Dist_n\})$ (namely there exists a QPT algorithm that can sample $\ket{\psi_k}$ according to the distribution $\Dist_n$), an (inefficiently computable) function
          $r(n) \in [n]$, and efficient function $\Delta(n) = \omega(\log(n))$ such that
          $\DGapH(r, r+\Delta)$ is hard on average over $\Dist_n$.
  \end{itemize}
\end{theorem}
\begin{remark}
    Note that in contrast to all other quantities the function $r$ here is not necessarily efficiently computable. In our upcoming proof, the $r$ will be the entropy of $\EFI$, which does not necessarily have an efficient QPT algorithm. So this is also a non-uniform characterization of $\EFI$.
\end{remark}
Before proving the theorem, we first establish several useful lemmas.

\begin{lemma}\label{lem:U-hmax-bound}
  For any (not necessarily efficiently) sampleable family of mixed states
  $\{\rho_n\}$, where $\rho_n$ is an $n$-qubit mixed state, and for any pure
  state $\ket{\psi}$ in the support of $\rho_n$ (namely there exists $\eps>0$
  such that $\eps \ket{\psi}\bra{\psi} \leq \rho_n$), we have
  $\Umin (\ket{\psi}) \leq \Hmax (\rho_n)+\log n + C$
  for some constant $C$.
\end{lemma}

\begin{proof}
  Let $r$ be the max-entropy of $\rho_n$, and let
  $\rho_n = \sum \lambda_i \ket{\psi_i}\bra{\psi_i}$ be a spectral
  decomposition.
  As the max-entropy of $\rho_n$ is $r$, there are at most $2^r$ different
  $\ket{\psi_i}$, so we can encode all these eigenstates with a program of size
  $r+\log n + C$ for some constant $C$ dependent on the state generation
  algorithm.
  We can therefore bound the $\Knet^0$ of all the eigenstates of $\rho_n$:
  $\Knet^0(\ket{\psi_i})\leq r+\log n+C$.
  So the projector $\Pi = \sum_{i} \ket{\psi_i}\bra{\psi_i}$ spanned by the
  support of $\rho_n$ can be bounded as $2^{-(r+\log n +C)}\Pi \leq \udm$.
  As a result, any state $\ket{\psi}$ in the support of $\rho_n$ also satisfies
  $2^{-(r+\log n+C)}\ket{\psi}\bra{\psi} \leq 2^{-(r+\log n +C)}\Pi \leq \udm$.
  By \cref{lem:U-bound}, we have $\Umin(\psi) \le r + \log n + C$.
\end{proof}

\begin{lemma}\label{lem:hbar-hmax-bound}
  For any (not necessarily efficiently) sampleable family of mixed states
  $\{\rho_n\}$, where $\rho_n$ is an $n$-qubit mixed state, and for any pure
  state $\ket{\psi}$ in the support of $\rho_n$ (namely there exists $\delta>0$
  such that $\delta \ket{\psi}\bra{\psi} \leq \rho_n$), we have
  $\Hbar^{1-\eps} (\ket{\psi}) \leq \Hmax (\rho_n)+\log n + \log \frac 1 \eps +C$
  for some constant $C$.
\end{lemma}

\begin{proof}
  It follows from \cref{lem:U-hmax-bound,lem:U-lower-bound}.
\end{proof}

\begin{corollary}\label{cor:support-state-low}
  For any (not necessarily efficiently) samplable family of mixed states
  $\{\rho_n\}$, where $\rho_n$ is an $n$-qubit state, and for
  $m \geq 36 n^2 \log^2 n$ and any decomposition of
  $\rho^{\otimes m} = \sum p_k \ket{\psi_k}\bra{\psi_k}$ we have that except
  with probability at most $2^{-O(\log^2 n)}$ over $\ket{\psi_{k}}$ with
  probability $p_{k}$,
  $\Hbar^{1-\eps}(\ket{\psi_{k}}) \le m(S(\rho)+1) + \log n + \log \frac{1}{2\eps} + C$.
\end{corollary}

\begin{proof}
  Using \cref{cor:min-entropy-explicit-bound} with $\xi = 2^{-\log^2 n}$, we
  deduce that
  \begin{equation*}
    \Hmax^\xi(\rho^{\otimes m}) \leq m \Bigl(S(\rho) +
    6n \sqrt \frac{\log 1/\xi }{m}\,\Bigr) < m \bigl(S(\rho) + 1\bigr).
  \end{equation*}
  Thus, $\rho^{\otimes m}$ is $\xi$-close to a semi-density matrix $\rho'$ with
  max-entropy at most $m \bigl(S(\rho)+1\bigr)$.
  Let $\Pi$ be the projector onto the span of $\rho'$.
  Then we have
  \begin{equation}\label{eq:support-state}
    \begin{split}
      \frac{\Pi \rho^{\otimes m} \Pi}{\Tr(\Pi \rho^{\otimes m})}
      & = \frac{1}{\Tr(\Pi \rho^{\otimes m})}
        \sum_{k} p_k \Pi \ketbra{\psi_k}{\psi_k}\Pi\\
      & = \sum_{k} p_k \frac{\norm{\Pi \ket{\psi_k}}^2}{\Tr(\Pi \rho^{\otimes m})}
        \frac{\Pi\ket{\psi_k}}{\norm{\Pi\ket{\psi_k}}}
        \frac{\bra{\psi_k}\Pi}{\norm{\Pi\ket{\psi_k}}}.
    \end{split}
  \end{equation}
  This gives a decomposition of $\Pi \rho^{\otimes m} \Pi$ composed of states
  $\frac{\Pi\ket{\psi_k}}{\norm{\Pi\ket{\psi_k}}}$.
  Therefore, by~\cref{lem:hbar-hmax-bound}, we have
  $\Hbar^{1-2\eps}\left(\frac{\Pi\ket{\psi_k}}{\norm{\Pi\ket{\psi_k}}}\right)
  \leq r+\log n + C + \log \frac{1}{2\eps}$
  for any $\ket{\psi_k}$.

  And we know that
  $\Tr \Pi \rho = \sum p_k \braket{\psi_k|\Pi|\psi_k}
  \geq 1 - \TD(\rho, \rho') = 1-2^{-\log^2n}$,
  so except for negligible probability, we have that
  $\norm{\Pi\ket{\psi_{k}}}^2 \geq 1- \eps^{2}$ and
  \begin{equation*}
    D\Bigl(\ket{\psi_{k}}, \frac{\Pi\ket{\psi_k}}{\norm{\Pi\ket{\psi_k}}}\Bigr) =
    \sqrt{1 - \norm{\Pi \ket{\psi_{k}}}^{2}} \le \eps.
  \end{equation*}
  Therefore, except for negligible probability,
  $\Hbar^{1 - \eps}(\ket{\psi_k}) \leq m(S(\rho)+1) + \log n + \log\frac{1}{2\eps} + C$.
\end{proof}

\begin{lemma}\label{lem:hbar-hmin-bound}
  For any $\delta > 0$, any state ensemble $\{(p_{k}, \ket{\psi_{k}})\}$ of
  $n$-qubit states, and $\rho = \sum_{k} p_{k} \ketbra{\psi_{k}}{\psi_{k}}$, we
  have $\Hbar^0(\ket{\psi_k}) \geq \Hmin(\rho)-\delta$ with probability at least
  $2^{-\delta}$ over $\ket{\psi_k}$.
\end{lemma}

\begin{proof}
  Let $\Hmin(\rho) = r$; then we have $\rho \leq 2^{-r} I$.
  Consequently,
  \begin{equation*}
    \sum p_k \braket{\psi_k|\udm|\psi_k} =
    \Tr(\udm \rho) \leq 2^{-r} \Tr(\udm) \le 2^{-r}.
  \end{equation*}
  Thus, by Markov's inequality, with probability at least $1 - 2^{-\delta}$
  over the choice of $\ket{\psi_k}$,
  $\braket{\psi_k|\udm|\psi_k} \leq 2^{-r + \delta}$.
  In other words, $\Hbar^0(\ket{\psi_k}) \geq r - \delta$ with probability at
  least $1 - 2^{-\delta}$.
\end{proof}

\begin{corollary}\label{cor:support-state-high}
  For any $n$-qubit state $\rho$ and any $m\geq 36 n^2 \log \frac{1}{\eps}$ and
  any decomposition $\rho^{\otimes m} = \sum p_k \ket{\psi_k}\bra{\psi_k}$,
  except with probability $2^{-\log^2n}+2^{-\delta}$ we have
  $\Hbar^\eps(\ket{\psi_k}) \geq m(S(\rho)-1) - \delta$.
\end{corollary}

\begin{proof}
  According to~\cref{cor:min-entropy-explicit-bound}, setting
  $\xi = 2^{-\log^2 n}$, we have
  \begin{equation*}
    \Hmin^\xi(\rho^{\otimes m}) \geq m \Bigl(S(\rho)
    -6n \sqrt{\frac{\log 1/\xi}{m}}\, \Bigr) > m(S(\rho)- 1).
  \end{equation*}
  Thus $\rho^{\otimes m}$ is $\xi$-close to a state $\rho'$ with min-entropy at
  least $m(S(\rho)- 1)$.
  Thus we can find a projector $\Pi$ spanned by eigenstates of $\rho$ such that
  $\Tr \Pi\rho \geq 1-\xi$ and $\frac{\Pi\rho \Pi}{\Tr \Pi\rho}$ has min-entropy
  at least $m(S(\rho)- 1)$.
  Then, by the expansion in \cref{eq:support-state}, we have a decomposition of
  $\frac{\Pi\rho\Pi}{\Tr \Pi\rho}$ composed of
  $\frac{\Pi \ket{\psi_k}}{\norm{\Pi\ket{\psi_k}}}$.
  Thus according to~\cref{lem:hbar-hmin-bound}, we have that except with
  probability $2^{-\delta}$,
  $\Hbar^0(\frac{\Pi\ket{\psi_k}}{\norm{\Pi\ket{\psi_k}}})
  \geq m(S(\rho)-1) - \delta$.
  The distance
  $D \bigl(\ket{\psi_{k}}, \frac{\Pi\ket{\psi_k}}{\norm{\Pi\ket{\psi_k}}} \bigr)
  = \sqrt{1 - \norm{\Pi\ket{\psi_k}}^{2}}$,
  thus by Markov's inequality,
  \begin{equation*}
    \Pr\biggl[ D \Bigl(\ket{\psi_{k}},
    \frac{\Pi\ket{\psi_k}}{\norm{\Pi\ket{\psi_k}}} \Bigr) \le \eps \biggr] =
    \Pr\bigl[\norm{\Pi\ket{\psi_{k}}}^{2} \ge 1 - \eps^{2}\bigr]
    \ge  1 - \frac{\xi}{\eps^2}.
  \end{equation*}
  That is, except for negligible probability,
  $\Hbar^\eps(\ket{\psi_k}) \geq m(S(\rho)-1) - \delta$ holds for any
  $\delta = \omega(\log n)$.
\end{proof}

\begin{proof}[Proof of \cref{thm:unkeyed}]
  In the case where $\EFI$ exist, there exists an entropic $\EFI$ pair $\rho_0$ and
  $\rho_1$ such that $S(\rho_1) - S(\rho_0) \geq \sqrt{n}$ (if we have an
  entropy gap of at least $1$, we can generally boost it to $\sqrt{n}$ by taking
  multiple independent copies).
  We can prepare the purification of the state and measure the purification
  registers in the computational basis.
  This results in a distribution of pair $(k, \ket{\psi_k})$, where $k$ is the measurement
  outcome and $\ket{\psi_k}$ is the post-measurement state in the $\EFI$ registers. Name the distribution as $\Dist_n$.

  Then $(\ket{\psi_k}, \Dist_n)$ is a single-copy samplable state family and
  $\mathbb{E}_{k \sim \Dist_n} \ket{\psi_k}\bra{\psi_k} = \rho_0$.
  Thus, we can define a single-copy state family that with probability $1/2$
  samples according to $\rho_0$ and with probability $1/2$ samples according to
  $\rho_1$.
  According to \cref{cor:support-state-high,cor:support-state-high},
  $\Dist_n$ is a $\DGapH[r, r+\Delta]$ instance, where $r = mS(\rho)$ and
  $\Delta = \sqrt{m}$.
  Thus, according to the security of $\OnePRS$, $\DGapH$ is hard over
  $\Dist_n$ given a single copy of the state.

  Assuming $\EFI$ do not exist, we can apply the algorithm
  in~\cref{thm:efi-nu-gaph}.
  The argument is exactly the same, as we never rely on the fact that the state
  family is an efficiently samplable keyed family.
\end{proof}

\begin{remark}
    The same argument also works well for the $\GapU$ characterization of $\EFI$. With similar arguments, we can show that $\EFI$ exist if and only if $\mathsf{Double}\GapU[r, r+\Delta]$ is hard on average over some single-copy samplable state family, where $r$ is an inefficiently computable function.
\end{remark}

\paragraph*{Acknowledgements.}
Bruno C. is supported by the EPSRC project EP/Z534158/1 on "Integrated Approach
to Computational Complexity: Structure, Self-Reference and Lower Bounds". A.C. is thankful for support from the Google Research Scholar program.

\bibliographystyle{alpha}
\bibliography{bibliography}

\newpage

\appendix

\section{Quantum Kolmogorov complexities - Invariance and Equivalence}

In \Cref{sec:kolmogorov} we claimed that the definitions of $\Knet$ and $\Hbar$ are robust to changes in our choice of universal Turing machine and gate basis. We also claimed that our notion of $\Hbar$ is equivalent to the one introduced by G\'{a}cs. In this appendix we will give proofs for those three claims.

\subsection{\texorpdfstring{$\Knet$}{Knet} Invariance}\label{sec:robust-qk-gateset}

In \Cref{sec:kolmogorov}, we claimed that the definition of $\Knet$ only
changes by nearly a constant when $B$ and $U$ are universal. 
Because the definition of $\Knet$ goes through circuits which are
represented as strings, we can change our choice of universal Turing
machine while only incurring a fixed constant difference. However, because we
cannot guarantee that $B$ circuits can exactly simulate all the gates in
$B'$, it is possible that there exists states for which
$\Knet^{U,B,0}(\ket{\psi}) = \infty$ while $\Knet^{U',B',0}(\ket{\psi}) =
c$. 
Thus, an ``ideal'' invariance theorem of the following form
cannot hold:
$\forall
U,U',B,B'\exists c : \Knet^{U,B,0}(\ket{\psi}) \leq
\Knet^{U',B',0}(\ket{\psi})+c$.
However if we allow for
some very small additional error term $\delta$, and a just slightly super
constant difference, we are able to show an invariance, as the following
lemma states.

\begin{lemma}\label{lem:knet-invarience}
For any universal Turing machines $U$ and $U'$, universal quantum gate sets $B$, and $B'$ with computable amplitudes, and $m$ qubit state $\ket{\psi}$, we get that 
\begin{equation*}
\Knet^{U,B,\eps+\delta}(\ket{\psi}) \leq \Knet^{U',B',\eps}(\ket{\psi}) + O(1) + \min_{v >1/\delta}[\K(v |m)].
\end{equation*}
\end{lemma}
\begin{proof}

    For convenience we will label $\alpha = \Knet^{U',B',\eps}(\ket{\psi})$. By the definition of $\Knet$ we know that there exists some $B'$ circuit $C_{B'}$ such that $\K_{U'}(C_{B'}) = \alpha$ and $\abs{\braket{\psi|C_{B'}|0^m}}^2 \geq 1-\eps$. By the invariance of $K$ we know that $\K_{U}(C_{B'}) \leq \alpha + O(1)$.
    
    By the multi-qubit Solovay-Kitaev theorem and algorithm~\cite{DawsonN06}, for any pair of universal gate sets $B$, $B'$ with computable amplitudes there exists a constant length program $SK_{B,B'}(\eps,b)$ which takes in $\eps$ and $b \in B'$ and outputs a $B$ circuit $C_b$ approximately computing $b$ such that $\abs{\abs{C_b - b}}_\infty \leq \eps$. Define $C_B$ to be the circuit that results from replacing each of the $\abs{C_{B'}}$ gates with the circuit $SK_{B,B'}(\delta/3\abs{C_{B'}},b)$. The resulting circuit will approximate $C_{B'}$ such that $\abs{\abs{C_B - C_{B'}}}_\infty \leq \delta/3$ and consequently $\forall \ket{\psi}: \abs{\bra{\psi}C_{B}^\dagger C_{B'}\ket{\psi}} \leq 1-\delta/3$ and  $\abs{\bra{\psi}C_{B}^\dagger C_{B'}\ket{\psi}}^2 \leq 1-\delta$. 

    Let $\min_{n>1/\delta}[K(n)\mid m] = l$, then there exists a program with size $l$ that can output some $n_0>1/\delta$. If we consider an optimal program which first generates $C_{B'}$, then generates some number $1/3n_0 < \delta/3$ (this can either be done from no input or taking as input some value that the program has already generated such as the length of the output of the circuit i.e. $m$), then runs $SK_{B,B'}(1/3v\abs{C_{B'}},b)$ this will generate a state $\ket{\psi'}$ such that $|\braket{\psi'|C_{B'}|0}|^2 \leq 1/v \leq \delta$, and by the triangle inequality $|\braket{\psi'|C_{B'}|0}|^2 \leq \delta + \eps$. This program will will have length at most $\K_{U}(C_{B'}) \leq \alpha + O(1)$ to generate $C_{B'}$, plus $l + O(1)$ to generate a number ($1/3v$)  which is smaller than $1/3\delta$, plus $O(1)$ to implement $SK_{B,B'}$ and apply $SK_{B,B'}(1/3v\abs{C_{B'}},b)$ to each gate of the circuit.
\end{proof}

Note that as observed by Li and Vitanyi~\cite[sec 3.3]{LiV19},
the summand $\min_{n
>1/\delta}[\K(n)] \leq \min_{n >1/\delta}[\K(n) | m] $ $+ O(1)$ goes to
infinity more slowly than any unbounded monotonic computable function. 
Thus,
while the final term in the lemma above is not quite constant, it is
of the order
$o(f(1/\delta))$ for any monotonic unbounded computable function in
$1/\delta$, including $\log(\log(\ldots\log(1/\delta)))$ for any number of
composed logarithm's. Consequently, even for exceptionally small $\delta$,
the complexity
$\Knet^{U,B,\eps+\delta}(\ket{\psi})$ is only ever a just barely super
constant amount larger than $\Knet^{U',B',\eps}(\ket{\psi})$.

\subsection{\texorpdfstring{$\Hbar$}{H} Invariance}

As in our discussion of invariance for $\Knet$, the invariance of $\Hbar$'s
with respect to the choice of universal Turing machines follows
straightforwardly from the invariance of the prefix-free Kolmogorov
complexity of the
classical string describing a circuit. We can derive the invariance with
respect to the choice of gate set $B$ using similar ideas as those in
\Cref{lem:knet-invarience}, but here we are able to achieve the ideal version
of the invariance.

\begin{lemma}
    For any universal Turing machines $U$ and $U'$, and universal quantum gate sets $B$, and $B'$ with computable amplitudes we get that $$\Hbar^{U,B}(\ket{\psi}) \leq \Hbar^{U',B'}(\ket{\psi}) + O(1).$$
\end{lemma}
\begin{proof}
    We label the set of prefix free programs as $P \subset \{0,1\}^*$, the
    circuit given by $U(p)$ as $C_{U,B,p}$, and the state
    $C_{U',B',p}\ket{0^m}$ as $\ket{\psi_p}$.

    For any $U,U'$, there exists some constant length program $p_{U,U'}$
    such that $U(p_{U,U'},p) = U'(p)$. There also exists a constant length
    program which if given a quantum circuit will output the length $m$ of
    the quantum state it operates on. And by applying the multi-qubit
    Solovay-Kitaev theorem and algorithm to entire circuits we get that for
    any two computable gate sets $B,B'$ there exists a constant length
    program $SK_{B,B'}(2^{-2m}, C_{B'}) \rightarrow C_B$ such that $\forall
    \ket{\psi}: \abs{\bra{\psi}C_{B'}\ket{0} - \bra{\psi}C_{B}\ket{0}} \leq
    2^{-2m}$ and consequently $\abs{\bra{\psi}C_{B'}\ket{0}}^2 -
    \abs{\bra{\psi}C_{B}\ket{0}}^2 \leq 2^{-2m+1}$. 
    
    For each $U',B'$ program $p$ consider the following $U,B$ program $p'$: first run $U(p_{U,U'},p) = U'(p) = C_{U',B',p}$, second extract the length of the resulting state $m$ from  $C_{U',B',p}$, third run $SK_{B,B'}(2^{-2m}, U(p_{U,U'},p))$ and call the resulting $B$ circuit $C'_{U,B,p}$, finally compute and output $C'_{U,B,p}\ket{0}$ which will call $\ket{\psi_{p'}}$. Since all three steps are computable by constant length programs there exists some constant $c$ such that for each $p$, $\abs{p'} \leq \abs{p} + c$.

    By the definition of $\Hbar$ we know that

    \begin{equation*}
        \begin{split}
        \Hbar^{U,B}(\ket{\psi}) & \geq -\log\left(\sum_{p \in P} 2^{-\abs{p'}} \abs{\braket{\psi | \psi_{p'}}}^2 \right)\\
        & \geq -\log\left(\sum_{p \in P} 2^{-\abs{p}-c} \left(\abs{\braket{\psi | \psi_{p}}}^2 - 2^{-2m+1}\right)\right) \\
        & \geq -\log\left(-2^{-2m+1} + \sum_{p \in P} 2^{-\abs{p}-c} \left(\abs{\braket{\psi | \psi_{p}}}^2\right)\right) \\
        & \geq -\log\left(-2^{-2m+1} + 2^{-c}\sum_{p \in P} 2^{-\abs{p}} \left(\abs{\braket{\psi | \psi_{p}}}^2\right)\right) \\
        & \geq -\log\left(2^{-c-1}\sum_{p \in P} 2^{-\abs{p}} \left(\abs{\braket{\psi | \psi_{p}}}^2\right)\right) \\
        & = c+1+\Hbar^{U',B'}(\ket{\psi}).
        \end{split}
    \end{equation*}

    The second to last inequality follows from the fact that $\Hbar(\ket{\psi})$ is at most $m + c'$ for some global constant $c'$, meaning $\sum_{p \in P} 2^{-\abs{p}} \left(\abs{\braket{\psi | \psi_{p}}}^2\right) \geq 2^{-m-c'}$, and consequently the right term is at least $2^{-m-c'-c}$ which is more than twice $2^{-2m+1}$.
\end{proof}

\subsection{Equivalence of \texorpdfstring{$\Hbar$}{H} notions}

G\'{a}cs' original version of $\Hbar$ is defined by defining the universal semi-mixed state
$$\udm = \sum_{p \in P}2^{-|p|}\ketbra{\phi_p}{\phi_p},$$
where $\ket{\phi_p}$ is the state corresponding to reading the output of $U(p)$ as an amplitude vector with each amplitude being read as an algebraic number. 
The 
difference between our approaches is to interpret $U(p)$ as quantum circuits.

Given our argument for the previous lemma, it suffices to prove the equivalence between our notions to show that for any universal gate set $B$ there exists a constant length classical Turing machine which maps from algebraic amplitude vectors representing $\ket{\phi}$ to quantum circuits $C_B$ such that 

$$\abs{\bra{\phi}C_B\ket{0^m}} \geq 1 - 2^{-2m}.$$

There exists a constant length program $M$ which maps algebraic numbers in
$[0,1]$ and error parameters $(a,\eps)$ to standard binary rationals $x$
where $a-x \leq \eps$
The rational amplitude vector
$\ket{r}$ resulting from running $M(a,2^{-3m})$ on each of the $2^m$
algebraic amplitudes and normalizing the vector will satisfy $\braket{r |
\phi} \geq 1-2^{-3m+1}$. For each $B$, there exists a constant length
program which given a rational amplitude vector for $\ket{\phi}$ and error
term $\eps$ can iterate through all possible $B$-circuits until it finds
one such that $\bra{\phi}C_B\ket{0} \geq 1-\eps$, such a circuit is
guaranteed to exist by the Solovay-Kitaev theorem. Running this program on
$(\ket{r}, 2^{-3m})$ will give us the desired circuit.

\end{document}